\documentclass[prx,superscriptaddress,twocolumn,longbibliography]{revtex4-2}

\usepackage[utf8]{inputenc}

\usepackage[T1]{fontenc}
\usepackage{lmodern}

\usepackage{qcircuit}
\usepackage{braket}
\usepackage{amsfonts,amsmath,amssymb,amsthm}
\usepackage{mathtools}

\usepackage{bm}
\usepackage{graphicx}
\usepackage{ascmac}
\usepackage{mathrsfs}
\usepackage{float}
\usepackage{url}
\usepackage{natbib}

\usepackage{algorithmic}
\usepackage{algorithm}

\theoremstyle{plain}
\newtheorem{thm}{Theorem}
\newtheorem{proposition}[thm]{Proposition}
\newtheorem{lem}[thm]{Lemma}

\newtheorem*{thm*}{Theorem}
\newtheorem*{lem*}{Lemma}
\newtheorem*{cor*}{Corollary}

\theoremstyle{definition}
\newtheorem{dfn}{Definition}

\theoremstyle{remark}
\newtheorem{rem}[thm]{Remark}
\newtheorem*{rem*}{Remark}

\usepackage[colorlinks=true, allcolors=blue]{hyperref}

\hypersetup{breaklinks=true}

\begin{document}
\title{Heisenberg-limited adaptive gradient estimation for multiple observables}

\author{Kaito Wada}
\email{wkai1013keio840@keio.jp}
\affiliation{Graduate School of Science and Technology, Keio University, 3-14-1 Hiyoshi, Kohoku, Yokohama, Kanagawa, 223-8522, Japan}

\author{Naoki Yamamoto}
\email{yamamoto@appi.keio.ac.jp}
\affiliation{Department of Applied Physics and Physico-Informatics, Keio University, 3-14-1 Hiyoshi, Kohoku-ku, Yokohama, Kanagawa, 223-8522, Japan}
\affiliation{Quantum Computing Center, Keio University, Hiyoshi 3-14-1, Kohoku, Yokohama 223-8522, Japan}

\author{Nobuyuki Yoshioka}
\email{nyoshioka@ap.t.u-tokyo.ac.jp}
\affiliation{Department of Applied Physics, University of Tokyo, 7-3-1 Hongo, Bunkyo-ku, Tokyo 113-8656, Japan}
\affiliation{Theoretical Quantum Physics Laboratory, RIKEN Cluster for Pioneering Research (CPR), Wako-shi, Saitama 351-0198, Japan}
\affiliation{JST, PRESTO, 4-1-8 Honcho, Kawaguchi, Saitama, 332-0012, Japan}
\affiliation{\mbox{International Center for Elementary Particle Physics, The University of Tokyo, 7-3-1 Hongo, Bunkyo-ku, Tokyo 113-0033, Japan}}

\begin{abstract}
    In quantum mechanics, measuring the expectation value of a general observable has an inherent statistical uncertainty that is quantified by variance or mean squared error of measurement outcome. While the uncertainty can be reduced by averaging several samples, the number of samples should be minimized when each sample is very costly. This is especially the case for fault-tolerant quantum computing that involves measurement of multiple observables of non-trivial states in large quantum systems that exceed the capabilities of classical computers. In this work, we provide an adaptive quantum algorithm for estimating the expectation values of $M$ general observables within root mean squared error $\varepsilon$ simultaneously, using $\mathcal{O}(\varepsilon^{-1}\sqrt{M}\log M)$ queries to a state preparation oracle of a target state. This remarkably achieves the scaling of Heisenberg limit $1/\varepsilon$, a fundamental bound on the estimation precision in terms of mean squared error, together with the sublinear scaling of the number of observables $M$. The proposed method is an adaptive version of the quantum gradient estimation algorithm and has a resource-efficient implementation due to its adaptiveness. Specifically, the space overhead in the proposed method is $\mathcal{O}(M)$ which is independent from the estimation precision $\varepsilon$ unlike non-iterative algorithms. In addition, our method can avoid the numerical instability problem for constructing quantum circuits in a large-scale task (e.g., $\varepsilon\ll 1$ in our case), which appears in the actual implementation of many algorithms relying on quantum signal processing techniques. Our method paves a new way to precisely understand and predict various physical properties in complicated quantum systems using quantum computers.
\end{abstract}
\maketitle

\begin{figure*}[tb]
    \centering
    \includegraphics[width=0.9\textwidth]{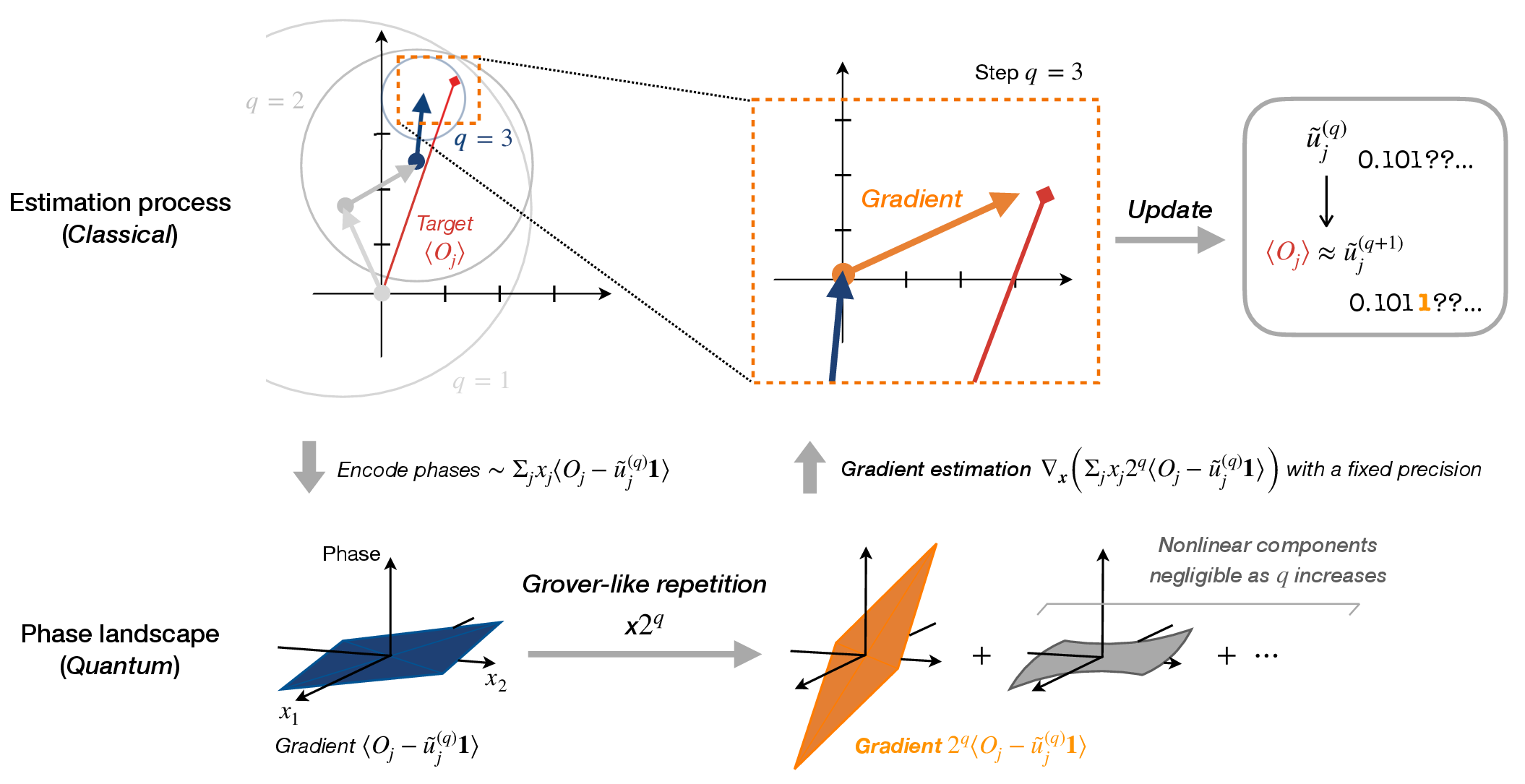}
    \caption{Graphical summary of the proposed method.
    Our algorithm estimates the expectation values of general observables $\{\braket{O_j}\}_{j=1}^M$ regarding a quantum state $\ket{\psi}$ prepared by a unitary $U_{\psi}$, with an adaptive procedure.
    In the $q$th iteration step, 
    we first encode the values of estimation errors $\{\braket{O_j-\tilde{u}^{(q)}_j\bm{1}}\}$ for a unitary gate, where $\tilde{u}^{(q)}_j\in [-1,1]$ denotes a temporal estimated value from the previous step, for the target value $\braket{O_j}$.
    In the upper left illustration, each circle denotes a closed ball with $1/2^q$ radius centered at the $q$th estimate $\{\tilde{u}^{(q)}_j\}$.
    Then, we coherently amplify the estimation errors by the factor of $2^q$ via $\mathcal{O}(2^q\sqrt{M})$ repetitions of Grover-like operation
    (or Hamiltonian simulation) on the encoding unitary and prepare the (approximate) probing state $\ket{\Upsilon(q)}$ (defined as in Eq.~\eqref{eq:targetstate}) whose phase at each point $\ket{\bm{x}}=\ket{x_1}\cdots\ket{x_M}$ is depicted in the lower right illustration. 
    The nonlinear components come from the Grover-like repetition, as detailed in Sec.~\ref{subsec:sp4itergradest_rev}.
    From the probing state $\ket{\Upsilon(q)}$, we read out \textit{zoomed-in} values $\{2^q\braket{O_j-\tilde{u}^{(q)}_j\bm{1}}\}$ simultaneously by quantum gradient estimation protocol with a fixed measurement precision.
    Using the measurement results, we update the temporal estimates $\tilde{u}_j^{(q)}$ to $\tilde{u}_j^{(q+1)}$.
    By repeating this procedure for $q=0,1,...,\lceil \log_2(1/\varepsilon)\rceil$, this algorithm estimates the $M$ observables simultaneously within at most root MSE $\varepsilon$.
    The algorithm uses totally $\mathcal{O}(\varepsilon^{-1}\sqrt{M}\log M)$ queries to the state preparation unitary $U_{\psi}$ that indicates the scaling of Heisenberg limit, that is, the inverse of root MSE $\varepsilon$, together with the nearly squared root dependence on the number of observables $M$, which comes from the uniform singular value amplification process in the encoding of $\{\braket{O_j-\tilde{u}^{(q)}_j\bm{1}}\}$.
    }
    \label{fig:main}
\end{figure*}

\section{Introduction}
\subsection{Background}
Attaining quantum enhancement in unknown parameter estimation lies as one of the most fundamental tasks in quantum technology.
It has been noticed in the quantum metrological community from the late 20th century that there exists a gap between the standard quantum limit (SQL), the statistical scaling of the measurement count $\mathcal{O}(1/\varepsilon^2)$ for target precision of $\varepsilon$ that is based on the central limit theorem, and the Heisenberg limit (HL), the scaling $\mathcal{O}(1/\varepsilon)$ due to quantum uncertainty relations that are quantified by a natural metric called the mean squared error (MSE).
Theoretically, the HL can be achieved 
using entanglement or coherence (i.e., sequential applications of a sensing channel) in quantum probes under the absence of
noise~\cite{giovannetti2004quantum, giovannetti2011advances}, while its experimental verification has not been realized for a long time until the phase estimation procedure by Higgins {\it et al.} that uses a novel adaptive measurement~\cite{higgins2007entanglement}.
Aside from Ref.~\cite{higgins2007entanglement},
intensive quest for experimental realization of the scaling beyond the SQL has provoked numerous interesting ideas such as the use of quantum error correction~\cite{kessler2014quantum, arrad2014increasing, dur2014improved, zhou2018achieving}, exploiting the non-Markovianity of the environment~\cite{matsuzaki2011magnetic, chin2012quantum}.

The preceding pursuit for the fundamental limitation by quantum mechanics naturally pertains to the quantum-enhanced measurement of multiple observables.
This long-standing question has been posed mainly in the community of quantum computing, since an overwhelming number of quantum algorithms estimate the expectation values of local and/or global observables; this task is typically summarized as $M=\mathcal{O}(N^k)$ (possibly non-commuting) observables estimation on an $N$-qubit target system, for some $k$.
Not to mention the scientific application of quantum computers for quantum simulation of many-body systems in natural science~\cite{lloyd1996universal, aspuru2005simulated, childs2018toward, gilyen2019quantum, beverland2022assessing, yoshioka2024hunting},
such a task is also ubiquitous in industrial uses such as quantum machine learning~\cite{biamonte2017quantum, 
liu2021rigorous, huang2022provably} and quantum finance~\cite{Stamatopoulos2020optionpricingusing, chakrabarti2021threshold}.
Note that, while one may employ the modified amplitude estimation algorithm~\cite{suzuki2020amplitude,ZapataQAE2021,wada2022quantum} individually for $M$ observables to simply obtain the HL scaling of $\mathcal{O}(M/\varepsilon)$ in queries to a target state preparation with respect to the root MSE $\varepsilon$, here we seek for simultaneous improvement on the target precision $\varepsilon$ and observable count $M$.
Simply put, our goal is to construct a protocol that achieves both the HL scaling $\mathcal{O}(1/\varepsilon)$ and the sublinear scaling with respect to $M$ in estimating all of $M$ observables.

Indeed, some previous works based on the gradient estimation algorithm~\cite{jordangradest,gilyen2019optimizing} have reported nearly HL scaling, although with multiplicative logarithmic correction $\log (1/\varepsilon)$~\cite{van2023quantum,PhysRevLett.129.240501}; 
there is no work that achieves the ultimate precision scaling following the HL in light of its definition.
Furthermore, existing works argue the estimation precision in terms of confidence intervals, which crucially fails to bound the worst-case behavior of a single run, unlike the MSE.
As is well known in the field of quantum metrology~\cite{berry2001optimal,higgins2009demonstrating,giovannetti2011advances,PhysRevA.102.042613}, the MSE successfully bounds other measures of uncertainty, while the converse does not hold in general.
Considering that many quantum algorithms  that estimate observables make queries to the state preparation oracle of a target quantum state whose complexity usually scales with the system size, it is crucial to design an estimation algorithm that achieves the HL scaling in terms of the MSE.

A key insight for enhancement can be borrowed from the history of phase estimation algorithm~\cite{kitaev1995quantum}.
It has been well-recognized that the textbook style of the phase estimation algorithm~\cite{nielsen2010quantum} encounters two major bottlenecks that prevents the algorithm from practical benefits; the large ancilla consumption of $\mathcal{O}(\log(1/\varepsilon))$ and the poor query complexity scaling with the target root MSE $\varepsilon$, i.e., only obeying the SQL $\mathcal{O}(1/\varepsilon^2)$~\cite{higgins2007entanglement,PhysRevLett.98.090501,PhysRevA.90.062313}.
While the former can be addressed by the iterative phase estimation~\cite{kitaev1995quantum}, 
it was not until the work by Higgins {\it et al.}~\cite{higgins2007entanglement} that 
both issues were overcome by utilizing an adaptive measurement scheme to achieve the HL with a constant number of ancilla.
Although a follow-up work has shown that HL scaling can be achieved even without relying on the adaptive scheme (but using an iterative procedure), there is in practice an increase of the estimation variance by a constant factor~\cite{higgins2009demonstrating}.
On the other hand, if the large ancilla consumption is acceptable, it is known that the quantum phase estimation algorithm with an entangled ancillary state, instead of a uniform superposition state, can achieve the HL~\cite{PhysRevA.54.4564,PhysRevLett.98.090501,PhysRevA.90.062313,PhysRevX.8.041015}.
From the above observations, we argue that adaptive strategy is crucial for resource-efficient implementation of quantum estimation that saturates the ultimate scaling limited by purely fundamental principles of quantum physics.

\subsection{Summary of results}\label{sec:intro_mainres}
Driven by such a situation, we make a significant contribution to multiple observables estimation of number $M$ with root MSE of $\varepsilon$. 
Concretely, we explicitly construct an adaptive estimation scheme as summarized as Algorithm~\ref{alg:main}, or graphically as Fig.~\ref{fig:main} which satisfies the following features (each correspond to Theorem~\ref{thm:main_query_complexity}, Theorem~\ref{thm:main_thm_improved}, Lemmas~\ref{thm:sp_HS}, and \ref{thm:sp_Grover}, respectively):
\begin{itemize}
    \item {\it Heisenberg-limited scaling $\mathcal{O}(1/\varepsilon)$ and sublinear scaling regarding $M$.} The proposed observable estimation achieves the pure HL scaling with root MSE $\varepsilon$ as $$\mathcal{O}(\varepsilon^{-1}\sqrt{M}\log M)$$ 
    in query to the state preparation.
    Also, this scaling indicates a nearly quadratic improvement regarding $M$, compared to using the (modified) quantum amplitude estimation~\cite{brassard2002quantum,knill2007optimal,rall2020quantum,suzuki2020amplitude,ZapataQAE2021,wada2022quantum} $M$ times.
    For a particular set of observables e.g., projectors onto each basis state, we can also achieve the query complexity
    \begin{equation*}
        \mathcal{O}(\varepsilon^{-1}\log M)
    \end{equation*}
    i.e., the logarithmic dependence of $M$ combined with the HL scaling.
    \item {\it Constant space overhead.}
    The quantum circuits in the proposed scheme require at most additional $\mathcal{O}(M)$ qubits, which is independent from the target root MSE $\varepsilon$ unlike non-iterative methods~\cite{PhysRevLett.129.240501,van2023quantum} that use additional $\mathcal{O}(M\log(1/\varepsilon))$ qubits.
    \item {\it Robustness in high-precision regime.}
    When highly precise estimates (i.e., $\varepsilon \ll 1$) are required, the quantum circuits in our method have at most $\mathcal{O}(\log(1/\varepsilon))$ parameterized gates for QSVT~\cite{gilyen2019quantum}, which can be tuned by $\mathcal{O}({\rm polylog}(1/\varepsilon))$ classical computation, while the previous methods~\cite{PhysRevLett.129.240501,van2023quantum} require to tune $\mathcal{O}(1/\varepsilon)$ gates with $\mathcal{O}({\rm poly}(1/\varepsilon))$ classical computation.
    This significant reduction in classical computation allows us to avoid the numerical instability problem~\cite{Haah2019product,chao2020finding,dong2021efficient,Ying2022stablefactorization,PhysRevResearch.6.L012007, yamamoto2024robust} that spoils the quantum enhancement.
\end{itemize}
The first feature is a mathematical guarantee of estimation performance, and the other features highlight ease of practical implementation of our method.
We also mention as another important feature that  the proposed algorithm is the first adaptive extension of quantum gradient estimation~\cite{jordangradest, gilyen2019optimizing}, to the best of our knowledge.
Therefore, our method would be applicable to various tasks in which the gradient estimation algorithm offers quantum speedup, such as 
quantum machine learning~\cite{gilyen2019optimizing,cornelissen2021quantum,Cornelissen2022}, financial risk calculation~\cite{Stamatopoulos2022towardsquantum}, and molecular force calculation~\cite{PhysRevResearch.4.043210}.

Let us briefly describe the abstract structure of the algorithm and where the features stem from.
In our algorithm, we sample from an $\mathcal{O}(M+\log_2 d)$-qubit circuit 
with an alternating sequence of a global interaction to encode the expectation values of target observables and a controlled rotation over the target $\log_2 d$-qubit system (and some ancillary system). 
The global interaction for observables and the total circuit length are adjusted in an adaptive manner, so that the expectation values can be read out in high resolution with the space overhead kept to be independent from the target precision $\varepsilon$.
Then, we classically process the samples and use the processing results to construct the quantum circuit at the next step, in a similar way as the Heisenberg-limited phase estimation algorithms~\cite{higgins2009demonstrating,kimmel2015robust,PhysRevA.102.042613,dutkiewicz2022heisenberg}.
The sublinear scaling in the number $M$ of observables comes from the uniform singular value amplification process~\cite{low2017hamiltonian,gilyen2019quantum} embedded in the global interaction.
This reason for the speedup on $M$ is the same as the previous method~\cite{van2023quantum}.

We remark that the classical computation mentioned in the third point is required for the circuit construction, especially for quantum signal processing (QSP)~\cite{PhysRevX.6.041067,Low2019hamiltonian}.
QSP provides a systematic way to operate a 1-qubit system under a wide range of polynomial functions of degree $n$, using $\mathcal{O}(n)$ parameterized quantum gates, and it is also a key component of a more general technique for quantum matrix polynomials, called quantum singular value transformation (QSVT)~\cite{gilyen2019quantum}.
In the framework of QSP (also QSVT), we need to tune the parameterized gates classically for a desired polynomial.
Although finding this parameter for a degree-$n$ polynomial can be achieved in $\mathcal{O}({\rm poly}(n))$ classical computation time, it exhibits numerical instability for large $n$; this instability leads to undesired algorithmic errors in the resulting quantum circuit, posing a central challenge in the practical application of QSVT~\cite{Haah2019product}.
To resolve this, various optimization techniques have been investigated~\cite{Haah2019product,chao2020finding,dong2021efficient,Ying2022stablefactorization,PhysRevResearch.6.L012007, yamamoto2024robust}, and currently, they require $10^2$--$10^4$ seconds to tune $n\sim 10^4$ parameters.
In the previous method for multiple observables estimation~\cite{van2023quantum}, the number of circuit parameters is given by $n=\tilde{\mathcal{O}}(\sqrt{M}/\varepsilon)$, leading to a large runtime in classical computation.
For instance, we can naively extrapolate that classical runtime for $M=10^4$ and $\varepsilon=10^{-4}$ already requires prohibitively long $10^4$--$10^6$ seconds.
In contrast, our method requires to tune only $\mathcal{O}(\sqrt{M}\log(M/\varepsilon))$ parameters under a certain condition, by partially using a special polynomial whose parameters can be analytically determined.
This exponential improvement in classical computation time, regarding estimation precision $1/\varepsilon$, significantly lowers the barrier of the practical implementation of our method.

Finally, we provide numerical simulations to compare the number of total queries and ancilla qubits in our adaptive method and the previous non-iterative method~\cite{PhysRevLett.129.240501}.
The results clearly show the quantitative advantages of our method in terms of query and space complexity.
In particular, the previous method indicates the Heisenberg-limited scaling $\mathcal{O}(1/\varepsilon)$ with logarithmic factors on $1/\varepsilon$ regarding root MSE $\varepsilon$, while our method has no $\varepsilon$ dependence except for $\mathcal{O}(1/\varepsilon)$ in the total queries.

The remainder of this paper is organized as follows.
The problem discussed in this paper is described in Sec.~\ref{sec:problem_setup}.
In Sec.~\ref{sec:main_III}, we focus on the estimation performance of the proposed method, setting aside its circuit implementation.
Then, we describe the implementation of our method in Sec.~\ref{subsec:sp4itergradest_rev}.
Sec.~\ref{sec:num_simulation} provides the numerical simulations.
Finally, we conclude the paper in Sec.~\ref{sec:conclusion}.
The full proof of our propositions is given in Appendix.

\section{Problem setup}\label{sec:problem_setup}

We consider estimating quantum expectation values of given $d$-dimensional $M$ Hermitian operators $\{O_j\}_{j=1}^M$ with the spectral norm $\|O_j\|\leq 1$, regarding a quantum state $\ket{\psi}$. Here, $d$ is a power of 2 and $N$ denotes the number of target system qubits i.e., $N\equiv\log_2 d$.
The target state $\ket{\psi}$ is assumed to be prepared by a state preparation oracle $U_{\psi}:\ket{\bm{0}}\mapsto \ket{\psi}$ and an initial state $\ket{\bm{0}}:=\ket{0}^{\otimes \log_2 d}$, and we assume oracular access to $U_{\psi}$ and $U_{\psi}^\dagger$.
The observables are assumed to be accessed by some block-encoded unitaries over the $d$-dimensional system and some ancilla system; that is, for each observable $O_j$, we assume access to a unitary gate $B_j$ whose top-left block matrix is $O_j$.
Various descriptions of observables e.g., sparse matrices and linear combinations of Pauli strings can be easily converted to such an encoding~\cite{gilyen2019quantum}.
(Precise definition of this encoding is provided in Sec.~\ref{subsec:sp4itergradest_rev}.)
In this setup, our goal is to \textit{efficiently} obtain samples from estimators 
for the target values $\braket{O_j}:=\bra{\psi}O_j\ket{\psi}$ within the root mean squared error (MSE) $\varepsilon$. 
In particular, we aim to simultaneously achieve the HL scaling regarding $\varepsilon$ and the sublinear scaling regarding $M$, for estimating all $\braket{O_j}$. 

The performance of quantum algorithms for this task is usually quantified by the total number of queries to the state preparation $U_{\psi}$ and $U_{\psi}^\dagger$.
This is because $U_{\psi}$ has complexity that usually scales with the system size, and as a result, the state preparation $U_{\psi}$ is the most dominant factor in the total execution time for various settings.
Under the natural assumption that $U_{\psi}^\dagger$ has the same cost as $U_{\psi}$,
it is crucial to design an estimation algorithm that minimizes the statistical uncertainty using 
a limited number of queries to $U_{\psi}$ and $U_{\psi}^\dagger$.

\section{Adaptive gradient estimation for multiple quantum observables}\label{sec:main_III}

\renewcommand{\baselinestretch}{1.2}
\begin{figure}[tb]
\begin{algorithm}[H]
    \caption{Adaptive gradient estimation\\for multiple observables}\label{alg:main}
    \begin{algorithmic}[1]
    \smallskip
    \REQUIRE $\log_2{d}$-qubit state preparation unitary $U_{\psi}$ and $U_{\psi}^\dagger$; observables $\{O_j\}_{j=1}^M$ with the spectral norm $\|O_j\|\leq 1$ such that $M>\mathcal{O}(\log d)$ holds;
    confidence parameter $c\in (0,3/8(1+\pi)^2]$; target root mean squared error (MSE) $\varepsilon\in (0,1)$.

    \smallskip
    \ENSURE A sample $(\tilde{u}_1,...,\tilde{u}_M)$ from an estimator $\hat{\boldsymbol{u}}=(\hat{u}_1,...,\hat{u}_M)$ whose $j$-th element estimates $\langle\psi|O_j|\psi\rangle$ within the MSE $\varepsilon^2$ as
    $$
    \max_{j=1,2,...,M}~\mathbb{E}\left[\left(\hat{u}_j-\langle\psi|O_j|\psi\rangle\right)^2\right]\leq \varepsilon^2
    $$

    \STATE Set a fixed precision parameter $p:=3$ and temporal estimates $\tilde{u}_j^{(0)}:= 0$ for all $j$.
    
    \FOR{$q=0,1,...,{q_{\rm max}:=}\lceil\log_2(1/\varepsilon)\rceil$}
    
    

    
    \STATE Measure $\mathcal{O}(\log M/\delta^{(q)})$ approximate copies of the probing state
    \begin{equation*}
        \ket{\Upsilon(q)}:=\frac{1}{\sqrt{2^{pM}}}\sum_{\bm{x}\in G_p^M} e^{2\pi i2^p \sum_{j=1}^M x_j {2^{q}\pi^{-1}\braket{{O}_j-\tilde{u}^{(q)}_j\bm{1}}}} \ket{\boldsymbol{x}}
    \end{equation*}
    after $p$-qubit inverse quantum Fourier transformations.
    Here, $\delta^{(q)}:=c/8^{{q_{\rm max}}-q}$ is a failure probability.
    
    \STATE Set the coordinate-wise median of the measurement outputs as ${g}_j^{(q)}$.

    \STATE Update $\tilde{u}_j^{(q+1)}:=\tilde{u}_j^{(q)}+{\pi}{2^{-q}}g_j^{(q)}$
    
    \STATE Truncate $\tilde{u}_j^{(q+1)}$ in $[-1,1]$
    
    \ENDFOR
    
    \RETURN final estimates $\tilde{u}_j:= \tilde{u}_j^{({q_{\rm max}}+1)}$
    \end{algorithmic}
\end{algorithm}
\end{figure}
\renewcommand{\baselinestretch}{1.0}

To address the problem stated in Sec.~\ref{sec:problem_setup}, we propose an adaptive extension of the gradient estimation algorithm.
In the following, we first provide a brief overview and minimal examples of the gradient estimation algorithm in Sec.~\ref{subsec:grad_est}, and then describe the adaptive gradient estimation algorithm for multiple observables estimation in Sec.~\ref{subsec:adaptive_grad_est_performance}. Finally, the performance guarantee of the algorithm is provided in Sec.~\ref{subsec:estimation_performance}.

\subsection{Quantum gradient estimation} \label{subsec:grad_est}
Here, we briefly describe the gradient estimation algorithm~\cite{jordangradest, gilyen2019optimizing}; see Appendix~\ref{supple_sec:qgradest} for more details.
To clarify the idea of this algorithm, we first consider the gradient estimation of a 1-dimensional function $f(x)=gx$ at $x=0$.
In this algorithm, the function $f(x)$ is coherently evaluated over multiple grid points $\{x\}\subset \mathbb{R}^1$ around 0.
Specifically, we initialize $p$-qubit registers as the uniform superposition state $(1/\sqrt{2^p})\sum_{x}\ket{x}$ of $2^p$ grid points $x\in G_p$.
Here, $p$ denotes a precision parameter and $G_p$ is defined as a set of grid points:
\begin{equation}\label{eq:main_Gpdef}
    G_{p}:=\left\{\frac{\mu}{2^p}-\frac12 +\frac{1}{2^{p+1}}:\mu\in\{0,1,\cdots,2^p-1\}\right\}.
\end{equation}
In the following, we label the $p$-qubit computational basis $\ket{\mu}$ $(\mu\in \{0,1,...,2^{p}-1\})$ by each grid point $x$ in $G_p$ via one-to-one correspondence between $G_p$ and the computational basis set.
Then, using $\mathcal{O}(2^p)$ applications of a (phase) oracle that puts a relative phase $e^{if(x)}$ in each point $\ket{x}$, we have the resulting state 
\begin{equation}\label{eq:1d_gradest}
    \frac{1}{\sqrt{2^p}} \sum_{x\in G_p}e^{2\pi i 2^p gx}\ket{x}.
\end{equation}
Therefore, as well as the standard phase estimation algorithm, we can estimate $g$ in precision $\mathcal{O}(1/2^p)$ by $p$-qubit Fourier basis measurement on this state. 
Note that the Fourier transformation is slightly modified as 
\begin{equation}\label{supple_eq:modi_QFT}
    {\rm QFT}_{G_p}:\ket{x}\mapsto \frac{1}{\sqrt{2^p}}\sum_{k\in G_p} e^{2\pi i 2^p xk}\ket{k}.
\end{equation}
This is the same as the usual $p$-qubit QFT up to conjugation with a tensor product of $p$ single-qubit gates~\cite{gilyen2019optimizing}.

The 1-dimensional case can be easily generalized to multidimensional cases.
That is, for $M$-dimensional functions such as $f(\bm{x})\approx \sum_j g_jx_j$, we can perform the above procedure in $M$ parallel, using $\mathcal{O}(2^p)$ queries to an oracle for $\ket{\bm{x}}\mapsto e^{if(\bm{x})}\ket{\bm{x}}$, and approximately obtain the $M$ tensor products of the quantum state Eq.~\eqref{eq:1d_gradest} as
\begin{equation}\label{eq:Md_gradest}
    \bigotimes_{j=1}^M\left[\frac{1}{\sqrt{2^{p}}}\sum_{x_j\in G_p} e^{2\pi i2^p g_j x_j }\ket{x_j}\right].
\end{equation}
Then, applying $M$ tensor products of $p$-qubit Fourier basis measurements yields an $\mathcal{O}(1/2^p)$-precise estimate of the gradient ${g}_j$ for all $j=1,2,...,M$,
with a remarkable property that the implementation is independent 
from $M$ in terms of query to an oracle for $f(\bm{x})\approx \sum_j g_jx_j$. 
However, it must be noted that such an advantage is achieved
under the condition that the approximation $f(\bm{x})\approx \sum_j g_jx_j$ holds very well. 
In most cases, the higher order cannot be negligible and additional oracle queries are required to enhance the linearity.
For certain smooth functions, the smallest overhead is proven to be $\mathcal{O}(\sqrt{M})$, and this can be achieved, up to poly-logarithmic factors, by the method using higher-degree central-difference formula~\cite{gilyen2019optimizing}.
For practical applications of this algorithm, it is crucial in total quantum resources to efficiently prepare an oracle for each task.
Namely, we want to design  
a sufficiently linear function $f(\bm{x})\approx \sum_j g_jx_j$ whose gradient $\{g_j\}$ encodes the target quantities for each task and further prepare the corresponding oracle for $\ket{\bm{x}}\mapsto e^{if(\bm{x})}\ket{\bm{x}}$ or the final state Eq.~\eqref{eq:Md_gradest}.



\subsection{Proposed adaptive algorithm}\label{subsec:adaptive_grad_est_performance}
Now we are ready to introduce our algorithm for estimating the expectation values of multiple observables that uses the gradient estimation as a subroutine.
The outline of the algorithm is provided in Algorithm~\ref{alg:main}, and also its high-level graphical overview is shown in Fig.~\ref{fig:main}.
In the following, we first introduce the foundation of Algorithm~\ref{alg:main}, with a particular focus on the quantum part (i.e., the quantum state $\ket{\Upsilon (q)}$ in Step~3).
Then, we prove our main theorem (Theorem~\ref{thm:main_query_complexity}) to guarantee that Algorithm~\ref{alg:main} achieves both the HL scaling $\mathcal{O}(1/\varepsilon)$ and the sublinear scaling on $M$ for the total query complexity regarding the use of $U_{\psi}$ and $U^\dagger_{\psi}$.
Appendix~\ref{supple_sec:proposed_alg} gives complete proof of theorems together with showing concrete implementation method for the algorithm.

The key idea of the proposed adaptive algorithm is to prepare the following \textit{probing} state $\ket{\Upsilon(q)}$ at each iteration step $q$, from which the approximated expectation values can be extracted via the gradient estimation algorithm described in Sec.~\ref{subsec:grad_est}:
\begin{align}\label{eq:targetstate}
    &\ket{\Upsilon(q)}:=\notag\\
    &~~~\frac{1}{\sqrt{2^{pM}}}\sum_{\bm{x}\in G_p^M} e^{2\pi i2^p \sum_{j=1}^M x_j {2^{q}\pi^{-1}\braket{{O}_j-\tilde{u}^{(q)}_j\bm{1}}}} \ket{\boldsymbol{x}},
\end{align}
where $\bm{1}$ denotes the identity operator.
Also, $p$ denotes a fixed precision parameter and $G_p^M$ is defined as in Eq.~(\ref{eq:main_Gpdef}).
The quantities $\tilde{u}^{(q)}_j\in [-1,1]$ in the probing state Eq.~(\ref{eq:targetstate}) are temporal estimated values for $\braket{O_j}$ at the iteration step $q$.
Later we will show that $p=3$ is sufficient for our algorithm to successfully work.
We remark that the probing state $\ket{\Upsilon(q)}$ is equivalent to the quantum state Eq.~(\ref{eq:Md_gradest}) with gradient $g_j= {2^{q}\pi^{-1}\braket{{O}_j-\tilde{u}^{(q)}_j\bm{1}}}$.

The probing state $\ket{\Upsilon(q)}$ can be prepared by first intializing the $pM$-qubit probe system, the $\log_2 d$-qubit target system, and some ancilla systems to encode the observables.
Then, applying a $q$-dependent interaction unitary 
on the whole system, 
we approximately obtain the probing state. 
The query complexity regarding $U_{\psi}$ and $U^\dagger_{\psi}$ of this state preparation process is given as follows.
\begin{proposition}[Probing-state preparation]\label{thm:probing_sp_informal}
    Suppose that we have access to block-encoded observables $\{O_j\}_{j=1}^M$ in $d$ dimension via some unitaries, a $\log_2 d$-qubit state preparation $U_{\psi}$, and its inverse $U_{\psi}^\dagger$, such that $M>\mathcal{O}(\log d)$ holds.
    Then, we can prepare the probing state $\ket{\Upsilon(q)}$, up to 1/12 Euclidean distance error, for (any) integer $q\geq 0$ and $\tilde{u}^{(q)}_j\in [-1,1]$, using $\mathcal{O}(2^q\sqrt{M\log d})$ queries to $U_{\psi}$ and $U_{\psi}^\dagger$ in total.
\end{proposition}
\noindent
In the next section, we describe this probing-state preparation in detail.
In particular,
this proposition directly follows from Lemmas~\ref{thm:sp_HS} and \ref{thm:sp_Grover} shown in the next section, both of which give 
methods to prepare the probing state $\ket{\Upsilon(q)}$: one is based on the Hamiltonian simulation protocol (Lemma~\ref{thm:sp_HS}), and the other is based on the Grover-like repetition (Lemma~\ref{thm:sp_Grover}).

Suppose we have the probing state $\ket{\Upsilon(q)}$ for Step~3 at the $q$th iteration, with a temporal estimate $\tilde{u}^{(q)}_j$ for $\braket{O_j}$. 
Then, on the phase of $\ket{\bm{x}}$ in $\ket{\Upsilon(q)}$, the estimation error $\braket{O_j-\tilde{u}^{(q)}_j\bm{1}}$ is amplified by the factor of $2^q$; this means that when $\tilde{u}^{(q)}_j$ matches $\braket{O_j}$ in the first $q$ binary (fraction) digits as
\begin{equation}\label{eq:input_condition}
    \left|\braket{O_j}-\tilde{u}^{(q)}_j\right| \leq \frac{1}{2^q},
\end{equation}
the $2^q$-fold amplification \textit{zooms in} on the significant digits of $\braket{O_j}-\tilde{u}^{(q)}_j$ to shift the values toward upper digits. 
As proved in Ref.~\cite{gilyen2019optimizing}, the gradient estimation for $\ket{\Upsilon(q)}$ can simultaneously extract these zoomed-in values $2^q \pi^{-1}\braket{{O}_j-\tilde{u}^{(q)}_j\bm{1}}$ or equivalently the gradient of the linear function on the phase (with $1/\pi$): 
\begin{equation}\label{eq:target_phase_fn}
    {f(\bm{x})}:=\sum_j x_j 2^q\braket{O_j-\tilde{u}^{(q)}_j\bm{1}},
\end{equation}
with an additive error specified by the precision parameter $p$ and a certain success probability (see Lemma~\ref{supple_lem:original_gradest}).
Note that, in our case, the gradient estimation algorithm is simply to perform the computational basis measurement (with rewriting $\ket{\mu}$ as $\ket{x}$) on the probing state after applying a slightly modified version of $p$-qubit inverse quantum Fourier transformation, which is defined in Eq.~(\ref{supple_eq:modi_QFT}).

In Step~4 of Algorithm~\ref{alg:main}, we use the measurement outputs $\bm{k}^{(1)},\bm{k}^{(2)},...,\bm{k}^{(\mathcal{O}(\log M/\delta^{(q)}))}$ obtained in Step~3 to construct the coordinate-wise median $g_j^{(q)}$; that is, $g_j^{(q)}$ is defined as the middle value separating the greater and lesser halves of $\{k_j^{(1)},k_j^{(2)},...\}$. 
Then, under the condition of Eq.~\eqref{eq:input_condition}, we can prove that the median $g_j^{(q)}$ satisfies
\begin{equation}\label{eq:success_output_gradest}
    \left|g_j^{(q)}-\frac{2^q\braket{{O}_j-\tilde{u}^{(q)}_j\bm{1}}}{\pi}\right|\leq \frac{1}{2\pi},
\end{equation}
for all $j=1,2,...,M$ 
with success probability bigger than $1-\delta^{(q)}$; 
later we will carefully choose the probability $\delta^{(q)}$ in order that a final estimator has the minimal statistical uncertainty. 
Using the successfully obtained ${g}_j^{(q)}$ satisfying Eq.~\eqref{eq:success_output_gradest}, we update the temporal estimate in Step 5 as 
\[
    \tilde{u}_{j}^{(q+1)}:=\tilde{u}_{j}^{(q)}+\pi 2^{-q}g_{j}^{(q)}.
\]
Then, it is straightforward to prove 
\begin{equation*}
    \left|\braket{O_j}-\tilde{u}^{(q+1)}_j\right| \leq \frac{1}{2^{q+1}},
\end{equation*}
for all $j$.
Thus, using the outcomes from the gradient estimation algorithm at each step, we can iteratively determine one binary fraction digit of the target value $\braket{O_j}$
by updating $\tilde{u}^{(q)}_j$ to $\tilde{u}^{(q+1)}_j$. 
Repeating this iteration step for $q_{\rm max}:=\lceil \log_2(1/\varepsilon)\rceil$ times, we obtain estimates for $\{\braket{O_j}\}_j$ within the root MSE $\varepsilon$.

\subsection{Estimation performance and query complexity}\label{subsec:estimation_performance}

Now, we provide a performance guarantee, namely the HL scaling of the multiple observables estimation precision in terms of MSE, for the proposed Algorithm~\ref{alg:main}.
While a more detailed proof is given in Appendix~\ref{supple_sec;complexity_of_alg}, here we present the Theorem with a sketch of its proof in the following:
\begin{thm}[{Heisenberg-limited multiple observables estimation}]\label{thm:main_query_complexity}
    Let $\varepsilon\in(0,1)$ be a target precision.
    For given $M$ observables $\{O_j\}_{j=1}^M$ ($\|O_j\|\leq 1~\forall j$) and a state preparation oracle $U_{\psi}$, 
    there exists a quantum algorithm that outputs a sample from estimators $\{\hat{u}_j\}_{j=1}^M$ for $\{\braket{O_j}\}$ satisfying
    \begin{equation}\label{eq:thm2_maxmse}
        \max_{j=1,2,...,M}~{\rm MSE}\left[\hat{u}_j\right]\leq \varepsilon^2,
    \end{equation}
    using $\mathcal{O}(\varepsilon^{-1}\sqrt{M}\log M)$ queries to the state preparation oracles $U_{\psi}$ and $U_{\psi}^\dagger$ in total.
    Here, the mean squared error of an estimator $\hat{u}_j$ is defined as
    $$
        {\rm MSE}\left[\hat{u}_j\right] := \mathbb{E}\left[\left(\hat{u}_j-\braket{O_j}\right)^2\right].
    $$
\end{thm}
\begin{proof}[Sketch of the proof]
    While the outline of the proof is inspired from the Heisenberg-limited phase estimation or its application~\cite{higgins2009demonstrating,kimmel2015robust,dutkiewicz2022heisenberg,dutkiewicz2023advantage}, it is nontrivial for the case of adaptive gradient estimation to show that the condition given in Eq.~\eqref{eq:input_condition} is satisfied for all $j$.
Note that, if this condition holds and the precision parameter $p$ is taken as $p=3$, then a single shot measurement result $\bm{k}:=(k_1,...,k_M)\in G_p^M$ in Step~3 of Algorithm~\ref{alg:main} follows 
\begin{equation}\label{eq:grad_est_output_dist}
    {\rm Pr}\left[\left|k_j-\frac{2^q(\braket{O_j}-\tilde{u}_j^{(q)})}{\pi}\right|>\frac{1}{2\pi}\right]< \frac{1}{3},
\end{equation}
for every $j=1,2,\cdots,M$. 
The derivation of this inequality is based on our numerical finding; see Appendix~\ref{supple_sec;complexity_of_alg}. 
We remark that the failure probability represented by the left hand side of Eq.~\eqref{eq:grad_est_output_dist} can be exponentially suppressed to $\delta^{(q)}/M$ by using the coordinate-wise median ${g}_j^{(q)}$ of the measurement results over $\mathcal{O}(\log M/\delta^{(q)})$ copies of the (approximate) probing state, instead of the single shot result $k_j$.
This repetition is the origin of the $\log M$ term in the total query complexity.

Let us first discuss how to assure that the condition~\eqref{eq:input_condition} is satisfied.
In the first iteration step $q=0$, the condition Eq.~\eqref{eq:input_condition} is trivially satisfied, and the gradient estimation can yield $1/2\pi$-close estimates of the target quantities for all $j$.
This probability is at least $1-\delta^{(0)}$ due to the union bound.
As for the iteration step $q\geq 1$, if all of the previous steps succeed in the gradient estimation, 
we can show that Eq.~\eqref{eq:input_condition} holds for all $j$ at the iteration step $q$, as discussed in the previous subsection.
On the other hand, if the gradient estimation fails at the iteration step $q$ (while all processes have been successfully executed up to the $(q-1)$th step), the condition Eq.~\eqref{eq:input_condition} is not satisfied; as a result, outputs of the gradient estimation may not improve the temporal estimate in the subsequent processes.
However, in this case, it can be shown that the additive error $|\tilde{u}_j-\braket{O_j}|$ of the final estimate $\tilde{u}_j:=\tilde{u}_j^{(q_{\rm max}+1)}$ is at most $(1+\pi)/2^q$ because the outputs of gradient estimation are always in $[-1/2,1/2]$ from the definition of $G_p^M$.
From the above analysis, we can bound the additive error of the final estimate for all branches in Fig.~\ref{fig_main:final_est_and_branches}.

\begin{figure}[tb]
    \centering
    \includegraphics[width=0.4\textwidth]{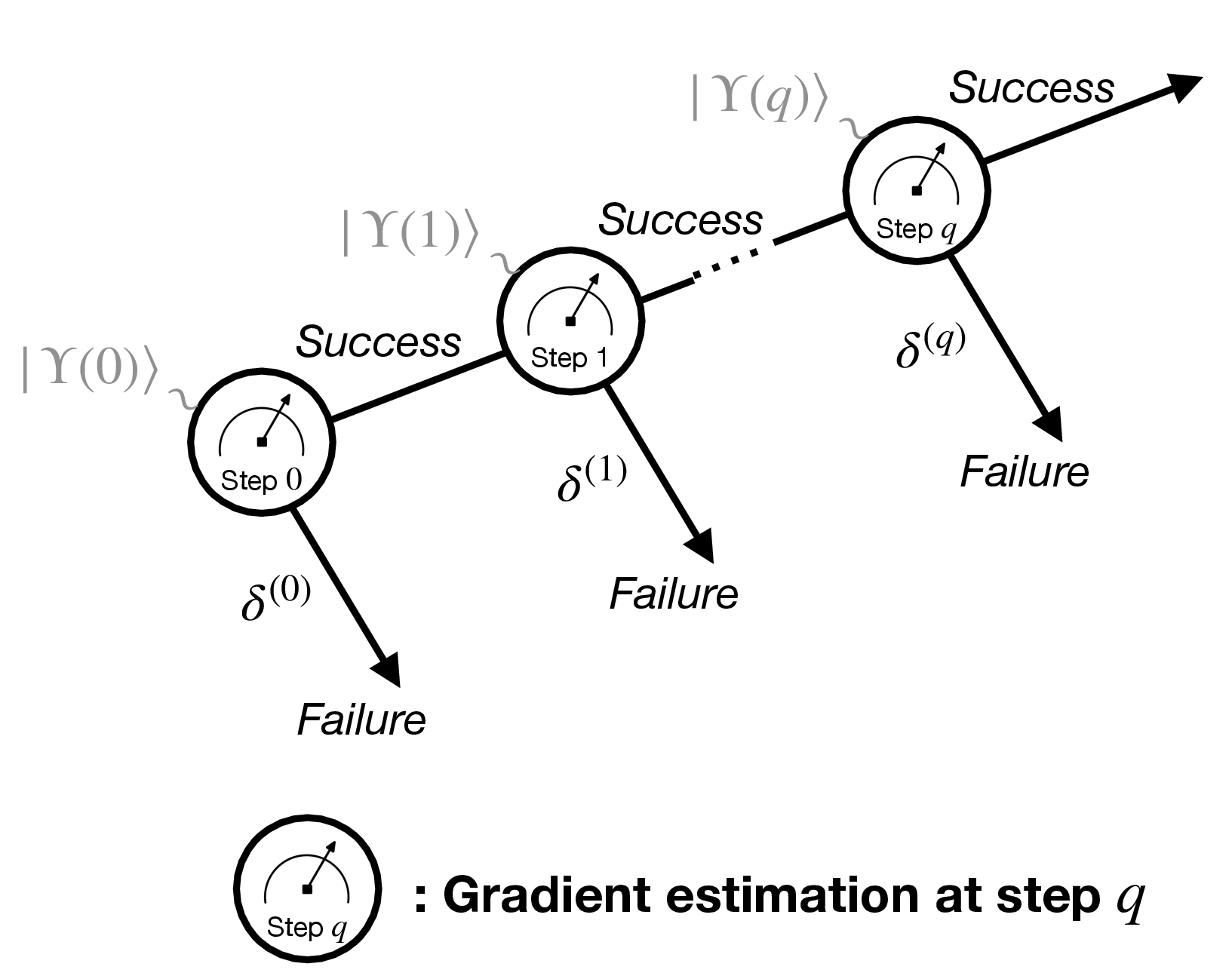}
    \caption{Probability tree diagram in Algorithm~\ref{alg:main}.}
    \label{fig_main:final_est_and_branches}
\end{figure}

The failure probability $\delta^{(q)}$ should be very small at the beginning of iteration steps $q$ because the error of top fraction bits have a significant impact on the MSE 
of the final estimator.
From the analysis in Ref.~\cite{dutkiewicz2023advantage}, the choice of $\delta^{(q)}:=c/8^{q_{\rm max}-q}$ ($c\in (0,3/8(1+\pi)^2)$) is sufficient, and then the MSE 
of the final estimator ${\hat{u}_j}$ (its realization is $\tilde{u}_j$) is calculated as at most $\varepsilon^2$ for all $j=1,2,...,M$, from the tree diagram Fig.~\ref{fig_main:final_est_and_branches} and the above error evaluation in each branch.

To derive the query complexity for the state preparation $U_{\psi}$, we use the result in Proposition~\ref{thm:probing_sp_informal}.
In the iteration step $q$, preparing $\mathcal{O}(\log M/\delta^{(q)})$ copies of the probing state Eq.~\eqref{eq:targetstate} uses $\mathcal{O}(2^q \sqrt{M\log d}\log (M/\delta^{(q)}))$ queries to the state preparation unitary $U_{\psi}$ and its inverse.
Thus, the summation of the queries over $q=0,1,...,q_{\rm max}~(q_{\rm max}:=\lceil \log_2(1/\varepsilon)\rceil$) is directly calculated as $\mathcal{O}(\varepsilon^{-1}\sqrt{M\log d}\log M)$, which completes the proof of Theorem~\ref{thm:main_query_complexity}. 
\end{proof}

In Theorem~\ref{thm:main_query_complexity}, the scaling $1/\varepsilon$ of queries to $U_{\psi}$ with respect to the root MSE $\varepsilon$ achieves the same scaling to the Heisenberg limit (HL) in the gradient estimation, which is derived in Appendix~\ref{supple_sec:qgradest}.
We remark that, to the best of our knowledge, this is the first derivation of the HL in gradient estimation in terms of MSE, while the previous work~\cite{gilyen2019optimizing} shows a similar lower bound (additionally, including $\sqrt{M}$ dependency) on queries to an oracle for a target function $f(\bm{x})$ on $\bm{x}\in G_p^M$, in another metric of statistical uncertainty (i.e., a confidence interval).
Also, in the stringent measure of statistical uncertainty, the root MSE $\varepsilon$~\cite{berry2001optimal}, all the existing works~\cite{PhysRevLett.129.240501,van2023quantum} in multiple observables estimation only proves the nearly HL scaling that includes logarithmic terms on $\varepsilon$ such as $\log (M/\varepsilon)$.

In addition to the Heisenberg-limited scaling regarding $\varepsilon$, our protocol has the nearly squared root dependence regarding the total number of observables $M$.
This is a nearly quadratic improvement regarding $M$ compared to the standard method (i.e., quantum amplitude estimation) of estimating the expectation values of observables.

Furthermore, when we focus on another metric of statistical uncertainty --- a confidence interval comprised of estimators with an additive error $\pm \epsilon_{\rm add}$ for some confidence level, which is considered in Refs.~\cite{PhysRevLett.129.240501,van2023quantum}, the query complexity of our method is essentially optimal in terms of $M$ and $\epsilon_{\rm add}$ as well as these previous methods.
Here, the query lower bound in multiple observables estimation in terms of a confidence interval is proved in Ref.~\cite{PhysRevLett.129.240501} as follows.
\begin{lem}
    [Corollary~3 in Ref.~\cite{PhysRevLett.129.240501}]\label{lem:google_querylb}
    Let $M$ be a positive integer power of 2 and let $\epsilon_{\rm add}\in (0,1/3\sqrt{M})$. Let $\mathcal{A}$ be any algorithm that takes as an input an arbitrary set of $M$ observables $\{O_j\}_{j=1}^M$. Suppose that, for every quantum state $\ket{\psi}$, accessed via a state preparation oracle $U_{\psi}$, $\mathcal{A}$ outputs estimates of each $\bra{\psi}O_j\ket{\psi}$ to within additive error $\epsilon_{\rm add}$ with high probability at least $2/3$. Then, there exists a set of observables $\{O_j\}_{j=1}^M$ such that $\mathcal{A}$ applied to $\{O_j\}_{j=1}^M$ must use $\Omega(\epsilon_{\rm add}^{-1}\sqrt{M})$ queries to $U_{\psi}$.
\end{lem}
\noindent
For a given $\varepsilon$, Algorithm~\ref{alg:main} also satisfies
\begin{equation}
    {\rm Pr}\left[{\rm max}_j\left|\hat{u}_j-\braket{O_j}\right|\leq\varepsilon/2\right]\geq 1-8c/7,
\end{equation}
by using $\mathcal{O}(\varepsilon^{-1}\sqrt{M}\log (M/c))$ queries to $U_{\psi}$ and $U_{\psi}^\dagger$.
This can be confirmed by calculating the probability of the all-success branch in Fig.~\ref{fig_main:final_est_and_branches}.
Consequently, we establish from the Lemma~\ref{lem:google_querylb} that our protocol achieves the worst-case query optimality (up to the $\log M$ correction) in the high-precision regime $\epsilon_{\rm add}\in (0,1/3\sqrt{M})$.

For a particular set of observables, Ref.~\cite{van2023quantum} shows an improvement on the dependence of $M$ in the query complexity.
A similar improvement can be attained while retaining the HL scaling $\mathcal{O}(1/\varepsilon)$ in MSE.
\begin{thm}[Improved $M$-dependence]\label{thm:main_thm_improved}
    In Theorem~\ref{thm:main_query_complexity}, we can improve the total number of queries to $U_{\psi}$ and $U^\dagger_{\psi}$ in terms of $M$, while keeping the HL scaling $\mathcal{O}(\varepsilon^{-1})$ for root MSE $\varepsilon$, to
    \begin{equation}
        \mathcal{O}\left({\varepsilon}^{-1}\sqrt{\mathcal{B}_M}\log M\right)
    \end{equation}
    for any \textit{known} $\mathcal{B}_{M}$ satisfying 
    \begin{equation}
        \left\|\sum_{j=1}^M O_j^2\right\|\leq \mathcal{B}_{M}\leq M.
    \end{equation}
\end{thm}
\noindent
The proof of this theorem is essentially the same as that of Theorem~\ref{thm:main_query_complexity} except for the state preparation (more precisely, we use Algorithm~\ref{alg:main} with a modified Step~3); see Appendix~\ref{sec:proof_improved_thm}.
This theorem recovers the query complexity in Theorem~\ref{thm:main_query_complexity} when we only know $\|O_j\|\leq 1$ on observables, leading to $\mathcal{B}_M =M$.
For a particular set of observables, $\mathcal{B}_M$ can be taken as e.g., $\mathcal{B}_M=\mathcal{O}(1)$ for 
$O_j=\ket{j}\bra{j}$ in multidimensional amplitude estimation~\cite{van2021quantum,van2023quantum}, so the logarithmic dependence of $M$ together with the HL scaling $\mathcal{O}(1/\varepsilon)$ is attained in such cases.

Finally, we remark that the preparation and measurement for multiple copies of the probing state in Step 3 of Algorithm~\ref{alg:main} can be performed in parallel if we are allowed to use multiple quantum computers, which may be practically important to reduce the total execution time.

\section{State preparation for adaptive gradient estimation}\label{subsec:sp4itergradest_rev}

In the performance guarantee (Theorem~\ref{thm:main_query_complexity}) provided in Sec.~\ref{subsec:adaptive_grad_est_performance}, we have exploited Proposition~\ref{thm:probing_sp_informal} that states the query complexity to prepare the probing state defined in Eq.~\eqref{eq:targetstate}.
In this section, we aim to show Proposition~\ref{thm:probing_sp_informal} by providing two quantum algorithms to prepare the probing state Eq.~\eqref{eq:targetstate}, using the state preparation $U_{\psi}$, $U_{\psi}^\dagger$, and the unitary gates $\{B_j\}_{j=1}^M$ that encode the target observables $\{O_j\}_{j=1}^M$ with the help of some ancilla system.
The two methods are summarized in Lemmas~\ref{thm:sp_HS} and \ref{thm:sp_Grover}; Lemma~\ref{thm:sp_HS} provides a method based on Hamiltonian simulation protocol, and Lemma~\ref{thm:sp_Grover} provides a method based on Grover-like repetition.
Both Lemmas~\ref{thm:sp_HS} and \ref{thm:sp_Grover} prove that the total space complexity to prepare the probing state $\ket{\Upsilon(q)}$ is $\mathcal{O}(M+\log_2 d)$.
In addition, Lemma~\ref{thm:sp_Grover} shows an exponential improvement in the classical computation time for constructing explicit quantum circuits under a certain condition, compared to Lemma~\ref{thm:sp_HS}.

In the following, we first review an efficient quantum computation method of block-encoded matrices. Then, we describe the two methods for the probing-state preparation.


\begin{figure*}[bt]
\centering
\begin{tabular}{c}
\\
\Qcircuit @C=1.2em @R=1.2em {
  \lstick{\ket{0}^{\otimes \lceil\log_2 m\rceil}}& {/} \qw & \gate{V} & \multigate{1}{W} & \gate{V^\dagger} & \qw \\
  & {/} \qw & \qw & \ghost{W} & \qw & \qw \\
  }
\\
\\
\end{tabular}
\caption{Linear combination of unitaries (LCU) method. 
The circuit is a $\lceil\log_2 m\rceil$-block-encoding of $\|c\|_1^{-1}\sum_{i=1}^m c_i U_i$ ($c_i>0$).
PREPARE $V:\ket{0}^{\otimes \lceil\log_2 m\rceil} \mapsto \sum_{i=1}^{m} \sqrt{\frac{c_i}{\|c\|_1}}|i\rangle$ and SELECT $W=\sum_{i=1}^m|i\rangle\langle i|\otimes U_i$.}
\label{fig:lcu_main}
\end{figure*}

\begin{figure*}[bt]
\centering
\begin{tabular}{c}
\\
\\
~~~~~\Qcircuit @C=0.6em @R=1.4em {
  \lstick{\ket{0}}&\gate{H}                           & \targ  & \gate{e^{-iZ\phi_n }} & \targ 
  &\qw & \targ & \gate{e^{-iZ\phi_{n-1} }} & \targ
  &\qw&\cdots&& \qw & \targ  & \gate{e^{-iZ\phi_1 }} & \targ
  &\gate{H} & \qw \\
  \lstick{\ket{0}^{\otimes a}}&\multigate{1}{U}  & \ctrlo{-1} & \qw & \ctrlo{-1}  
  &\multigate{1}{U^\dagger}& \ctrlo{-1}& \qw & \ctrlo{-1} 
  &\qw&\cdots& & \multigate{1}{U}& \ctrlo{-1}& \qw & \ctrlo{-1}
  & \qw&\qw \\
                              &\ghost{U}         & \qw & \qw & \qw                             
  &\ghost{U^\dagger}& \qw & \qw & \qw
  &\qw &\cdots&& \ghost{U} & \qw & \qw & \qw
  & \qw&\qw \\
  }
\\
\\
\end{tabular}
\caption{Quantum circuit for quantum singular value (eigenvalue) transformation for real polynomials $P$ of odd degree $n$.
$U$ denotes an $a$-block-encoding of a Hermitian operator $A$.
The NOT gate is controlled by $\ket{0}^{\otimes a}$ of the $a$ qubits, which is represented by the white circle $\circ$.
For a given real polynomial $P$ of degree $n$, the $n$ circuit parameters $\{\phi_i\}_{i=1}^{n}$ are calculated in $\mathcal{O}({\rm poly}(n))$ classical computation time.
With such parameters, this circuit results in an $(a+1)$-block-encoding of $P(A)$.
}
\label{fig:qet_circuit}
\end{figure*}

\begin{figure*}[htb]
    \centering
    \includegraphics[width=0.8\textwidth]{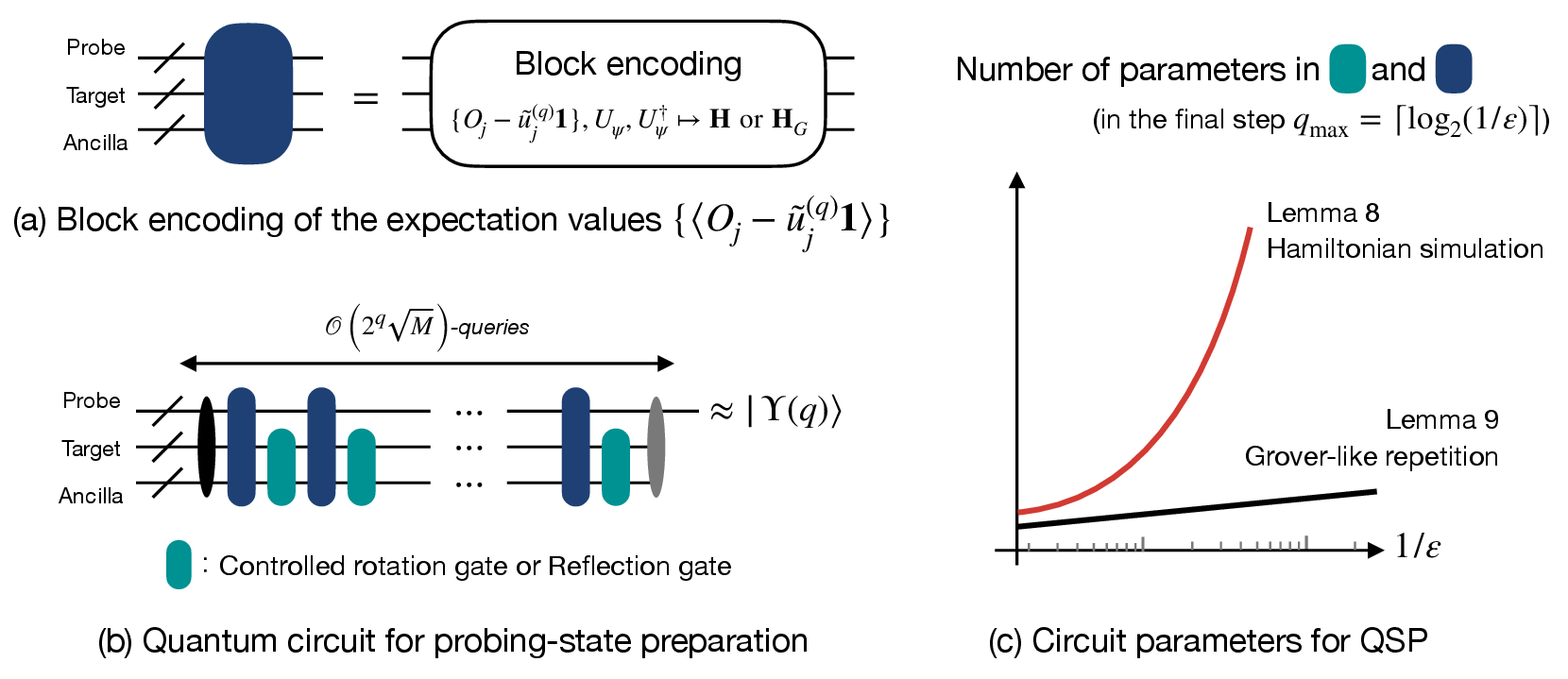}
    \caption{Preparation of the probing state $\ket{\Upsilon(q)}$ at the $q$-th iteration.
    The blue unitary gate in (a) and (b) denotes a block-encoding of the Hamiltonian $\mathbf{H}$ or $\mathbf{H}_{\rm G}$ in Eq.~\eqref{eq:sp_HS_targetHamiltonian} or \eqref{eq:Grover_sp_Hamiltonian}, respectively, and the green unitary gate in (b) is a controlled rotation or reflection gate.
    The blue and green unitary gates contain parameterized gates for QSP, and the total number of circuit parameters scales as in (c) with respect to the target root MSE $\varepsilon$.
    As the graph in (c) indicates, the method in Lemma~\ref{thm:sp_Grover} has a much smaller number of circuit parameters for QSP compared to that of Lemma~\ref{thm:sp_HS}.}
    \label{fig:sp_circuit_prop}
\end{figure*}

\subsection{Overview of quantum matrix arithmetics}\label{supple_sec:preliminary}
Here, we review some important results in efficiently calculating block-encoded matrices on a quantum computer.
The basic tool to represent matrices by unitary operators of dilated quantum systems is \textit{block encoding}: 
\begin{dfn}
    [Block encoding]\label{dfn:block_encoding}
    For positive values $\alpha,\varepsilon$ and a non-negative integer $a$, we say that an $(n+a)$-qubit unitary $U$ is an ($\alpha,a,\varepsilon$)-block-encoding of an $n$-qubit operator $A$, if 
    $$
    \|A-\alpha (\bra{0}^{\otimes a}\otimes \bm{1})U (\ket{0}^{\otimes a}\otimes \bm{1})\|\leq \varepsilon.
    $$
    For simplicity, we shorten the perfect (i.e., $\alpha=1$ and $\varepsilon=0$) block encoding of $A$ as $a$-block-encoding of $A$.
\end{dfn}
\noindent
For instance, any unitary operator (e.g., a Pauli operator $X\otimes Z\otimes \cdots$) is trivially a $0$-block-encoding of itself.
There are various ways to construct block encodings; see Ref.~\cite{gilyen2019quantum}.
Specifically, we here focus on the method called the linear combination of unitaries (LCU)~\cite{childs2012hamiltonian,berry2015hamiltonian}.
Let $A:=\sum_{i=1}^m c_iU_i$ be a linear combination of unitary operators $\{U_i\}_{i=1}^m$ with real coefficients $c_i\in \mathbb{R}$. 
Without loss of generality, we assume $c_i>0$ because $-1$ can be absorbed into $U_i$.
In order to implement $A$, we use the following two unitary operations.
The first one, called PREPARE, encodes the positive coefficients $\{c_i\}_{i=1}^{m}$ as
$$
V:\ket{\bm{0}} \mapsto \sum_{i=1}^{m} \sqrt{\frac{c_i}{\|c\|_1}}|i\rangle,
$$ 
where $\|\bullet\|_1$ denotes $L^1$-norm, and $|\bm{0}\rangle$ and $|i\rangle$ denote an initial state and the computational basis in a $\lceil {\log_2 m\rceil}$-qubit ancilla system, respectively.
The other, called SELECT, encodes the unitary operartors $U_i$ conditioned by the $\lceil {\log_2 m\rceil}$-qubit ancilla system:
$$
W=\sum_{i=1}^m|i\rangle\langle i|\otimes U_i.
$$
Using the two operations $V$ and $W$, it can be shown that the unitary operator $(V^\dagger\otimes \bm{1})W(V\otimes \bm{1})$ is a $(\|c\|_1,\lceil {\log_2 m\rceil},0)$-block-encoding of $A$, as in Fig.~\ref{fig:lcu_main}.
Note that if the coefficients in LCU are controlled by other qubit registers, it may be useful to modify the PREPARE operator instead of including the phase of $c_i$ to the SELECT operator, in order to save the number of controlled operations.

Once we have a block encoding of a target operator, we can systematically transform the block encoding to perform various tasks. Here, we show examples of such transformations that will be used in the following subsection.

\begin{lem}
    [Uniform singular value amplification~\cite{low2017hamiltonian,gilyen2019quantum}]\label{supple_lem:uniform_amp}
    Let $\gamma>1$ and let $\delta,\varepsilon\in (0,1/2)$.
    Suppose we have an $a$-block-encoding $U$ of $A,~(\|A\|\leq (1-\delta)/\gamma)$. Then, we can implement a $(1,a+1,\varepsilon)$-block-encoding of $\gamma A$ with $m=\mathcal{O}({\gamma}\delta^{-1}\log({\gamma/\varepsilon}))$ queries to $U$ or $U^\dagger$, $2m$ uses of NOT gates controlled by $a$-qubit, $\mathcal{O}(m)$ single-qubit gates, and $\mathcal{O}({\rm poly}(m))$ classical computation to find quantum circuit parameters.
\end{lem}

\begin{lem}
    [Quantum eigenvalue transformation by Chebyshev polynomials~\cite{gilyen2019quantum}]\label{supple_lem:QET_by_Tmx}
    Let $m$ be a positive integer, and let $U$ be an $a$-block-encoding of a Hamiltonian $H$. Then, for the $m$-th Chebyshev polynomial of the first kind $T_m(x)$, we can implement a $(1,a,0)$-block-encoding of $T_m(H)$, with $m$ uses of $U$ or $U^\dagger$ and $m$ uses of reflection on $\ket{0}^{\otimes a}$.
\end{lem}

\begin{lem}
    [Optimal block-Hamiltonian simulation~\cite{Low2019hamiltonian}]\label{supple_lem:opt_blockHS}
    Let $t\in \mathbb{R}\backslash\{0\}$, $\varepsilon''\in(0,1)$, and let $U$ be a $(1,a,0)$-block-encoding of a Hamiltonian $H$.
    Then, we can implement a $(1,a+2,\varepsilon'')$-block-encoding of $e^{itH}$, with $4Q$ queries to controlled $U$ or its inverse, $2Q$ uses of NOT gates controlled by $(a+1)$-qubit, $\mathcal{O}(Q)$ uses of single-qubit or two-qubit gates, and $\mathcal{O}({\rm poly}(Q))$ classical computation to find quantum circuit parameters, where $Q=\mathcal{O}(|t|+\log(1/\varepsilon''))$.
\end{lem}

Lemma~\ref{supple_lem:uniform_amp} and Lemma~\ref{supple_lem:QET_by_Tmx} can be implemented with a quantum circuit in Fig.~\ref{fig:qet_circuit}.
Also, we can implement the optimal Hamiltonian simulation Lemma~\ref{supple_lem:opt_blockHS} with a similar circuit as Fig.~\ref{fig:qet_circuit}; the explicit circuit constructions are provided in Refs.~\cite{low2017optimalHSbyQSP,Low2019hamiltonian,gilyen2019quantum}.

Importantly, the quantum circuit in Fig.~\ref{fig:qet_circuit} reflects the underlying structure that are common in the above lemmas; this quantum circuit implements a general method, called the quantum singular value transformation (QSVT), to transform singular values (eigenvalues) of a block-encoded matrix based on a large class of polynomials~\cite{gilyen2019quantum}.
The QSVT uses the idea of quantum signal processing (QSP)~\cite{PhysRevX.6.041067,Low2019hamiltonian} that characterizes achievable 1-qubit unitary transformations comprised of an alternating 1-qubit gate sequence of the \textit{signal} rotation with a unknown angle and the \textit{processing} rotation with a controllable angle.
To bridge the gap between QSP and QSVT, \textit{Qubitization}~\cite{Low2019hamiltonian} is a crucial technique that splits (a part of) ancilla-target systems into some qubits labeled by the eigenvalue (singular values~\cite{gilyen2019quantum}) and constructs parallel signal rotations over the qubits (e.g., $U$ and $U^\dagger$ in Fig.~\ref{fig:qet_circuit}).
Here, the rotation angle depends on the corresponding singular value.
Then, using additional processing rotations with controllable parameters $\{\phi_i\}$ (likewise the controlled rotation between $U$ and $U^\dagger$ in Fig.~\ref{fig:qet_circuit}), we can transform the singular values in parallel by a polynomial that depends on the controllable parameters $\{\phi_i\}$; the achievable polynomials in QSVT are characterized by QSP.
See the review~\cite{PRXQuantum.2.040203} for details of the theoretical perspective of QSVT.

In practice, a typical flow of QSVT consists of two steps: (i) finding the circuit parameter $\{\phi_i\}_{i=1}^n$ (called the phase sequence) for a given degree-$n$ real polynomial $P$ on classical computers, (ii) running the $\mathcal{O}(n)$-depth quantum circuit in Fig.~\ref{fig:qet_circuit} on a quantum computer, using the classically tuned parameters.
Note that in the process (ii), we may need a post-selection on ancilla qubits. 
For a given degree-$n$ polynomial $P$ that has definite parity and $P(x)\in[-1,1]$ for $x\in [-1,1]$, the $n$ circuit parameters $\{\phi_i\}_{i=1}^n$ in Fig.~\ref{fig:qet_circuit} can be found by $\mathcal{O}({\rm poly}(n,\log(1/\delta)))$ classical computation for some error $\delta$~\cite{gilyen2019quantum}.
Then, using the parameters this circuit results in an $(a+1)$-block-encoding of $P(A)$.
In Lemma~\ref{supple_lem:uniform_amp}, we can take an odd real polynomial $P(x)$ with degree $m=\mathcal{O}({\gamma}\delta^{-1}\log({\gamma/\varepsilon}))$ such that $P(x)\approx \gamma x$ holds~\cite{low2017hamiltonian,gilyen2019quantum}.
In particular, the phase sequence $\{\phi_i\}$ for the Chebyshev polynomial of the first kind $T_{n}$ is analytically calculated and has a unique structure (Lemma 9 in Ref.~\cite{gilyen2019quantum}); as a result, we can eliminate the additional ancilla qubit and replace the $2n$ controlled NOT gates with $n$ reflections on $\ket{0}^{\otimes a}$ in the circuit of Fig.~\ref{fig:qet_circuit}.

\subsection{Probing-state preparation}
Before proceeding to the proof of Lemmas~\ref{thm:sp_HS} and \ref{thm:sp_Grover}, we first provide an overview of our method to prepare $\ket{\Upsilon(q)}$ in Fig.~\ref{fig:sp_circuit_prop}. 
As seen in the figure, the proposed two methods to prepare $\ket{\Upsilon(q)}$ have a similar structure: Fig.~\ref{fig:sp_circuit_prop}(a) shows the block encoding of a Hamiltonian that encodes the expectation values of observables $\braket{O_j-\tilde{u}^{(q)}_j\bm{1}}$ and Fig.~\ref{fig:sp_circuit_prop}(b) shows alternating applications of the block encoding and a processing operation with tuned parameters, which is based on Lemma~\ref{supple_lem:QET_by_Tmx} or Lemma~\ref{supple_lem:opt_blockHS}.
In the proof of Lemmas~\ref{thm:sp_HS} and \ref{thm:sp_Grover}, we 
depict the circuit for (a) and (b), respectively; then we evaluate the approximation error between the final state of the circuit and the probing state $\ket{\Upsilon(q)}$.

More detailed proofs and explicit quantum circuit diagrams for these lemmas are provided in Appendix~\ref{supple_sec:proposed_alg}.

\subsubsection{Hamiltonian simulation}
Let us first present an informal lemma (its formal version is Lemma~\ref{supple_lem:statepre_HS} in Appendix~\ref{supple_sec;sp4iterativegradest}) regarding the complexity of probing-state preparation when we utilize the Hamiltonian simulation based on QSP:

\begin{lem}
    [Informal]\label{thm:sp_HS}
    Suppose that we have access to block-encoded observables $\{O_j\}_{j=1}^M$ in $d$ dimension, a $\log_2 d$-qubit state preparation $U_{\psi}$, and its inverse $U_{\psi}^\dagger$, such that $M>\mathcal{O}(\log d)$. Then, we can prepare the probing state $\ket{\Upsilon(q)}$ for any integer $q\geq 0$ and $\tilde{u}^{(q)}_j\in [-1,1]$ up to $1/12$ Euclidean distance error, using an $$\mathcal{O}(M+\log_2 d)\mbox{-qubit}$$ circuit regardless of $q$.
    Furthermore, each quantum circuit with $q$ requires $$\mathcal{O}({\rm poly}(2^q\sqrt{M\log d})+{\rm poly}(\sqrt{M}(q+\log M)))$$ classical computation for finding circuit parameters, and it consists of $$\mathcal{O}(2^q\sqrt{M\log d})~~\mbox{uses~of}~~U_{\psi}~\mbox{and}~U_{\psi}^\dagger,$$ 
    $\mathcal{O}(2^qM(q+\log M))$ uses of unitary gates for block-encoded observables, and $\mathcal{O}(2^qM(q+\log M)\log dM)$ uses of single-qubit and two-qubit gates.
\end{lem}

\begin{proof}[Sketch of the proof]
    To encode the ideal phase $f(\bm{x})$ in Eq.~\eqref{eq:target_phase_fn},
    we construct the Hamiltonian $\mathbf{H}$ with the observables $\{O_j\}$ as
\begin{equation}\label{eq:sp_HS_targetHamiltonian}
        \mathbf{H}:=\sum_{\bm{x}\in G_p^M} \tilde{f}(\bm{x})\ket{\bm{x}}\bra{\bm{x}},
\end{equation}
where the target observables are approximately encoded in the eigenvalues $\tilde{f}(\bm{x})$ as
\begin{equation}\label{eq:main_tildef_aaprx}
    \tilde{f}(\bm{x}) \approx \frac{1}{{\sigma}}\sum_{j=1}^M x_j\braket{\tilde{O}^{(q)}_j},~~\sigma=\mathcal{O}(\sqrt{M\log d}),
\end{equation}
Here, $\sigma$ denotes the rescaling factor of the Hamiltonian $\mathbf{H}$ such that $\mathbf{H}$ can be encoded in a unitary operator.
Also, we defined $\tilde{O}_j^{(q)}:=(O_j-\tilde{u}^{(q)}_j\bm{1})/2$ for the identity $\bm{1}$.
It is crucial for the total query complexity of $U_{\psi}$ to construct the block encoding of $\mathbf{H}$ with the $\mathcal{O}(\sqrt{M})$ rescaling factor, while naive block encoding for $\mathbf{H}$ has $\mathcal{O}(M)$ rescaling factor, as mentioned below.

Setting aside the details for now, we suppose that we have a (perfect) block encoding of the Hamiltonian $\mathbf{H}$ that acts on the $pM$-qubit probe system, $\log_2 d$-qubit target system, and ancilla systems (specified below), as illustrated in Fig.~\ref{fig:sp_circuit_prop}.
Then, the optimal Hamiltonian simulation protocol Lemma~\ref{supple_lem:opt_blockHS} yields a quantum circuit $W$ for an $\epsilon'$-precise block encoding of time evolution operator 
\begin{equation}\label{eq:sphs_propagator}
    e^{i\mathbf{H}t} = \sum_{\bm{x}\in G_p^M} e^{i\tilde{f}(\bm{x})t}\ket{\bm{x}}\bra{\bm{x}}
\end{equation}
with time $t>0$, using $\mathcal{O}(t+\log(1/\epsilon'))$ queries to the block encoding of $\mathbf{H}$.
Note that in the case $t=2\sigma t'$ for some positive integer $t'$, the resulting time evolution operator approximates $t'$ times applications of the \textit{phase oracle} for an affine linear function $f(\bm{x})=\sum_j x_j \braket{O_j}$ in the theory of gradient estimation.
Then, applying $W$ for $$t:=2^{p+q+2}\sigma$$ 
to the uniform superposition state $\ket{+}^{\otimes pM}$ in the probe system and the initial state $\ket{\bm{0}}$ in the ancilla-target systems, we can approximately prepare the probing state:
\begin{align}\label{eq:HS_targstate_approx}
    W\ket{+}^{\otimes pM}\ket{\bm{0}}\approx \ket{\Upsilon(q)}\ket{\bm{0}}.
\end{align}
From the triangle inequality, this approximation error is given by the sum of $\epsilon'$ (more precisely, $\epsilon'+\sqrt{2\epsilon'}$ in Euclidean distance) from the transformation $\mathbf{H}\mapsto e^{i\mathbf{H}t}$ and the error from the observable encoding in Eq.~\eqref{eq:main_tildef_aaprx}, that is,
\begin{equation}\label{eq:sp_HS_approx_error}
    \epsilon'+\sqrt{2\epsilon'}+\left\|\frac{1}{\sqrt{2^{pM}}}\sum_{\bm{x}\in G_p^M} e^{i\tilde{f}(\bm{x})t}\ket{\bm{x}}-\ket{\Upsilon(q)}\right\|.
\end{equation}

To quantify the error of the third term in Eq.~\eqref{eq:sp_HS_approx_error} (or the error of Eq.~\eqref{eq:main_tildef_aaprx}), we here construct the block encoding of the Hamiltonian $\mathbf{H}$ from the state preparation $U_{\psi}$ ($U_{\psi}^\dagger$) and unitary gates $\{B_j\}$ that are $a$-block-encodings of observables $\{O_j\}$ for some $a\in \mathbb{N}$.
The block encoding of $\mathbf{H}$ consists of (i) the LCU method, (ii) adding the $pM$-qubit control to the LCU operations, (iii) the singular value amplification, and (iv) multiplying the conjugation of $U_{\psi}$.
Let us describe these procedures step by step in the following.

First, we utilize the LCU method with the controlled version of each $B_j$ to construct a quantum circuit $U^{(\bm{x})}$ for a block encoding of $$\frac{1}{M}\sum_{j=1}^M x_j \tilde{O}^{(q)}_j$$ (with the rescaling factor $M$, instead of $\sigma$) for a given $\bm{x}\in G_p^M$. 
The LCU requires additional $\mathcal{O}(\log M)$ qubits to encode the coefficients $\{x_j/M\}$.
In step (ii), using controlled versions of $U^{(\bm{x})}$, we can obtain a quantum circuit for 
\begin{equation}\label{eq:U_sel}
    U'_{\rm SEL}:=\sum_{\bm{x}\in G_{p}^M} \ket{\bm{x}}\bra{\bm{x}}\otimes U^{(\bm{x})},
\end{equation}
which can be implemented with at most $\mathcal{O}(pM\log(1/\delta))$ depth regarding the implementation error $\delta$ of the coefficients $\{x_j/M\}$; see Remark~\ref{rem:design_prepare} in Appendix~\ref{supple_sec:proposed_alg}.

To understand why we must have step (iii), it is informative to observe the action realized in step (iv).
Namely, using $U'_{\rm SEL}$, $U_{\psi}$, and its conjugation $U^\dag_{\psi}$, we obtain
\begin{equation}\label{eq:hs_proof_1}
    \left(\bm{1}\otimes U_{\psi}^\dagger\right)\cdot U'_{\rm SEL}\cdot \left( \bm{1}\otimes U_{\psi}
\right),
\end{equation}
which is a block encoding of the following Hamiltonian:
$$
\mathbf{H}':=\sum_{\bm{x}} \left(\sum_j{x}_j\cdot\frac{\langle \tilde{O}^{(q)}_j\rangle}{M}\right) \ket{\bm{x}}\bra{\bm{x}}.
$$
This Hamiltonian $\mathbf{H}'$ has the normalization factor $M$ that is quadratically larger than $\sigma=\mathcal{O}(\sqrt{M})$ in $\mathbf{H}$.
In estimating the expectation value $\langle \tilde{O}^{(q)}_j\rangle$ via gradient estimation, we need to amplify $\mathbf{H}'$ by the factor $M$ to prepare Eq.~(\ref{eq:targetstate}).
This can be achieved by using a time evolution
$e^{i\mathbf{H}'t}$ of $t=2^{p+q+2}M$, but the linear dependence on $M$ in the evolution time $t$ results in no speedup for the total queries to $U_{\psi}$ regarding the number $M$ of observables.

To obtain the quadratic speedup regarding $M$, the method in Ref.~\cite{van2023quantum} uses the singular value amplification~Lemma~\ref{supple_lem:uniform_amp}.
Importantly, this amplification can be performed with no use of $U_{\psi}$.
From the random matrix series inequality, we can show that for a large part of $G_{p}^M$ (more precisely, for a subset $F\subset G_p^M$ such that $|F|\geq (1-\delta')|G_p^M|$ for any $\delta'>0$), the condition of the amplification is satisfied for $\gamma=M/\sigma=\mathcal{O}(\sqrt{M})$, as well as the analysis in Ref.~\cite{van2023quantum}.
Therefore, we can amplify the block encoding $U^{(\bm{x})}$ by $\gamma$ for such $\bm{x} \in F$; as a result, we have a quantum circuit for 
\begin{equation}\label{eq:U_obs}
    U_{\rm obs}:=\sum_{\bm{x}\in G_{p}^M} \ket{\bm{x}}\bra{\bm{x}}\otimes U^{(\bm{x})}_{\rm obs},
\end{equation}
where $U^{(\bm{x})}_{\rm obs}$ is an $\epsilon''$-precise block encoding of the Hamiltonian $\sigma^{-1}\sum_j x_j\tilde{O}^{(q)}_j$ if $\bm{x}\in F$.
By multiplying $U_{\psi}$ and $U_{\psi}^\dagger$ into $U_{\rm obs}$ as well as Eq.~\eqref{eq:hs_proof_1}, we finally arrive at the block encoding of the Hamiltonian $\mathbf{H}$.
We note that the amplification is valid when $\sigma = \mathcal{O}(\sqrt{M\log_2 d})$ is smaller than $M$, and this is satisfied by the condition $M>\mathcal{O}(\log d)$.

The main contribution to the space complexity comes from the steps (i) the LCU method and (ii) adding the $pM$-qubit control to $U^{(\bm{x})}$.
The other processes (iii), (iv), and the Hamiltonian simulation protocol $\mathbf{H}\mapsto e^{i\mathbf{H}t}$ introduce constant or no ancilla qubits.
The LCU method (i) introduces $\mathcal{O}(\log M)$-qubit registers, and thus the number of qubits of $U'_{\rm SEL}$ scales as $pM+\log M+\log d$.
Recalling that the precision parameter $p$ is fixed to a constant (i.e, $p=3$) in the probing state $\ket{\Upsilon(q)}$, we conclude that the total space complexity is $\mathcal{O}(M+\log_2 d)$, which is independent from the root MSE $\varepsilon$.

As for the gate complexity, we here focus on the number of queries to $U_{\psi}$ (or $U_{\psi}^\dagger$) and $U'_{\rm SEL}$.
A comprehensive analysis on the total gate complexity is provided in Lemma~\ref{supple_lem:statepre_HS} in Appendix~\ref{supple_sec;sp4iterativegradest}.
Here, we need to carefully choose the parameters $\varepsilon',\varepsilon''$, and $\delta'$ so that the entire approximation error Eq.~(\ref{eq:sp_HS_approx_error}) is smaller than 1/12; in particular, this is achieved by taking $\varepsilon'$ and $\delta'$ as some constants, and $\varepsilon''=\mathcal{O}(1/(2^q\sigma))$.
Since the block encoding of $\mathbf{H}$ consists of two uses of $U_{\psi}$ and $U_{\psi}^\dagger$ and $m:=\mathcal{O}(M\sigma^{-1}\log(M/(\sigma \varepsilon'')))$ uses of control-$U'_{\rm SEL}$, the total queries can be evaluated by multiplying $\mathcal{O}(t)$ for the Hamiltonian simulation protocol, which yields the gate complexity as stated in Lemma~\ref{thm:sp_HS}.
In addition, the number of circuit parameters for the singular value amplification in the step (iii) and the Hamiltonian simulation protocol are given by $m$ and $t$, respectively.
Thus, we need to tune the parameters via classical computation with time complexity of $\mathcal{O}({\rm poly}(m))$ and $\mathcal{O}({\rm poly}(t))$ time, respectively, as discussed in Sec.~\ref{supple_sec:preliminary}.
\end{proof}

Before proceeding to the Grover-like repetition, we remark that
our scheme has a significant improvement in space complexity compared to the previous non-iterative counterparts~\cite{PhysRevLett.129.240501,van2023quantum}.
While the non-iterative methods determine all of the $\mathcal{O}(\log 1/\varepsilon)$ binary fraction bits of $\braket{O_j}$ by a quantum circuit with additional $\mathcal{O}(M\log(1/\varepsilon))$ (or $\mathcal{O}(M\log(M/\varepsilon))$) qubits to read out the observables, our scheme determines only 1 binary fraction bit of $\braket{O_j}$ at each iteration step.
Specifically, the previous non-iterative method~\cite{van2023quantum} uses the Hamiltonian simulation with $t=2^{p}$ in Eq.~\eqref{eq:sphs_propagator}, where the number of read-out qubits is $p=\mathcal{O}(\log(\sqrt{M\log d}/\varepsilon))$ (i.e., $q=0$).
In contrast, by the adaptive nature of our scheme, space overhead of quantum circuits for the probing state $\ket{\Upsilon(q)}$ is $\mathcal{O}(M)$ ($\mathcal{O}(M+\log_2 d)$ in total), and this is independent of the estimation precision $\varepsilon$.
Here, a similar improvement of space overhead can be found in the previous works on the adaptive (iterative) versions of quantum phase estimation~\cite{kitaev1995quantum,higgins2007entanglement,higgins2009demonstrating,kimmel2015robust,dutkiewicz2022heisenberg}.
Also, we remark that the significant reduction of space overhead directly leads to the reduction of the number of controlled operations, which is also crucial in practical implementation.




\subsubsection{Grover-like repetition}

The quantum circuit employed in Lemma~\ref{thm:sp_HS} requires classical tuning of $\tilde{\mathcal{O}}(\sqrt{M}/\varepsilon)$ circuit parameters (more precisely, $\tilde{\mathcal{O}}(2^q \sqrt{M})$ parameters in step $q$) for QSP.
The classical computation for finding $\tilde{\mathcal{O}}(\sqrt{M}/\varepsilon)$ parameters is also required in the existing work by van Apeldoorn {\it et al.}~\cite{van2023quantum}.
As discussed in Sec.~\ref{sec:intro_mainres}, this is in practice a significant overhead in classical computation time, for instance, $10^4$--$10^6$ seconds for the case of $M=10^4$ and $\varepsilon=10^{-4}$.
Although the classical computation of $n$ parameters can be done with time complexity of $\mathcal{O}({\rm poly}(n))$~\cite{gilyen2019quantum}, it is extremely challenging if we desire to compute for $n>10^6$ which may be required in, e.g., quantum chemistry applications. 
Even if we choose to compensate for the computation time with low accuracy in parameter finding, this may eventually negate the quantum enhancement in the gradient estimation algorithm.


To avoid the numerical instability, we here provide an alternative way to prepare the probing state~\eqref{eq:targetstate} using the \textit{Grover-like repetition}, which is a special case of QSP such that the corresponding quantum circuit parameters are analytically derived.

\begin{lem}
    [Informal]\label{thm:sp_Grover}
    Suppose we have access to block-encoded observables $\{O_j\}_{j=1}^M$ in $d$ dimension, a $\log_2 d$-qubit state preparation $U_{\psi}$, and its inverse $U_{\psi}^\dagger$, such that $M>\mathcal{O}(\log d)$.
    If the iteration step $q\geq 0$ satisfies the following condition $(\delta':=2^{-14})$
    \begin{equation}\label{xabst:q_threshold}
        q\geq \log_4\left[\frac{2^{{3}}\cdot 33^3}{625\ln (2d/\delta')}\frac{\left\lceil{\sqrt{2(M+1)\ln(2d/\delta')}}\right\rceil}{\sqrt{\ln (2d/\delta')}}\right],
    \end{equation}
    and for given $\tilde{u}^{(q)}_j\in [-1,1]$,
    \begin{equation}\label{eq:cond4grad_est_spG}
        \left|\bra{\psi}\left(O_j-\tilde{u}^{(q)}_j\bm{1}\right)\ket{\psi}\right|\leq 2^{-q}~~\mbox{for~all}~~j=1,2,...,M,
    \end{equation}
    holds,
    then we can successfully prepare the probing state $\ket{\Upsilon(q)}$ up to $1/12$ Euclidean distance error, with probability at least $0.462$ with ancilla qubits measurement result indicating success, using an $$\mathcal{O}(M+\log_2 d)\mbox{-qubit}$$
    circuit regardless of $q$.
    Furthermore, each quantum circuit with $q$ requires $$\mathcal{O}({\rm poly}(\sqrt{M}(q+\log M)))$$ 
    classical computation, and it has the same gate complexity as that of Lemma~\ref{thm:sp_HS}.
\end{lem}
The formal version of this lemma is Lemma~\ref{supple_lem:statepre_Grover} in Appendix~\ref{supple_sec;sp4iterativegradest}.
Here, we remark that the condition Eq.~\eqref{eq:cond4grad_est_spG} is trivially satisfied at the iteration step $q$ in Algorithm~\ref{alg:main} when all the previous gradient estimations are successfully performed, as discussed in Sec.~\ref{sec:main_III}.

\begin{proof}[Sketch of the proof]

This method uses eigenvalue transformation of (a slightly modified version of) the Hamiltonian $\mathbf{H}$ in Eq.~(\ref{eq:sp_HS_targetHamiltonian}) based on the Chebyshev polynomial of the first kind $T_{t}(x)$, that is, Lemma~\ref{supple_lem:QET_by_Tmx}.
Here, the degree $t$ is given by $t=2^{p+q+2}\sigma$.
We note that since the resulting operator $T_t[\mathbf{H}]$ is non-unitary, post-selection of ancilla qubits in the block encoding is required to implement $T_t[\mathbf{H}]$, meaning that the state preparation by this method is stochastic.
As mentioned in Sec.~\ref{supple_sec:preliminary}, the eigenvalue transformation based on the Chebyshev polynomials $T_t(x)$ can be considered as a special case in QSVT; that is, it has an analytical solution of circuit parameters (or phase sequence).
Therefore, in this alternative method, the circuit parameter finding is required only for the construction of the Hamiltonian $\mathbf{H}$; the total runtime for the classical computation is at most
$\mathcal{O}({\rm polylog}(1/\varepsilon))$ (more precisely, $\mathcal{O}({\rm poly}(\sqrt{M}(q+\log M)))$ at step $q\leq \lceil\log_2 (1/\varepsilon) \rceil$).

Now, we consider the action of the resulting operator $T_t[\mathbf{H}]$.
Applying $T_t[\mathbf{H}]$ to the uniform superposition state $\ket{+}^{\otimes pM}$, we obtain an unnormalized state proportional to
\begin{align}\label{eq:show_nonlineariy}
    \sum_{s=\pm 1}\sum_{\bm{x}\in G_p^M}e^{ist \arccos{\left[\tilde{f}(\bm{x})\right]}}\ket{\bm{x}}.
\end{align}
Therefore, we need to correct the phase of $\ket{\bm{x}}$ from $e^{\pm it\arccos{[\tilde{f}(\bm{x})]}}$ to $e^{it\tilde{f}(\bm{x})}$ in order to match the probing state $\ket{\Upsilon(q)}$.

As for the sign $s=\pm 1$ of the phase, we can correct it by slightly modifying the Hamiltonian $\mathbf{H}$ with the help of an additional 1-qubit ancilla system.
Let $\mathbf{H}_{\rm G}$ be a Hamiltonian defined as 
\begin{equation}\label{eq:Grover_sp_Hamiltonian}
    \mathbf{H}_{\rm G} := \sum_{(\bm{x},y)\in G_p^M\times G_1} \tilde{f}'(\bm{x},y)\ket{\bm{x},y}\bra{\bm{x},y},
\end{equation}
where $\ket{y}$ denotes a computational basis on the additional ancilla system and 
\begin{equation}\label{eq:gv_temp_eq}
    \tilde{f}'(\bm{x},y) \approx \frac{1}{{\sigma'}}\left(y\braket{O_{M+1}}+\sum_{j=1}^M x_j\braket{\tilde{O}^{(q)}_j}\right).
\end{equation}
Here, $O_{M+1}$ denotes a 1-qubit observable proportional to the identity, and $\sigma'=\mathcal{O}(\sqrt{M\log d})$ is the rescaling factor for block encoding.
The block encoding of $\mathbf{H}_{\rm G}$ can be constructed in a similar way to that of the Hamiltonian $\mathbf{H}$; the corresponding circuit has the same complexity as that of $\mathbf{H}$.
Then, applying $T_t[\mathbf{H}_{\rm G}]$ to the uniform superposition state $\ket{+}^{\otimes pM+1}$, we can show that the following holds:
(for simplicity, we write $O_{M+1}$ as $\tilde{O}^{(q)}_{M+1}$)
\begin{widetext}
\begin{align}\label{eq:spg_temp1}
    &\frac{1}{\mathcal{N}_t\sqrt{2^{pM}}}\sum_{\bm{x},y} T_{t}\left(\tilde{f}'(\bm{x},y)\right)\ket{\bm{x},y}\approx\frac{1}{2\sqrt{2^{pM}}}\sum_{s=\pm 1}\sum_{(\bm{x},x_{M+1})\in G_p^M\times G_1} e^{ist\left({\pi}/{2}-{(\sigma')}^{-1}\sum_{j=1}^{M+1}x_j\langle \tilde{O}^{(q)}_j\rangle\right)}\ket{\bm{x}}\ket{x_{M+1}}\notag\\[6pt]
    &~~~~~=\frac{1}{\sqrt{2}}\sum_{s=\pm 1}\left(\frac{1}{\sqrt{2^{pM}}}\sum_{\bm{x}\in G_p^M} e^{2\pi i 2^p \sum_{j=1}^M (sx_j) \pi^{-1}2^{q+1}\langle \tilde{O}^{(q)}_j\rangle}\ket{\bm{x}}\otimes \frac{1}{\sqrt{2}}\sum_{k\in G_1}e^{2\pi i2k(s/4)}\ket{k}\right),
\end{align}
\end{widetext}
where $\mathcal{N}_t$ is the normalization factor. In the second equality, we chose the constant factor of $O_{M+1}$ as
$$
O_{M+1}:=\frac{\pi(1/4+4l)}{2^{p+q}}I,
$$
where $l$ is an arbitrary integer satisfying $\|2^{q+1}O_{M+1}\|\leq 1$ and $I$ denotes the 1-qubit identity.
In the next paragraph, we explain the approximation of the first line in Eq.~\eqref{eq:spg_temp1} in detail.
Suppose that the approximation holds, then we can correct the sign $s=\pm1$ in Eq.~\eqref{eq:spg_temp1} and as a result obtain the probing state $\ket{\Upsilon(q)}$.
That is, by applying the inverse quantum Fourier transformation to the 1-qubit additional anciila system in the state Eq.~\eqref{eq:spg_temp1}, followed by controlled $X^{\otimes pM}$ gate (which flips the sign as $X^{\otimes pM}\ket{\bm{x}}=\ket{-\bm{x}}$), we obtain the probing state Eq.~\eqref{eq:targetstate} in the $pM$-qubit registers.

Now, we argue that the approximation in the first line of Eq.~\eqref{eq:spg_temp1} follows from the non-linearity of the phase function $\arccos{x}$ in Eq.~\eqref{eq:show_nonlineariy}.
Here, we recall that $\tilde{f}'(\bm{x},y)$ in Eq.~\eqref{eq:gv_temp_eq} is approximately given by the linear combination of $\braket{O_j-\tilde{u}^{(q)}_j\bm{1}}/\sigma'$.
Using this fact, we can show that $|\tilde{f}'(\bm{x},y)|$ scales as $\mathcal{O}(\sqrt{M}/(2^{q}\sigma'))$ for a large part of $G_p^M\times G_1$, from the assumption of $|\braket{\tilde{O}^{(q)}_j}|\leq 1/2^{q+1}$ and the Hoeffding's inequality for $(\bm{x},y)$ which is uniformly distributed.
Moreover, the error between the function $\arccos(x)$ and a linear function $\pi/2-x$ can be upper bounded by $\mathcal{O}(|x|^3)$.
As a result, the non-linearity of $\arccos{x}$ in Eq.~\eqref{eq:show_nonlineariy} can be suppressed for a large part of $G_p^M\times G_1$ as
$$
\left|\arccos{[\tilde{f}'(\bm{x},y)]}-\frac{\pi}{2}+\tilde{f}'(\bm{x},y)\right|=\mathcal{O}\left(\left(\frac{\sqrt{M}}{2^q\sigma'}\right)^{3}\right).
$$
In the approximation error of Eq.~\eqref{eq:spg_temp1}, the above error is amplified by the factor of $t=\mathcal{O}(2^q\sigma')$ and the amplified error decreases 
\begin{equation}\label{eq:approxerr_arccos}
    \mathcal{O}\left(\frac{\sqrt{M}}{4^q}\frac{M}{(\sigma')^2}\right)
\end{equation}
for $\sigma'=\mathcal{O}(\sqrt{M\log d})$.
Then, we prove the approximation error in Eq.~\eqref{eq:spg_temp1} is at most $1/12$ in Euclidean distance error under the condition that the amplified error in Eq.~\eqref{eq:approxerr_arccos} is sufficiently small, which leads to the assumption in Eq.~\eqref{xabst:q_threshold} of this lemma.
At the same time, we can evaluate the normalization factor as $|\mathcal{N}_t-1|\leq 2\sqrt{6\delta'}$, thereby proving the post-selection probability $\mathcal{N}_t^2/2>0.462$.

Finally, we mention the space and gate complexity of this Lemma.
As well as in the proof of Lemma~\ref{thm:sp_HS}, we can count the number of qubits in the block encoding of $\mathbf{H}_{\rm G}$, and it scales as $pM+\log M+\log d$ in total.
Thus, the method has the space complexity $\mathcal{O}(M+\log_2 d)$ that is also independent of the target precision $\varepsilon$.
As for the gate complexity, the main contribution comes from the eigenvalue transformation $\mathbf{H}_{\rm G}\mapsto T_{t}[\mathbf{H}_{\rm G}]$, which uses $t=\mathcal{O}(2^q\sigma')$ queries of the block encoding of $\mathbf{H}_{\rm G}$.
This indicates that the gate complexity of the method is the same as that of the method in Lemma~\ref{thm:sp_HS}.
\end{proof}

In this method, we require finding $\tilde{\mathcal{O}}(q\sqrt{M})$ circuit parameters, which is due to the singular value amplification to construct the Hamiltonian Eq.~\eqref{eq:Grover_sp_Hamiltonian}, while the Grover-like repetition itself does not require any phase angle computation.
Thus, the alternative method shows an exponential improvement of runtime in classical computation with respect to the iteration $q$ (or $\varepsilon$), compared to the method in Lemma~\ref{thm:sp_HS} and the previous method~\cite{van2023quantum}, if certain conditions on the iteration step $q$ are satisfied.
Recalling that the iteration step is upper bounded by $q_{\rm max}:=\lceil \log_2(1/\varepsilon)\rceil$ for the target root MSE $\varepsilon$, the alternative method can be used if there exists the iteration steps $q$ satisfying $q\leq {q_{\rm max}}$ and the inequality Eq.~\eqref{xabst:q_threshold}.
In Sec.~\ref{sec:numerical_exp}, we numerically investigate the range of iteration steps $q$ for the use of Lemma~\ref{thm:sp_Grover}, especially for the case of $M=\mathcal{O}((\log_2 d)^k)$ for some $k$.


Finally, we mention the success probability of Lemma~\ref{thm:sp_Grover}.
While we cannot prepare the probing state $\ket{\Upsilon(q)}$ deterministically in the alternative method, the success probability can be easily boosted by a constant number of repetitions of the circuit runs.
Therefore, the total query complexity remains unchanged when we use Lemma~\ref{thm:sp_Grover} instead of Lemma~\ref{thm:sp_HS} to prepare the probing state of Eq.~\eqref{eq:targetstate}.

\subsubsection{Applicability condition for the probing-state preparation by Grover-like repetition}\label{sec:numerical_exp}
Here, we investigate the threshold of iteration step $q$ in Lemma~\ref{thm:sp_Grover}.
We rewrite the threshold Eq.~\eqref{xabst:q_threshold} as $q^*:=\log_2(1/\varepsilon^*)$ by defining $\varepsilon^*$ as
\begin{align}\label{eq:def_vareps_star}
    1/\varepsilon^* &:= \sqrt{{C{\ln^{-3/2} (2d/\delta')}\left\lceil{\sqrt{2(M+1)\ln(2d/\delta')}}\right\rceil}},
\end{align}
where $C:=2^{3}\cdot 33^3/625$.
Since the iteration step $q$ runs from 0 to $q_{\rm max}:=\lceil \log_2(1/\varepsilon)\rceil$, we need that the threshold $1/\varepsilon^*$ is smaller than the inverse of the target root MSE $1/\varepsilon$ in order to use the state preparation by the Grover-like repetition.

\begin{figure}[tb]
    \centering
    \includegraphics[width=0.4\textwidth]{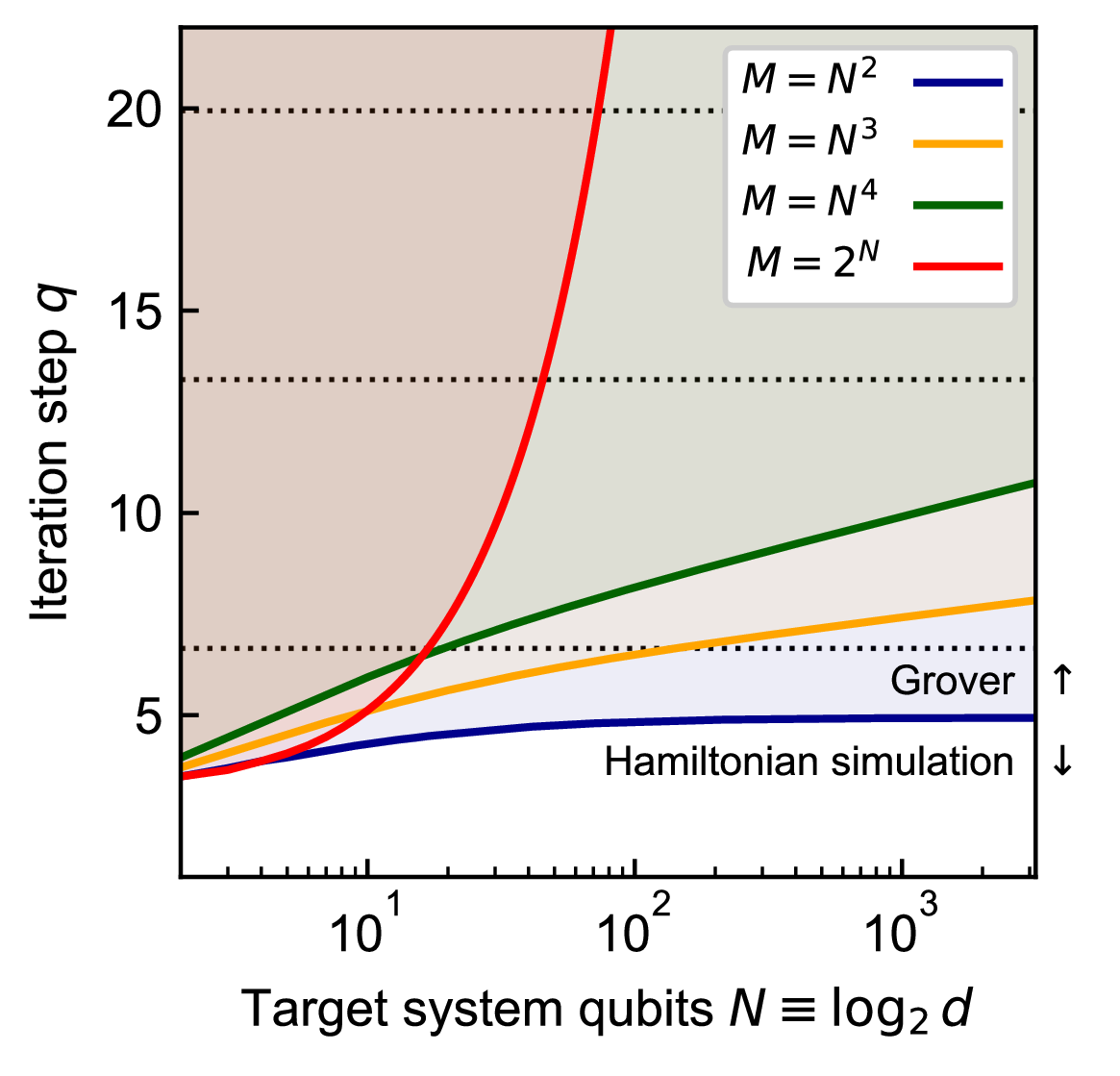}
    \caption{The condition for the Grover-based state preparation (Lemma~\ref{thm:sp_Grover}).
    The solid lines represent the threshold $q^*:=\log_2(1/\varepsilon^*)$ of iteration steps in Eq.~\eqref{eq:def_vareps_star} (or Eq.~\eqref{xabst:q_threshold}).
    When the iteration step $q$ exceeds the solid lines, we can use the Grover-like repetition to prepare the probing state Eq.~\eqref{eq:targetstate}, instead of using Hamiltonian simulation.
    The horizontal dotted lines represent $q_{\rm max}:=\lceil \log_2(1/\varepsilon)\rceil$ for $\varepsilon=10^{-2}$ (bottom), $10^{-4}$ (middle), and $10^{-6}$ (top), respectively.}
    \label{fig_main:threshold_Grover}
\end{figure}

To clarify this point, we here plot the threshold $1/\varepsilon^*$ for several cases $M=N^2,N^3,N^4,$ and $2^N$ under target system with the number of qubits $N\equiv \log_2 d$.
In Fig.~\ref{fig_main:threshold_Grover}, we confirm that the threshold $\log_2(1/\varepsilon^*)$ converges to some constant ($\sim 5$) in the case of $M=N^2$, which is consistent with Eq.~(\ref{eq:def_vareps_star}) because $(1/\varepsilon^*)^2=\mathcal{O}(\sqrt{M}/N)$ holds.
Thus, particularly in this case $M=\mathcal{O}(N^2)$, the method in Lemma~\ref{thm:sp_Grover} is available in a wide range of target precision $\varepsilon$ regardless of $N$.
On the other hand, if the number of observables scales as $2^N$, the threshold $\log_2(1/\varepsilon^*)$ increases linearly regarding $N$ (note that the horizontal axis in Fig.~\ref{fig_main:threshold_Grover} is logarithmic).
This means that in the case of $M=\mathcal{O}(2^N)$, the Grover-based method is applicable only for small-size systems, otherwise the target precision $\varepsilon$ is exponentially small with respect to $N$ (the base of the exponential is $2^{-1/4}\approx 0.841$).

Now, we focus on the case that we require more precise estimates as the number of observables $M$ increases.
From the definition of $\varepsilon^*$, if the desired precision $\varepsilon$ satisfy $\varepsilon^2 \ll {N/(C\sqrt{M})}$, then there exists iteration steps $q$ for the Grover-based method.
Specifically, considering the desired precision is given by $\varepsilon=c_{\rm mse}/\sqrt{M}\in (0,1/\sqrt{M})$ for some $c_{\rm mse}\in (0,1)$, we can evaluate the difference between the upper bound $q_{\rm max}$ of iteration steps and the threshold $q^*$ as follows:
\begin{equation}
    q_{\rm max}-q^* \geq \Omega\left(\log\left(N\sqrt{M}\right)\right).
\end{equation}
Thus, the range of iteration step $q$ such that the Grover-based method is available enlarges in this case, as $N$ or $M$ increases.

\section{Numerical simulation}\label{sec:num_simulation}
To evaluate the performance of observables estimation methods based on quantum gradient estimation, we here consider the following simple problem.
Let $U_{\psi}$ be a 1-qubit unitary that prepares a target state $\ket{\psi}$.
As for target observables $\{O_j\}_{j=1}^M$ with the spectral norm $\|O_j\|\leq 1$, we put the following assumptions for the sake of efficient classical simulation: $[O_i,O_j]=0$ holds for any pair of $i,j$, and the target state $\ket{\psi}$ is a common eigenstate of the observables $\{O_j\}_{j=1}^M$, that is, 
\begin{equation}
    O_j\ket{\psi}=g_j\ket{\psi} 
\end{equation}
holds for $j=1,2,...,M$ and eigenvalues $g_j\in [-1,1]$.
The goal is to estimate the expectation values $\langle O_j\rangle:=\bra{\psi}O_j\ket{\psi}=g_j$ with a root MSE $\varepsilon$ for all $j=1,2,...,M$.
In the following, we show the number of ancilla qubits and total queries to $U_{\psi}$ and its inverse in our method and the previous method~\cite{PhysRevLett.129.240501}, for solving the above problem with $M=30$.
For more details on the simulation setup, please see Appendix~\ref{apdx:num_sim_details}.

\begin{figure}[tb]
    \centering
    \includegraphics[width=0.45\textwidth]{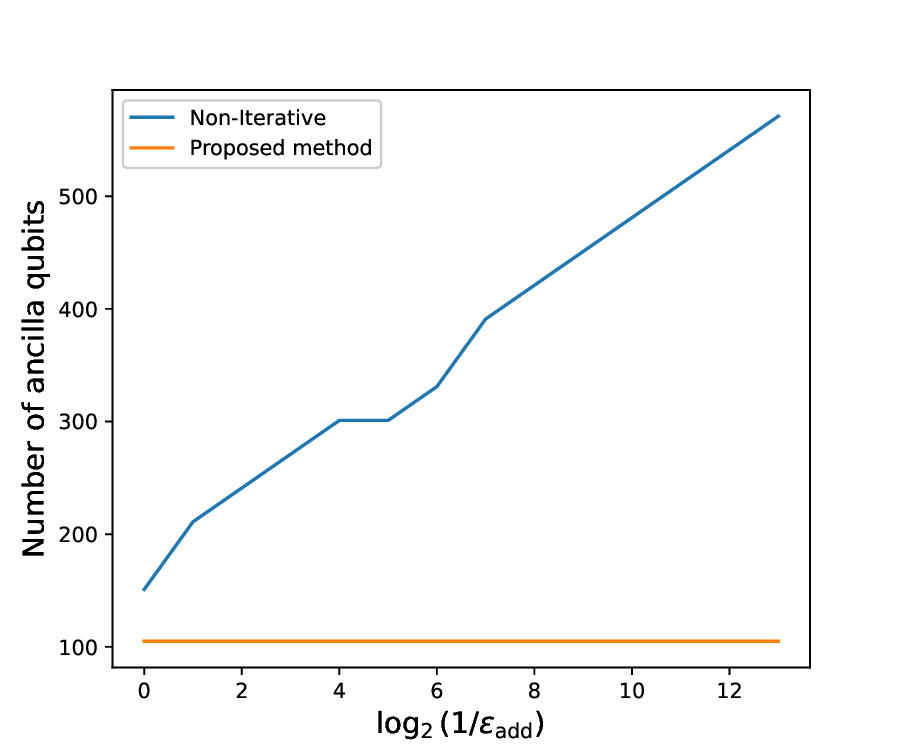}
    \caption{The number of additional ancilla qubits for $M=30$. 
    The ancilla count in the proposed adaptive method does not depend on target precision $\epsilon_{\rm add}$ (more precisely, target root MSE $\varepsilon$), unlike the previous one~\cite{PhysRevLett.129.240501}.}
    \label{fig:num_ancilla_qubits}
\end{figure}

\begin{figure}[tb]
    \centering
    \includegraphics[width=0.5\textwidth]{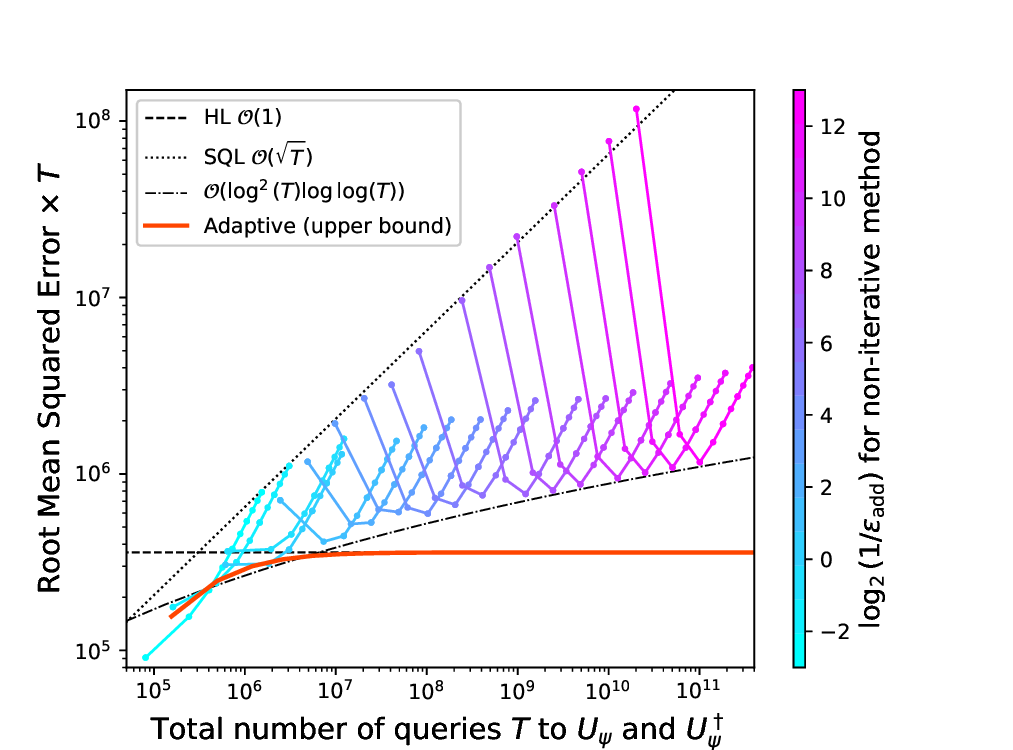}
    \caption{The comparison between the total number of queries and the root MSE.
    The orange line represents a theoretical upper bound in our method.
    The colored dots show the worst-case results of the previous non-iterative method~\cite{PhysRevLett.129.240501} among randomly-generated 26 sets of $\{g_j\}$, with various input pairs $(\varepsilon_{\rm add}, \delta)$.
    For each color, we fix $\varepsilon_{\rm add}$ as in the color bar and vary the failure probability $\delta=2^{-j}$ $(j=0,1,...)$, which gives the number of samples $N_{\rm med}=2\lceil\log(1/\delta)\rceil+1$ for a final median calculation.}
    \label{fig:rmse_vs_T_graph_rescaled}
\end{figure}

Figure~\ref{fig:num_ancilla_qubits} shows the number of additional ancilla qubits (except for target system qubits) in our method and the previous method~\cite{PhysRevLett.129.240501}.
We numerically calculate the number of ancilla qubits for the previous method, which is well approximated by 
\begin{equation}
    \lceil{\log_2(24/\varepsilon_{\rm add})}\rceil\times M.
\end{equation}
As for our method, the total number of qubits is explicitly calculated as
\begin{equation}
    3M+\lceil\log_2 M\rceil + \log_2 d + a + 9
\end{equation}
when we use the method based on Lemma~\ref{thm:sp_HS}; the multiplicative constant hidden in $\mathcal{O}(M)$ is exactly given by 3.
Note that in the figure, we ignore the factor $a$ for a block encoding of a single observable in our method and ancilla qubits that might be required for implementing Hamiltonian simulation for given observables in the previous method.
This figure clearly shows that our adaptive method has a constant ancilla overhead with respect to the target precision, unlike the non-iterative method.

Next, we show the relationship between the total queries to state preparation $U_\psi$ (and its inverse) and the resulting root MSE in Fig.~\ref{fig:rmse_vs_T_graph_rescaled}.
The number of total queries in our method is based on an upper bound of actually required queries, which is derived in the proof of Theorem~\ref{thm:main_query_complexity}.
As for the previous method which takes an additive error $\varepsilon_{\rm add}$ and a failure probability $\delta$ as input, we evaluate the resulting total queries and root MSEs by changing input pairs $(\varepsilon_{\rm add},\delta)$.
More precisely, we plot the dots with the same colors for a fixed $\varepsilon_{\rm add}=2^{-a}$ ($a=-3,-2,...,13$) and various $\delta$, by estimating the root MSE of a single observable $\braket{O_1}$ and calculating the corresponding total queries {$\tilde{T}_{\rm NonIter}$}.

We remark that $\tilde{T}_{\rm NonIter}$ in the previous method represents the actual number of queries divided by a factor ($>1$) that rapidly decreases and converges to a constant as $1/\varepsilon_{\rm add}$ increases, in order to highlight the asymptotic behavior.
When including the asymptotically negligible factor ($>1$), all the colored dots shift above our horizontal line that represents a theoretical upper bound; see Fig.~\ref{fig:rmse_vs_T_graph_norescaled} in Appendix~\ref{apdx:num_sim_details} for the worst-case and average-case results without any rescaling.

Our method achieves the Heisenberg-limited scaling in MSE; the line for our method is horizontal i.e., $\varepsilon T=\mathcal{O}(1)$ in Fig.~\ref{fig:rmse_vs_T_graph_rescaled}.
On the other hand, in the previous method, the sequence of minimal values of the same-colored dots shows the (asymptotic) scaling of $\Omega(\log^2 T\log\log T)$ for the total queries $T$.
This indicates that the Heisenberg-limited scaling with a logarithmic overhead is achieved by carefully choosing $\delta$ or $N_{\rm med}$ for a final median calculation. 
When we take a single shot outcome from a circuit as a final estimate ($N_{\rm med}=1$), the previous method has the scaling of SQL.

We remark that the previous method has $\mathcal{O}(\log(M/\delta))$ factor in the total query because it takes the median of independent $N_{\rm med}=\mathcal{O}(\log(M/\delta))$ samples as a final estimate.
Naively, in order to achieve MSE $\varepsilon^2$, we can take $\delta=\mathcal{O}(\varepsilon^2)$~\cite{lin2022heisenberg}.
Such a $N_{\rm med}=\mathcal{O}(\log(1/\varepsilon))$ (for a fixed $M$) factor might not be tight.
However, Fig.~\ref{fig:rmse_vs_T_graph_rescaled} also indicates that the logarithmic correction $N_{\rm med}=\mathcal{O}(\log(1/\varepsilon))$ is required, as follows.
Figure~\ref{fig:rmse_vs_T_graph_rescaled} indicates the following lower bound for the (rescaled) total number of queries: $\varepsilon \tilde{T}_{\rm NonIter}$ = $\Omega( \log^2(\tilde{T}_{\rm NonIter})\log\log(\tilde{T}_{\rm NonIter}))$, or 
$$\tilde{T}_{\rm NonIter} = \Omega({\varepsilon}^{-1} \log^2({1}/{\varepsilon})\log\log ({1}/{\varepsilon})).$$
Therefore, by comparing an upper bound $\tilde{T}_{\rm NonIter}=N_{\rm med}\times \mathcal{O}(\varepsilon_{\rm add}^{-1}\log(1/\varepsilon_{\rm add})\log\log(1/\varepsilon_{\rm add}))$~\cite{gilyen2019optimizing} with the above empirical lower bound, we consider $N_{\rm med}=\mathcal{O}(\log(1/\varepsilon))$ to be tight for a given root MSE $\varepsilon$.

Finally, these two algorithms cannot work in the query regime $<10^5$ under the current problem setup.
(Actually, the query limit is more severe in the previous method; see Fig.~\ref{fig:rmse_vs_T_graph_norescaled}.)
This is attributed to a large constant factor in the total number of queries; even if we want to obtain estimates within $\sim 1$ error, these algorithms require a much larger queries to run.

\section{Conclusion and discussion}\label{sec:conclusion}
In this work, we proved that the expectation values of multiple observables can be simultaneously estimated with quantum resources at the scaling of Heisenberg limit $1/\varepsilon$ for a root MSE $\varepsilon$.
At the same time, the total resource shows a nearly quadratic improvement with respect to the number of observables $M$, compared to the standard method for this task i.e., the (modified) quantum amplitude estimation~\cite{brassard2002quantum,knill2007optimal,rall2020quantum,suzuki2020amplitude,ZapataQAE2021,wada2022quantum}.
The resources are quantified by the total number of queries to a state preparation unitary $U_{\psi}$ whose complexity usually scales as the size of quantum system.
We prove these results by explicitly constructing an adaptive quantum algorithm.
The key idea of the proposed method is to prepare a quantum state with phases that encode the expectation values of observables, followed by the measurement for quantum gradient estimation~\cite{jordangradest, gilyen2019optimizing}.
Then, our method determines a single binary digit of the target observables from the measurement outcomes and update the quantum state so that the next measurement has sufficient resolution to determine the next fraction bit.
Importantly, our method can be considered as an extension of the iterative or adaptive phase estimation algorithms~\cite{kitaev1995quantum,higgins2007entanglement,higgins2009demonstrating,kimmel2015robust,dutkiewicz2022heisenberg} to the gradient estimation algorithm.

In addition to the Heisenebrg-limited scaling in MSE, the proposed method significantly reduces the requirement of actual implementation, compared to the state-of-the-art algorithms for multiple observables estimation~\cite{PhysRevLett.129.240501,van2023quantum}.
First, the adaptive nature of the proposed method allows us to reduce an additional space overhead to $\mathcal{O}(M)$ qubits, compared with $\mathcal{O}(M\log(1/\varepsilon))$ qubits in the previous methods.
This results in a significant improvement on the space overhead when precise estimates of observables are required.
Also, the proposed method can be executed in a parallel way during each iteration step, leading to reduction of the total execution time if we can use several quantum computers.

Next, in constructing quantum circuits of the proposed method, we provide two methods along with their quantum circuit diagrams; one is based on the optimal Hamiltonian simulation protocol and the other is based on Grover-like repetition.
While the former can be used in any iteration steps $q$, in the final step, 
it requires classical finding of $\mathcal{\tilde{O}}(\sqrt{M}/\varepsilon)$ circuit parameters for quantum signal processing (QSP)~\cite{PhysRevX.6.041067,Low2019hamiltonian}, as well as the previous methods.
As shown in the previous works~\cite{Haah2019product,chao2020finding,dong2021efficient,Ying2022stablefactorization,PhysRevResearch.6.L012007, yamamoto2024robust}, such a circuit parameter finding requires $\sim M/\varepsilon^2$ [sec] in classical computation time; the resulting numerical instability is one of the central problems in practical application of QSP or its extension, quantum singular value transformation (QSVT)~\cite{gilyen2019quantum}.
The alternative method based on Grover-like repetition can avoid this problem by partially using the QSP for Chebyshev polynomials whose circuit parameters are analytically derived.
As a result, the Grover-based method requires to tune only ${\mathcal{O}}(\sqrt{M}\log (M/\varepsilon))$ circuit parameters, under a specific condition on $q$.
This shows an exponential improvement in classical computation with respect to the target root MSE $\varepsilon$, thereby reducing the barrier of actual implementation significantly.
As for the condition of the alternative method, we numerically investigate it in various setups such as $M=\mathcal{O}(N^k)$ for some $k$ and the number of target system qubits $N$.

We leave several questions to be further investigated.
For the non-asymptotic advantage regarding $M$ of the proposed algorithm over the quantum amplitude estimation, it is crucial to minimize the constant factor that is hidden in the total queries $\mathcal{O}(\varepsilon^{-1}\sqrt{M}\log M)$ to the state preparation.
Specifically, the advantage of the proposed method over the individual applications of the (Heisenberg-limited) amplitude estimation, whose total queries are $\mathcal{O}(M/\varepsilon)$, is valid, 
when the constant factor $\mathcal{C}$ (divided by the constant factor of the amplitude estimation) of the proposed method is smaller than $\sqrt{M}/\log M$.
Therefore, it remains open to tightly evaluate $\mathcal{C}$, or equivalently  the crossover regarding $M$, and to optimize the proposed method in order to minimize its constant factor.
Also, the gate complexity for the other components such as block-encoded observables would be improved when we use initial knowledge of the target state and observables.

From the aspect of the proposed method being an adaptive version of quantum gradient estimation, it is valuable to discuss the application of our method for the tasks in which the original gradient estimation algorithm offers quantum speedup.
This certainly reduces the hardness of their actual implementation on early fault-tolerant quantum computers. 
Furthermore, exploring the speedup regarding the number of target parameters of the proposed method in the field of multiparameter quantum metrology is another interesting direction for future work.


\section*{Acknowledgements}
We thank fruitful discussions with Arjan Cornelissen, Yuki Koizumi, and Soichiro Morisaki.
K.W. was supported by JSPS KAKENHI Grant Number JP 24KJ1963 and JST SPRING, Grant Number JPMJSP2123.
N. Yamamoto was supported by MEXT Quantum Leap Flagship Program Grant Number JPMXS0118067285 and JPMXS0120319794, and JSPS KAKENHI Grant Number 20H05966. 
N. Yoshioka wishes to thank by 
JST CREST Grant Number JPMJCR23I4,  
JST PRESTO No. JPMJPR2119, 
JST ERATO Grant Number  JPMJER2302,  
JST Grant Number JPMJPF2221, 
and the support from IBM Quantum. 

\bibliography{ref}

\appendix
\newpage

\section{Quantum gradient estimation}\label{supple_sec:qgradest}
\subsection{Review of Jordan's algorithm}
For a given blackbox of a real scalar function $f(\bm{x})$ on $\bm{x}\in\mathbb{R}^{M}$, the quantum algorithm introduced by Stephen P. Jordan~\cite{jordangradest} can efficiently estimate the $M$-dimensional gradient of $\nabla f(\bm{0})$, with use of less queries to the blackbox compared to classical case.
We here review this algorithm based on its analysis by Ref.~\cite{gilyen2019optimizing}.
Here, the target point $\bm{x}=\bm{0}$ can be taken as $\bm{x}\neq \bm{0}$ by trivially redefining $f(\bm{x})$.
The quantum algorithm consists of three steps: (i) prepare a superposition state of grid points $\bm{x}$ around the target point $\bm{x}=\bm{0}$, (ii) apply the blackbox of the target function $f$ to the state and evaluate a phase $e^{if(\bm{x})}$ at each grid point, and (iii) measure the resulting state by the computational basis labeled by the grid points after the inverse quantum Fourier transformations.

To begin with, we define a set of grid points $G_{p}^M$ around $\bm{x}=\bm{0}$ to evaluate $f$, as follows:
\begin{equation}\label{supple_eq:G_p}
    G_{p}^M:=\left\{\frac{\mu}{2^p}-\frac12 +\frac{1}{2^{p+1}}:\mu\in\{0,1,\cdots,2^p-1\}\right\}^M,
\end{equation}
where $p$ denotes a positive integer which specifies the estimation precision later.
Note that because there is a bijection map
\begin{equation}\label{eq:varphi_bijection}
    \varphi:\mu\mapsto \varphi(\mu):=\frac{\mu}{2^p}-\frac12 +\frac{1}{2^{p+1}}\in G_p,
\end{equation}
we always label the $p$-qubit computational basis $\ket{\mu}$ by the corresponding element $x\equiv\phi(\mu)\in G_p$.
Then, applying the Hadamard gates to initialized $pM$-qubit registers, we have the superposition state of the grid points $\bm{x}:=(x_1,...,x_M)\in G^M_{p}$:
\begin{align}\label{supple_eq:suppos}
    &\frac{1}{\sqrt{2^{pM}}}\sum_{\bm{x}\in G_p^M} \ket{\bm{x}}\notag\\
    &~~~=\frac{1}{\sqrt{2^{pM}}}\sum_{(x_1,...,x_M)\in G_p^M} \ket{x_1}\ket{x_2}\cdots \ket{x_M}.
\end{align}
Here, each $\ket{x_j}$ contains $p$-qubit registers.
Next, we assume access to the \textit{phase oracle} $O_{f}$ for $f(\bm{x})$ 
defined as~\cite{gilyen2019optimizing}
\begin{equation}
    O_f:\ket{\bm{x}}\to e^{if(\bm{x})}\ket{\bm{x}}~~\mbox{for~all}~~\bm{x}\in G_p^M.
\end{equation}
Note that Ref.~\cite{gilyen2019optimizing} provides a generic analysis of Jordan's algorithm based on the phase (and probability) oracle, while the original paper~\cite{jordangradest} assumes another powerful oracle (i.e., $\eta$-accurate Binary oracle for $f$ that outputs $f(\bm{x})$ binarily with accuracy $\eta$ from an input $\bm{x}$).

Now, to clarify the basic idea of Jordan's algorithm, we consider a special case where the target function $f$ is an affine linear function, i.e., for a target gradient vector $\bm{g}\in [-1/3,1/3]^M$ and some constant $f_{0}\in \mathbb{R}$, $f(\bm{x})=\bm{g}\cdot \bm{x}+f_{0}$.
For such linear functions, the application of the (modified) phase oracle $O_f^{2\pi 2^{p}}$ to the input state Eq.~(\ref{supple_eq:suppos}) yields
\begin{align}\label{supple_eq:product_QPE}
    &\frac{1}{\sqrt{2^{pM}}}\sum_{\bm{x}\in G_p^M} e^{2\pi i 2^p f(\bm{x})}\ket{\bm{x}}\notag\\
    &~~~=e^{2\pi i 2^p f_0}\cdot\bigotimes_{j=1}^M\left[\frac{1}{\sqrt{2^{p}}}\sum_{x_j\in G_p} e^{2\pi i2^p x_j g_j }\ket{x_j}\right].
\end{align}
Then, we apply 
a slightly modified version of the inverse quantum Fourier transformation ${\rm QFT}^\dagger_{G_p}$ over the $p$-qubit system: for all $x\in G_p$, 
\begin{equation}
    {\rm QFT}_{G_p}:\ket{x}\mapsto \frac{1}{\sqrt{2^p}}\sum_{k\in G_p} e^{2\pi i 2^p xk}\ket{k}.
\end{equation}
Note that ${\rm QFT}_{G_p}$ is the same as the usual $p$-qubit QFT up to conjugation with a tensor product of $p$ single-qubit gates~\cite{gilyen2019optimizing}.
The statistics of the computational basis measurement to Eq.~(\ref{supple_eq:product_QPE}) with ${\rm QFT}^\dagger_{G_p}$ are similar to that of the output of standard quantum phase estimation algorithm~\cite{nielsen2010quantum}.
More precisely, let $(k_1,k_2,...,k_M)\in G_p^M$ be a result of the computational basis measurement, then the following holds~\cite{gilyen2019optimizing}:
\begin{equation}\label{supple_eq:ideal_gradest_stat}
    {\rm Pr}\left[\left|k_j-g_j\right|>\frac{3}{2^p}\right]\leq \frac{1}{4}~~\mbox{for~every}~~i=1,2,\cdots,M.
\end{equation}

In the above description, we assumed that the target function is affine linear. 
Moreover, if a target function is very close to some affine linear function and thereby the equality of Eq.~(\ref{supple_eq:product_QPE}) approximately holds, we can prove similar results to Eq.~(\ref{supple_eq:ideal_gradest_stat}), as follows.
\begin{lem}
[\cite{gilyen2019optimizing,van2023quantum}]\label{supple_lem:original_gradest}
    Let $\bm{g}\in \mathbb{R}^M$ such that $\|\bm{g}\|_{\infty}\leq 1/3$. 
    Suppose we can prepare the quantum state $\ket{\Psi}$ that is $1/12$-close in the Euclidean distance to the following state
    \begin{equation}\label{supple_eq:grad_est_approx_target}
        \left({\rm QFT}^\dagger_{G_p}\right)^{\otimes M}\frac{1}{\sqrt{2^{pM}}}\sum_{\bm{x}\in G_p^M} e^{2\pi i 2^p \bm{g}\cdot\bm{x}}\ket{\bm{x}}.
    \end{equation}
    Then, measuring the quantum state $\ket{\Psi}$ in the computational basis, we obtain a coordinate-wise estimate $(k_1,...,k_M)\in G_{p}^M$ satisfying
    $$
    {\rm Pr}\left[\left|k_j-g_j\right|>\frac{3}{2^p}\right]\leq \frac{1}{3}~~\mbox{for~every}~~j=1,2,\cdots,M.
    $$
\end{lem}
\begin{proof}
    By the closeness assumption, the trace distance between $\ket{\Psi}$ and Eq.~(\ref{supple_eq:grad_est_approx_target}) is upper bounded by 1/12.
    Since the trace distance provides the upper bound of the total variation distance of two probability measures defined by the two quantum states and an arbitrary common POVM (Theorem 9.1~\cite{nielsen2010quantum}), the probability~\eqref{supple_eq:ideal_gradest_stat} is modified at most 1/12 by measuring $\ket{\Psi}$ instead of Eq.~(\ref{supple_eq:grad_est_approx_target}).
\end{proof}
\noindent
For more regular functions, Ref.~\cite{gilyen2019optimizing} provides an improved version of Jordan's algorithm, using the higher-order finite-difference formulas to enhance the linearity of the target functions in a certain domain.

\subsection{Heisenberg limit}
The phase oracle $O_f$ for a smooth real function $f(\bm{x})$ on $\bm{x}\in \mathbb{R}^{M}$ with gradient $\bm{g}:=\nabla f(\bm{0})$ corresponds to the time evolution operator (with unit time) generated by the following Hamiltonian
\begin{equation}
    H(\bm{g})+V,~~H(\bm{g}):=\sum_{\bm{x}\in G_p^M}\left(\bm{g}\cdot\bm{x}\right)\ket{\bm{x}}\bra{\bm{x}},
\end{equation}
where $V$ is a $\bm{g}$-independent Hermitian operator that commutes with $H(\bm{g})$.
Thus, the estimation of $\bm{g}$ from $O_f$ with appropriate choice of an input state and a final measurement can be considered as one of the (multi-parameter) unitary estimation problems.
Such unitary estimation problems are widely studied in the field of quantum metrology, and one of the major topics in this field is to devise quantum estimation protocols that can achieve the higher precision than any classical protocol~\cite{Qmetrology2006,giovannetti2011advances}.
In particular, the Heisenberg limit provides a fundamental bound on the estimation precision (in terms of root mean squared error) under given resources, and it is typically given as $1/t$ where $t$ denotes the total number of resources to be used~\cite{giovannetti2011advances, PicorrectedHL2020}.
Here, we derive the Heisenberg limit in quantum gradient estimation.

To derive the fundamental bound, we here consider a general adaptive estimation protocol~\cite{Qmetrology2006}, which is illustrated in Fig.~\ref{supple_fig:adaptive_protocol}.
In the general protocol, we first prepare an arbitrary input state $\rho_{\rm in}$ between the system on which $O_f$ acts and an ancilla system with an arbitrary number of qubits.
Then, we apply the following sequence of quantum operations to the input state:
\begin{align}
    U_t(O_f\otimes \bm{1}) \cdots U_2(O_f\otimes \bm{1})U_1(O_f\otimes \bm{1}),
\end{align}
where $\bm{1}$ denotes the identity and $U_i~(i=1,2,...,t)$ are $\bm{g}$-independent arbitrary unitary operators acting on the whole system.
This sequence contains $t$ uses of the phase oracle in total.
Finally, the output state is measured by an arbitrary POVM $\{M_{\bm{g}}\}$, and we write the corresponding single-shot estimator for ${g}_j$ as $\hat{g}_j$.
Because the interaction $U_i$ to the ancilla systems can extract the information during the estimation process, this protocol also includes adaptive techniques.
Obviously, the original method for gradient estimation is captured in this general protocol, by setting $\rho_{\rm in}$ to the uniform superposition state without ancilla qubits, $U_i=\bm{1}$ for all $i$, and $M_{\bm{g}}$ to the computational basis measurement with the inverse quantum Fourier transformation.
Note that the $\bm{g}$-independent part of $O_f$ can be included in each $U_j$, and therefore in the following, we consider $O_f=e^{iH(\bm{g})}$ without loss of generality.

\begin{figure*}[tb]
    \centering
    \includegraphics[width=0.65\textwidth]{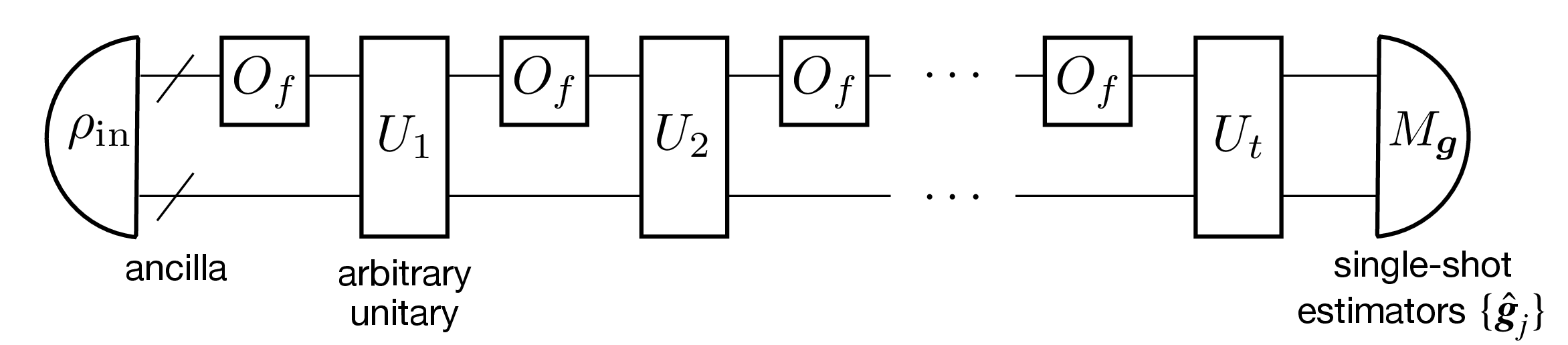}
    \caption{A general adaptive protocol for estimating $\bm{g}$ in $O_f$.}
    \label{supple_fig:adaptive_protocol}
\end{figure*}

We here remark that
$e^{iH(\bm{g})}$ can be written as a tensor product of time evolutions 
\begin{equation}\label{supple_eq:prod_time_evolutions}
    \bigotimes_{j=1}^M \left(\sum_{x_j\in G_p}e^{ ig_jx_j}\ket{x_j}\bra{x_j}\right)
\end{equation} 
by the generators $H_j:=\sum_{x_j}x_j\ket{x_j}\bra{x_j}$ with the bounded spectral norm $\|H_j\|\leq 1/2$.
Therefore, the estimation of $\bm{g}=(g_1,...,g_M)$ is essentially equal to the estimation of evolving time $g_j$ in each $e^{ig_j H_j}$.
Now, we are ready to prove the Heisenberg limit in quantum gradient estimation.
\begin{thm}[Heisenberg limit in quantum gradient estimation]\label{supple_thm:HL_in_GradEst}
    Suppose we have access to the phase oracle $O_f$ for a smooth real function $f$ on $\mathbb{R}^M$ with gradient $\bm{g}:=\nabla f(\bm{0})$ and $g_j$ belongs to some finite interval $\Theta \subset\mathbb{R}$ for all $j$.
    Then, the single-shot estimators $\hat{g}_j$, obtained from the general adaptive estimation protocol in Fig.~\ref{supple_fig:adaptive_protocol} with $t$ uses of the phase oracle $O_f$, satisfy the following inequality:
    \begin{equation}
        \min_{j=1,2,\cdots,M}~\widehat{{\rm MSE}}\left[\hat{g}_j\right] \geq \frac{\pi^2}{t^2}~~\mbox{as}~~t\to \infty,
    \end{equation}
    where $\widehat{{\rm MSE}}[\hat{g}_j]$ denotes the supremum of the mean squared error (MSE) for $\hat{g_j}$ over the known interval $\Theta$, i.e., the maximum value of MSE over all possible target values $g_j\in \Theta$.
\end{thm}

\begin{proof}
    We prove this theorem by contradiction.
    Suppose there is a protocol such that for an index $j\in \{1,...,M\}$, $\widehat{{\rm MSE}}\left[\hat{g}_j\right]<{\pi^2}/(t^2\Delta G^2_p)$ holds as $t$ goes to $\infty$, where $\Delta G_p$ denotes the difference between extreme eigenvalues of $ H_j=\sum_{x_j\in G_p} x_j\ket{x_j}\bra{x_j}$ i.e., $\max{G_p}-\min{G_p}\leq 1$.
    Now, we recall that the phase oracle $O_f$ can be considered as the product of $e^{ig_jH_j}$ without loss of generality.
    Then, focusing on the estimation of $g_j$ in the single-parameter unitary $e^{ig_jH_j}$, it can be shown that there is no protocol using the total $t$ uses of $e^{ig_jH_j}$ such that $\widehat{{\rm MSE}}[\hat{g}_j]$ can be less than ${\pi^2}/{t^2\Delta G^2_p}$ as $t$ increases~\cite{PicorrectedHL2020,gorecki2023heisenberg}, which contradicts the above assumption.
    Therefore, we conclude that
    $$
    \widehat{{\rm MSE}}[\hat{g}_j]\geq \frac{\pi^2}{t^2\Delta G_p^2}\geq  \frac{\pi^2}{t^2}~~\mbox{as}~~t\to \infty$$ 
    holds for all $j$, which completes the proof of Theorem~\ref{supple_thm:HL_in_GradEst}.    
\end{proof}

As for the achievability of the lower bound, the same methodology as the phase estimation with the minimum phase uncertainty~\cite{PhysRevA.54.4564,PhysRevLett.98.090501,PhysRevA.90.062313,PhysRevX.8.041015} can be applied to design an optimal protocol, because the gradient estimation can be considered as separable applications of the phase estimation protocol.
Hereafter, we describe an optimal protocol of gradient estimation that is applicable to an affine linear function $f(\bm{x})=\bm{x}\cdot \bm{g}$ without any adaptive operations.
Let us consider $M$ tensor products of an ansatz state $\sum_{x\in G_p} a_x \ket{x}$, where $a_x$ is a real amplitude, as an initial input state instead of the equal superposition state Eq.~(\ref{supple_eq:suppos}).
Then, applying the phase oracle $O_f$ $2^{p+1}$ times to the input state and performing the computational basis measurement after ${\rm QFT}_{G_p}^{\dagger}$s, we obtain the measurement result $k_j$ with the probability
\begin{equation}
    {\rm Pr}\left[k_j\right] = \left|\sum_{x_j=0}^{2^p-1} \frac{a_{x_j}}{\sqrt{2^p}} e^{2\pi i 2^p x_j(g'_j-k_j)}\right|^2,
\end{equation}
where $g_j':=g_j/\pi$.
For simplicity, we drop the subscript $j$ in the following.
Taking the measurement outcome $k$ as an estimator of $g'$, we can approximately evaluate the MSE of $k$ as follows:
\begin{equation}\label{supple_eq:collet_var}
    \mathbb{E}\left[(k-g')^2\right] \simeq \frac{1}{2\pi^2}\left(1-\mathbb{E}\left[\cos{\left(2\pi(k-g')\right)}\right]\right),
\end{equation}
where we can check the right hand side is close to the MSE when $(k-g')$ is small by the Taylor expansion.
As shown in Ref.~\cite{Ji2008estqchannel}, the expectation of cosine can be analytically calculated:
\begin{align}
    &\mathbb{E}\left[\cos{\left(2\pi(k-g')\right)}\right]\notag\\
    &= \sum_{k=1}^{2^p-1} a_{k-1}a_k+a_0a_{2^p-1}\cos{\left[2\pi 2^p \left(g'+\frac{1}{2}-\frac{1}{2^{p+1}}\right)\right]}.
\end{align}
Here, $a_k$ is equivalent to $a_{\varphi(k)}$ for the bijection map defined in Eq.~(\ref{eq:varphi_bijection}).
Therefore, if we take $a_0=0$, then the right hand side of Eq.~(\ref{supple_eq:collet_var}) can be written as a quadratic form $(1/2\pi^{2})\sum_{k,l=1}^{2^{p}-1}A_{kl}a_ka_l$ for the symmetric matrix $A$:
\begin{equation}
    A=\begin{pmatrix}
      1&-1/2&& &\\
      -1/2&1&-1/2&&\\
      &-1/2&1&\ddots&\\
      &&\ddots&\ddots&\\
    \end{pmatrix}.
\end{equation}
The minimum eigenvalue and the corresponding eigenvector of $A$ is known as 
\begin{equation}\label{supple_eq:lowest_eigVal}
    2\sin^2{{\frac{\pi}{2^{p+1}}}}
\end{equation}
and 
\begin{equation}\label{supple_eq:qpe_opt_input}
   a_k = \sqrt{\frac{2}{2^{p}}} \sin \frac{k\pi}{2^{p}},
\end{equation}
respectively.
Consequently, combining Eqs.~(\ref{supple_eq:collet_var}) and (\ref{supple_eq:lowest_eigVal}), we can confirm that this protocol achieves the lower bound of Theorem~\ref{supple_thm:HL_in_GradEst} when $t=2^{p+1}$ is sufficiently large, as
\begin{equation}
    \mathbb{E}\left[(\pi k-g)^2\right] \simeq \left(\frac{\pi}{2^{p+1}}\right)^2.
\end{equation}

Finally, we remark that the original implementation of the gradient estimation cannot achieve the quadratic speedup regarding total queries $t$, i.e., $\varepsilon^2 = \mathcal{O}(1/t^2)$ for a MSE $\varepsilon^2$, even in the target function is an affine linear.
This can be checked from Eq.~(\ref{supple_eq:collet_var}); the uniform superposition state $a_{x}=1/\sqrt{2^p}$ for all $x$ gives 
\begin{align}
    \min_{g'}\mathbb{E}\left[\cos{\left(2\pi(k-g')\right)}\right] 
    =1-\frac{2}{2^p},
\end{align}
and thus, $\max_{g} \mathbb{E}\left[(\pi k-g)^2\right]=\mathcal{O}(1/t)$, which is the same scaling as the shot noise or the standard quantum limit (SQL).
Note that this result is consistent with that the quantum phase estimation algorithm fails to achieve the Heisenberg limit when the input ancilla qubits are initialized as the uniform superposition state, unlike the optimal entangled state as in Eq.~(\ref{supple_eq:qpe_opt_input})~\cite{higgins2007entanglement,higgins2009demonstrating,PhysRevA.80.052114}.
We remark that the proposed adaptive method in this work performs gradient estimation at the same scaling as the Heisenberg limit $\mathcal{O}(1/\varepsilon)$ even when the target function is \textit{approximately} affine linear.

\section{Circuit implementation and theoretical guarantees of the proposed method}\label{supple_sec:proposed_alg}

In Sec.~\ref{subsec:sp4itergradest_rev} in the main text, we have provided two methods to prepare the probing state $\ket{\Upsilon(q)}$ that was defined in Eq.~\eqref{eq:targetstate} as
\begin{align}
    &\ket{\Upsilon(q)}:=\notag\\
    &~~~\frac{1}{\sqrt{2^{pM}}}\sum_{\bm{x}\in G_p^M} e^{2\pi i2^p \sum_{j=1}^M x_j {2^{q}\pi^{-1}\braket{{O}_j-\tilde{u}^{(q)}_j\bm{1}}}} \ket{\boldsymbol{x}}\notag.
\end{align}
In the following, we give more details regarding the state preparation in Appendix~\ref{supple_sec;sp4iterativegradest}, and then analyze the algorithm of the adaptive observables estimation using those state preparations in Appendix~\ref{supple_sec;complexity_of_alg}.
The quantum circuits are provided in Appendix~\ref{sec:list_circuits}.
In this appendix, we often refer to the detailed version Algorithm~\ref{supp_alg:main}, instead of Algorithm~\ref{alg:main} in Sec.~\ref{sec:main_III} in the main text.

\renewcommand{\baselinestretch}{1.2}
\setcounter{algorithm}{0}
\renewcommand{\thealgorithm}{1*}
\begin{figure}[htb]
\begin{algorithm}[H]
    \caption{Adaptive gradient estimation\\for multiple observables (A detailed version.)}\label{supp_alg:main}
    \begin{algorithmic}[1]
    \smallskip
    \REQUIRE $\log_2{d}$-qubit state preparation unitary $U_{\psi}$ and $U_{\psi}^\dagger$; observables $\{O_j\}_{j=1}^M$ with the spectral norm $\|O_j\|\leq 1$ such that \begin{equation}\label{supple_eq:num_obs_condition}
        M > 2\ln d + 24=\mathcal{O}(\log d)
    \end{equation}
    holds; confidence parameter $c\in (0,3/8(1+\pi)^2]$; target precision parameter $\varepsilon\in (0,1)$.

    \smallskip
    \ENSURE A sample $(\tilde{u}_1,...,\tilde{u}_M)$ from an estimator $\hat{\boldsymbol{u}}=(\hat{u}_1,...,\hat{u}_M)$ whose $j$-th element estimates $\braket{O_j}:=\langle\psi|O_j|\psi\rangle$ within MSE $\varepsilon^2$ as
    $$
    \max_{j=1,2,...,M}~\mathbb{E}\left[\left(\hat{u}_j-\braket{O_j}\right)^2\right]\leq \varepsilon^2
    $$

    \STATE Set $p:=3$ and $\tilde{u}_j^{(0)}:= 0,~(j=1,2,\cdots,M)$
    
    \FOR{$q=0,1,...,{q_{\rm max}}:=\lceil \log_2(1/\varepsilon)\rceil$}
    
    \STATE Set $$\tilde{O}^{(q)}_j:=\frac{O_j-\tilde{u}^{(q)}_j\bm{1}}{2}~~~\mbox{and}~~~\delta^{(q)}:=\frac{c}{8^{{q_{\rm max}}-q}}$$
    

    \STATE Prepare $\mathcal{O}(\log(M/\delta^{(q)}))$ copies of a quantum state that is $1/12$-close (in Euclidean distance) to
    $$
    \left({\rm QFT}^\dagger_{G_p}\right)^{\otimes M}\ket{\Upsilon(q)},~\mbox{where}
    $$
    $$
    \ket{\Upsilon(q)}:=\frac{1}{\sqrt{2^{pM}}}\sum_{\bm{x}\in G_p^M} e^{2\pi i2^p \sum_{j=1}^M x_j {2^{q+1}\pi^{-1}\braket{\tilde{O}^{(q)}_j}}} \ket{\boldsymbol{x}},
    $$
    using Lemma~\ref{supple_lem:statepre_HS} or Lemma~\ref{supple_lem:statepre_Grover} and perform the computational basis measurement on each of them.
    Note that each measurement outputs a result in the form of $(x_1,...,x_M)\in G_p^M$.

    \STATE Set coordinate-wise medians of the measurement results as ${g}_j^{(q)}$

    \STATE Set $\tilde{u}_j^{(q+1)}:=\tilde{u}_j^{(q)}+{\pi}{2^{-q}}g_j^{(q)}$
    
    \IF{there are some $j$ such that $\tilde{u}_j^{(q+1)}\geq 1$ (or $\leq -1$)}
    \STATE Set $\tilde{u}_j^{(q+1)}=1$ (or $-1$) for such $j$
    \ENDIF
    
    \ENDFOR
    
    \RETURN $\tilde{u}_j:= \tilde{u}_j^{({q_{\rm max}}+1)}$
    \end{algorithmic}
\end{algorithm}
\end{figure}
\renewcommand{\baselinestretch}{1}

\subsection{Probing-state preparation}\label{supple_sec;sp4iterativegradest}
Here, we provide two quantum algorithms to prepare the probing state $\ket{\Upsilon(q)}$ in Eq.~\eqref{eq:targetstate}, using block encodings of observables $\{{O}_j\}$ and the state preparation oracle $U_{\psi}$.
First, we describe a method to construct the block encoding of a linear combination of observables.

\subsubsection{Amplified block encoding for observables}\label{supple_sec:amp_BE4H}

Let $B_j~(j=1,...,M)$ be an $a$-block-encoding of a $d\times d$ observable $O_j$, and let $\mathbf{p}:=(p_1,...,p_M)$ be a sequence of positive integers.
Here, we define a set $G_{\mathbf{p}}$ of grid points as 
$$
G_{\mathbf{p}}:=G_{p_1}\times \cdots\times G_{p_M},
$$
where each $G_p$ is defined as in Eq.~\eqref{supple_eq:G_p}.
For any element $\bm{x}=(x_1,...,x_M)\in G_{\mathbf{p}}$, we can construct an $(a+\lceil \log_2 M\rceil+1)$-block-encoding of $$\tilde{H}^{(\bm{x})}:=M^{-1}\sum_{j=1}^M x_j O_j$$ 
by using the LCU method.
The corresponding SELECT operation denoted by $U_{\rm SEL}$ is given by
$$
U_{\rm SEL}:=\sum_{j=1}^M \ket{j}\bra{j}\otimes B_j.
$$
As for the PREPARE operation, we consider implementation with and without reflecting the signs of coefficients as
$P^{(\bm{x})}_{\rm R}$ and  $P^{(\bm{x})}_{\rm L}$, defined as
\begin{equation}\label{supple_eq:BE_obs_preR}
    P^{(\bm{x})}_{\rm R}:\ket{0}^{\otimes \lceil \log_2 M\rceil} \mapsto \sum_{j=1}^M {\rm sgn}(x_j)\sqrt{\frac{|x_j|}{\|\bm{x}\|_1}}\ket{j}.
\end{equation}
\begin{equation}\label{supple_eq:BE_obs_preL}
    P^{(\bm{x})}_{\rm L}:\ket{0}^{\otimes \lceil \log_2 M\rceil} \mapsto \sum_{j=1}^M \sqrt{\frac{|x_j|}{\|\bm{x}\|_1}}\ket{j}
\end{equation}
Then, we can confirm that the unitary $(P^{(\bm{x})}_{\rm L}\otimes \bm{1})^\dagger U_{\rm SEL}(P^{(\bm{x})}_{\rm R}\otimes \bm{1})$ is a block encoding of $\propto \tilde{H}^{(\bm{x})}$ because for any $\ket{\psi}$, 
\begin{align}
    &(P^{(\bm{x})}_{\rm L}\otimes \bm{1})^\dagger U_{\rm SEL}(P^{(\bm{x})}_{\rm R}\otimes \bm{1})\ket{0}^{\otimes \lceil \log_2 M\rceil} \ket{\psi}\notag\\
    &~~~= \sum_{j=1}^M {\rm sgn}(x_j)\sqrt{\frac{|x_j|}{\|\bm{x}\|_1}}(P^{(\bm{x})}_{\rm L})^\dagger\ket{j}\otimes B_j\ket{\psi}\notag\\
    &~~~= \ket{0}^{\otimes \lceil \log_2 M\rceil}\otimes  \sum_{j=1}^M \frac{x_j}{\|\bm{x}\|_1}B_j\ket{\psi}+\ket{\tilde{\psi}^{\perp}},
\end{align}
where $\bra{0}^{\otimes \lceil \log_2 M\rceil}\ket{\tilde{\psi}^{\perp}}={0}$.
While the block encoding has the normalization factor $\|\bm{x}\|_1$ which depends on $\bm{x}$, it can be modified to $M$ by introducing an additional single ancilla qubit and a single rotation gate $R_y(\theta)=e^{i\theta Y}$.
As a result, the following unitary is $(a+\lceil \log_2 M\rceil+1)$-block-encoding of $\tilde{H}^{(\bm{x})}$:
\begin{align}\label{supple_eq:BE_pre_amp_obs}
    (I\otimes P^{(\bm{x})}_{\rm L}\otimes \bm{1})^\dagger \left(I\otimes U_{\rm SEL}\right)(R_{y}(\theta_{\bm{x}})\otimes P^{(\bm{x})}_{\rm R}\otimes \bm{1}),
\end{align}
where $I$ denotes the 1-qubit identity, and $R_y(\theta_{\bm{x}}):\ket{0}\mapsto \frac{\|\bm{x}\|_1}{M}\ket{0}+\sqrt{1-\left(\frac{\|\bm{x}\|_1}{M}\right)^2}\ket{1}$.

In this block encoding of $\tilde{H}^{(\bm{x})}$, only PREPARE operation depends on the grid point $\bm{x}\in G_{\mathbf{p}}$. 
Thus, by adding a control to $P_{\rm L}^{(\bm{x})}$, $P_{\rm R}^{(\bm{x})}$, and $R_y$, and then sequentially applying them to each $\bm{x}$, we can implement the following unitary     \begin{equation}\label{supple_eq:obs_BE_oracle}
    U'_{\rm SEL}:=\sum_{\bm{x}\in G_{\mathbf{p}}} \ket{\bm{x}}\bra{\bm{x}}\otimes U^{(\bm{x})},
\end{equation}
where $U^{(\bm{x})}$ denotes the block encoding of $\tilde{H}^{(\bm{x})}$.
Because the detailed circuit description of $U'_{\rm SEL}$ is not required in the following, we deal with the unitary $U'_{\rm SEL}$ as an oracle for observables $\{O_j\}_{j=1}^M$, instead of $U_{\rm SEL}$, for a simple expression of gate complexity.

\begin{rem}\label{rem:design_prepare}
    The above straightforward construction of $U'_{\rm SEL}$ results in a very large circuit with $\mathcal{O}(2^{\|\mathbf{p}\|_1})$-depth.
    Fortunately, we can construct a more efficient circuit by further modifying the PREPARE operation, using a similar way to Ref.~\cite{laneve2023robust}.
    As shown in the Sec.~3 of Ref.~\cite{gilyen2019optimizing}, we can implement $\sum_{x\in G_p}\ket{x}\bra{x}\otimes e^{-ixZ}$ with $\mathcal{O}(p)$-depth circuit using $p+1$ (controlled) single-qubit rotation gates.
    By sequentially applying the $\ket{j}$-controlled version of this gate (see Fig.~\ref{fig:block_enc_Hx}) over ${x}_j$-register, the additional single ancilla qubit, and the $\lceil \log_2 M\rceil$ qubits, we obtain a $\bm{x}$-controlled block encoding of $e^{i\mathcal{H}^{(\bm{x})}}$, where $\mathcal{H}^{(\bm{x})}:=\sum_{j=1}^M x_j\ket{j}\bra{j}$.
    The corresponding circuit diagram is shown in Fig.~\ref{fig:block_enc_Hx2}.
    Then, we can obtain an $\epsilon$-precise block encoding of $(2/\pi)\mathcal{H}^{(\bm{x})}$, via the eigenvalue transformation for the logarihm of unitaries: $e^{i\mathcal{H}^{(\bm{x})}}\mapsto \mathcal{H}^{(\bm{x})}$ (Corollary 71 in~\cite{gilyen2019quantum}).
    This implementation consists of $\mathcal{O}(\log(1/\epsilon))$ uses of block-encoding of $e^{i\mathcal{H}^{(\bm{x})}}$ or its inverse; {see Fig.~\ref{fig:BEofsinH}}.
    Now, introducing the SELECT $U_{\rm SEL}$ and Hadamard gates, we have an $\epsilon$-precise block-encoding of $(2M/(\pi2^{\lceil \log_2 M\rceil}))\tilde{H}^{(\bm{x})}$ controlled by 
    $\bm{x}$, using at most $\mathcal{O}(\|\mathbf{p}\|_1\log(1/\varepsilon))$-depth circuit over $\mathcal{O}(\|\mathbf{p}\|_1+a+\lceil \log_2 M\rceil +\log_2 d)$ qubits in total:
    \begin{align}
        &\bra{\bm{0}}H^{\otimes \lceil\log_2 M\rceil}\otimes \bm{1}\cdot U_{\rm SEL}\cdot \frac{2}{\pi}\mathcal{H}^{(\bm{x})} H^{\otimes \lceil\log_2 M\rceil}\ket{\bm{0}}\otimes \bm{1} \notag\\
        &=\frac{2M}{\pi2^{\lceil \log_2 M\rceil}}\frac{1}{M}\sum_{j=1}^M x_j B_j.
    \end{align}
    Note that this block encoding has a single use of $U_{\rm SEL}=\sum_{j}\ket{j}\bra{j}\otimes B_j$ regardless of $\epsilon$; {see Fig.~\ref{fig:BEofnormalizedUsel}}. 
    
\end{rem}

Using $U'_{\rm SEL}$, $U_{\psi}$, and its conjugation $U^\dag_{\psi}$, we obtain a block encoding of the following Hamiltonian:
$$
\sum_{\bm{x}} \langle\psi|\tilde{H}^{(\bm{x})}|\psi\rangle \ket{\bm{x}}\bra{\bm{x}}=\sum_{\bm{x}} \left(\sum_j{x}_j\cdot\frac{\langle O_j\rangle}{M}\right) \ket{\bm{x}}\bra{\bm{x}}.
$$
This allows us to simulate the phase oracle for the affine linear function $\sum_j x_j\langle O_j\rangle/M$ by using the block-Hamiltonian simulation (Lemma~\ref{supple_lem:opt_blockHS}).
However, in estimating the expectation value $\langle O_j\rangle$, we need to amplify the function by the factor $M$ due to the normalization factor $1/M$ in $\tilde{H}^{(\bm{x})}$, which may result in no speedup for the total queries to $U_{\psi}$ regarding the number $M$ of observables.
To obtain the quadratic speedup regarding $M$, Ref.~\cite{van2023quantum} (Lemma 36) uses a linear amplification of $\tilde{H}^{(\bm{x})}$ by Lemma~\ref{supple_lem:uniform_amp}, and this amplification can be performed with no use of $U_{\psi}$.
Here, we provide a slightly modified version of the Lemma 36 in Ref.~\cite{van2023quantum} as follows.

\begin{lem}\label{supple_lem:BE_ampH}
    Let $\delta'>0$, $\varepsilon'\in (0,1/2)$, and let $\mathbf{p}=(p_1,p_2,...,p_M)$ be any sequence of positive integers. 
    Suppose that for observables $\{O_j\}_{j=1}^M$ in $d$ dimension ($d$ is a power of $2$) with $\|O_j\|\leq 1$, we have access to the oracle $$U'_{\rm SEL}=\sum_{\bm{x}\in G_{\mathbf{p}}}\ket{\bm{x}}\bra{\bm{x}}\otimes U^{(\bm{x})},$$ where $U^{(\bm{x})}$ is an $(a+\lceil \log_2 M\rceil)$-block encoding of observable $\tilde{H}^{(\bm{x})}:=M^{-1}\sum_{j=1}^M x_j O_j$ for some $a\in \mathbb{N}$.
    If the following inequality holds
    \begin{equation}\label{eq:condition_M}
        \sigma := \left\lceil{\sqrt{2M\ln(2d/\delta')}}\right\rceil<M,
    \end{equation}
    then there exists a subset $F$ of $G_{\mathbf{p}}$ with the cardinality $|F|$ satisfying 
    $$|F|\geq (1-\delta')|G_{\mathbf{p}}|,$$ 
    and we can implement a unitary (with $\|\mathbf{p}\|_1+\lceil\log_2{M}\rceil+\log_2 d+a+1$ qubits in total) 
    \begin{equation}\label{supple_eq:amplified_BE_obs}
            U_{\rm obs}:=\sum_{\bm{x}\in G_{\mathbf{p}}}\ket{\bm{x}}\bra{\bm{x}}\otimes U_{\rm obs}^{(\bm{x})}
    \end{equation}
    such that $U_{\rm obs}^{(\bm{x})}$ is a $(1,a+\lceil\log_2{M}\rceil+1,\varepsilon')$-block-encoding of the Hamiltonian
    \begin{equation}\label{eq:amplified_BE_temp}
            {H}^{(\bm{x})}:=\frac{1}{\sigma}\sum_{j=1}^M x_jO_j~~\mbox{if}~~\bm{x}\in F\subset G_{\mathbf{p}}.    
    \end{equation}
    For $m:=\mathcal{O}(M\sigma^{-1}\log(M\sigma^{-1}/\varepsilon'))$, this implementation of $U_{\rm obs}$ uses $m$ queries to $U'_{\rm SEL}$ or its inverse, $2m$ NOT gates controlled by $(a+\lceil\log_2{M}\rceil)$-qubit, and $\mathcal{O}(m)$ single-qubit gates, with $\mathcal{O}({\rm poly}(m))$ classical precomputation for finding quantum circuit parameters.
    The circuit diagram for $U_{\rm obs}$ is shown in Fig.~\ref{fig:amplified_BE_circuit}.

\end{lem}

\begin{proof}
    [Proof of Lemma~\ref{supple_lem:BE_ampH}]
    Let $\gamma:=M/\sigma>1$. 
    To perform the amplification in Lemma~\ref{supple_lem:uniform_amp}, it is required that the spectral norm $\|\tilde{H}^{(\bm{x})}\|$ is upper bounded by $(1-\delta)/\gamma$ for some $\delta\in (0,1/2)$.
    To this end, we define a subset $F\subset G_{\mathbf{p}}$ such that
    $$
    F:=\left\{\bm{x}\in G_{\mathbf{p}} : \|\tilde{H}^{(\bm{x})}\|<\frac{1}{2\gamma} \right\}.
    $$
    If $\bm{x}\in F$, we can perform the linear amplification by $\gamma$ and obtain $(1,a+\lceil \log_2 M\rceil+1,\varepsilon')$-block-encoding of $\gamma\tilde{H}^{(\bm{x})}$ from Lemma~\ref{supple_lem:uniform_amp}.
    Next, we show the cardinality of $F$ is equal or more than $(1-\delta')\prod_{j=1}^M 2^{p_j}$.
    Let us consider independent random variables $X_j~(j=1,...,M)$ that are uniformly distributed on each $G_{p_j}$.
    Because the random variables are also symmetrically distributed on each $G_{p_j}$, the following matrix series inequality holds (Theorem 35 in~\cite{van2023quantum}):
    \begin{align}\label{eq:prf_lm12_temp1}
        &{\rm Pr}_{X_1,...,X_M}\left[\left\|\sum_{j=1}^M (2X_j)O_j\right\|\geq t\right]\notag\\
        &~~~ \leq 2d\cdot{\exp}\left[-\frac{t^2}{2\left\|\sum_{j=1}^M O_j^2\right\|}\right].
    \end{align}
    Thus, inserting $M/\gamma$ into $t$, we obtain
    \begin{align}\label{eq:prf_lm12_tem2}
        {\rm Pr}_{X_1,...,X_M}\left[\left\|\tilde{H}^{(\bm{X})}\right\|\geq 1/2\gamma\right] \leq \delta'.
    \end{align}
    Since the probability of an event $X_j=x_j~(j=1,...,M)$ is equal to $1/\prod_{j=1}^M 2^{p_j}$ for all $x_j$, we obtain $|F|\geq (1-\delta')\prod_{j=1}^M 2^{p_j}$.
    The gate complexity follows from that of Lemma~\ref{supple_lem:uniform_amp}.
    
\end{proof}

Here, we make some remarks for Lemma~\ref{supple_lem:BE_ampH}.
In the proof of this lemma, we used the loose inequality: $\|\sum_{j=1}^M O_j^2\|\leq M$ (we assumed $\|O_j\|\leq 1$) to derive Eq.~\eqref{eq:prf_lm12_tem2} from Eq.~\eqref{eq:prf_lm12_temp1}, because explicitly calculating the factor $\|\sum_j O_j^2\|$ may be involved.
If we have some evaluation of the factor $\|\sum_j O_j^2\|$ in a particular setup, we can further improve the block encoding Eq.~\eqref{eq:amplified_BE_temp}; see Ref.~\cite{van2023quantum}.

Also, the condition Eq.~(\ref{eq:condition_M}) for the amplification, $M>\mathcal{O}({\log (d/\delta')})$, is crucial in our algorithm, and for a constant $\delta'$, this condition is satisfied when the number of observables is bigger than the number of qubits such as Eq.~(\ref{supple_eq:num_obs_condition}).

When we use the circuit construction of $U'_{\rm SEL}$ in Remark~\ref{rem:design_prepare}, $U^{(\bm{x})}$ in the assumption of Lemma~\ref{supple_lem:BE_ampH} should be replaced as the $\epsilon$-precise block encoding of $(2M/(\pi2^{\lceil \log_2 M\rceil}))\tilde{H}^{(\bm{x})}$.
In this case, the error $\varepsilon'$ in the block encoding $U^{(\bm{x})}_{\rm obs}$ of $H^{\bm{(x)}}$ becomes $\varepsilon'+\mathcal{O}(m\sqrt{\epsilon})$, by the error propagation associated with the uniform amplification Lemma~\ref{supple_lem:uniform_amp} (see; Lemma 22 in Ref.~\cite{gilyen2019quantum}).
Thus, replacing $\varepsilon'$ and $\epsilon$ with $\varepsilon'/2$ and $\mathcal{O}((\varepsilon'/m)^2)$, respectively, we obtain the $\varepsilon'$-precise block encoding $U^{(\bm{x})}_{\rm obs}$.

\subsubsection{Approximated state preparation}
Now, we are ready to show two quantum algorithms to prepare the probing state $\ket{\Upsilon(q)}$.
Note that in the following lemmas, we explicitly construct a quantum circuit in its proof.

\begin{lem}
    [State preparation by Hamiltonian simulation]\label{supple_lem:statepre_HS}
    
    Let $O_j~(j=1,2,...,M)$ be observables in $d$ dimension ($d$ is a power of $2$) with $\|O_j\|\leq 1$, and let $p$ be a positive integer.
    Also, let $U_{\psi}$ be a $\log_2 d$-qubit state preparation oracle, and $\braket{O_j}:=\bra{\psi}O_j\ket{\psi}$.
    Suppose we have access to $$U'_{\rm SEL}:=\sum_{\bm{x}\in G_p^M}\ket{\bm{x}}\bra{\bm{x}}\otimes U^{(\bm{x})},$$ where $U^{(\bm{x})}$ is an $(a+\lceil\log_2{M}\rceil)$-block-encoding of $M^{-1}\sum_{j=1}^M x_jO_j$ for some $a\in \mathbb{N}$.
    We assume that $$\sigma := \left\lceil{\sqrt{2M\ln(2d/\delta')}}\right\rceil<M$$ holds for $\delta'=2^{-10}$. 
    Then, for any non-negative integer $q$, we can prepare the $pM$-qubit quantum state 
    $$
    \frac{1}{\sqrt{2^{pM}}}\sum_{\bm{x}\in G_p^M} e^{2\pi i2^p \sum_{j=1}^M x_j {2^{q+1}\pi^{-1}\braket{{O}_j}}} \ket{\boldsymbol{x}}
    $$
    up to 1/12 Euclidean distance error, with $2Q$ uses of $U_{\psi}$ or $U^\dagger_{\psi}$, $4mQ$ uses of controlled $U'_{\rm SEL}$ or its inverse, $\mathcal{O}(mQ)$ uses of NOT gates controlled by at most $\mathcal{O}(a+\lceil\log_2{M}\rceil+\log_2 d)$ qubits, $\mathcal{O}(mQ+pM)$ uses of single-qubit or two-qubit gates, additional $(\lceil\log_2{M}\rceil+\log_2 d+a+3)$ ancilla qubits, and $\mathcal{O}({\rm poly}(Q)+{\rm poly}(m))$ classical precomputation for finding quantum circuit parameters, where $$
    Q=\mathcal{O}\left(2^{p+q+2}\sigma\right)~~\mbox{and}~~
    m=\mathcal{O}\left(\frac{{M}}{\sigma}(p+q+\log M)\right).
    $$
\end{lem}

Note that if we have a block encoding $B_j$ of $O_j$, then a block encoding of $(O_j-\tilde{u}^{(q)}_{j}\bm{1})/2$ for any $\tilde{u}^{(q)}_{j}\in [-1,1]$ is easily constructed by the LCU method; see Fig.~\ref{supplefig:LCU_O_u}.
Therefore, Lemma~\ref{supple_lem:statepre_HS} provides a method to approximately prepare the probing state $\ket{\Upsilon(q)}$.

\begin{proof}
    [Proof of Lemma~\ref{supple_lem:statepre_HS}]
    Let $\varepsilon''=2^{-14}\in (0,1)$ and $\varepsilon'=\sqrt{\delta'}/(2^{p+q+2}\sigma)\in (0,1/2)$.
    From the assumption and Lemma~\ref{supple_lem:BE_ampH}, we have a unitary 
    \begin{equation}
        U_{\rm obs}:=\sum_{\bm{x}\in G_{{p}}^M}\ket{\bm{x}}\bra{\bm{x}}\otimes U_{\rm obs}^{(\bm{x})},
    \end{equation}
    where $U_{\rm obs}^{(\bm{x})}$ is a $(1,a''=a+\lceil\log_2{M}\rceil+1,\varepsilon')$-block-encoding of the Hamiltonian ${H}^{(\bm{x})}:={\sigma}^{-1}\sum_{j=1}^M x_jO_j$ if $\bm{x}\in F$ with $|F|\geq (1-\delta')|G_p^M|$.
    Thus, using $U_{\psi}, U_{\psi}^\dagger$, and $U_{\rm obs}$, we can implement a $(1,a''+\log_2 d,0)$-block-encoding of the following Hamiltonian 
    \begin{equation}\label{eq:phase_fn_HS}
        \sum_{\bm{x}\in G_p^M} \tilde{f}(\bm{x})\ket{\bm{x}}\bra{\bm{x}},
    \end{equation}
    where
    $$
    \tilde{f}(\bm{x}):=\bra{0}^{\otimes a''+\log_2 d}U_{\psi}^\dagger\cdot U_{\rm obs}^{(\bm{x})}\cdot U_{\psi}\ket{0}^{\otimes a''+\log_2 d}.
    $$
    The block-Hamiltonian simulation (Lemma~\ref{supple_lem:opt_blockHS}) yields the quantum circuit $W$ for a $(1,a''+\log_2 d+2,\varepsilon'')$-block-encoding of time evolution operator 
    \begin{equation}
        \sum_{\bm{x}\in G_p^M} e^{i\tilde{f}(\bm{x})t}\ket{\bm{x}}\bra{\bm{x}}~~\mbox{and}~~t:=2^{p+2+q}\sigma.
    \end{equation}
    An example of this quantum circuit is shown in Fig.~\ref{fig:spbyHS_fullcircuit}.
    Note that a block encoding $W'$ of $\varepsilon''$-precise time evolution operator $e^{iHt}$ for some Hamiltonian $H$ yields $\mathcal{O}(\sqrt{\varepsilon''})$-precise time evolved states, that is, for any state $\ket{\psi}$, 
    $$
    \left\|W'\ket{\bm{0}}\ket{\psi}-\ket{\bm{0}}e^{iHt}\ket{\psi}\right\| \leq \varepsilon''+\sqrt{2\varepsilon''},
    $$
    where $\ket{\bm{0}}$ is the signal state for $W'$.
    Therefore, for $\ket{\bm{0}}=\ket{0}^{\otimes a''+\log_2 d+2}$ and the $+1$ eigenstate $\ket{+}$ of Pauli X , we obtain
    \begin{widetext}
    \begin{align}
        &\left\|W\ket{\bm{0}}\ket{+}^{\otimes pM}-\frac{1}{\sqrt{2^{pM}}}\sum_{\bm{x}\in G_p^M} e^{2\pi i 2^p \sum_{j=1}^M x_j \pi^{-1}2^{q+1}\langle O_j\rangle} \ket{\bm{0}}\ket{\bm{x}}\right\|\notag\\[6pt]
        &\leq\varepsilon''+\sqrt{2\varepsilon''}+\frac{1}{\sqrt{2^{pM}}}\left\|\sum_{\bm{x}\in G_p^M} e^{i\tilde{f}(\bm{x})t}\ket{\bm{x}}-\sum_{\bm{x}\in G_p^M} e^{2\pi i 2^p \sum_{j=1}^M x_j \pi^{-1}2^{q+1}\langle O_j\rangle} \ket{\bm{x}}\right\|\notag\\[6pt]
        &\leq \varepsilon''+\sqrt{2\varepsilon''}+\sqrt{5\delta'}<\frac{1}{12}.
    \end{align}
    Here, in the second inequality we used the following evaluation.
    \begin{align}
        &\left\|\sum_{\bm{x}\in G_p^M} e^{i\tilde{f}(\bm{x})t}\ket{\bm{x}}-\sum_{\bm{x}\in G_p^M} e^{2\pi i 2^p \sum_{j=1}^M x_j \pi^{-1}2^{q+1}\langle O_j\rangle} \ket{\bm{x}}\right\|^2\notag\\[6pt]
        &\leq \sum_{\bm{x}\in G_p^M} \left|e^{i\tilde{f}(\bm{x})t}-e^{2\pi i 2^p \sum_{j=1}^M x_j \pi^{-1}2^{q+1}\langle O_j\rangle} \right|^2\notag\\[6pt]
        &\leq 4\times 2^{pM}\delta'+\sum_{\bm{x}\in F} \left|e^{i\tilde{f}(\bm{x})t}-e^{2\pi i 2^p \sum_{j=1}^M x_j \pi^{-1}2^{q+1}\langle O_j\rangle} \right|^2~~~(\because |F|\geq (1-\delta')|G_p^M|)\notag\\[6pt]
        &\leq 4\times 2^{pM}\delta'+\sum_{\bm{x}\in F} \left|\tilde{f}(\bm{x})t-2\pi 2^p \sum_{j=1}^M x_j \pi^{-1}2^{q+1}\langle O_j\rangle \right|^2~~~(\because |e^{ia}-e^{ib}|\leq|a-b|)\notag\\[6pt]
        &\leq 4\times 2^{pM}\delta'+\sum_{\bm{x}\in F} (t\varepsilon')^2\leq 2^{pM}\times 5\delta'~~~(\because \varepsilon':=\sqrt{\delta'}/t).
    \end{align}
    The gate complexity follows from that of Lemma~\ref{supple_lem:BE_ampH} and Lemma~\ref{supple_lem:opt_blockHS}.
    \end{widetext}
\end{proof}

\begin{lem}
    [State preparation by Grover-like repetition]\label{supple_lem:statepre_Grover}

    Let $O_j~(j=1,2,...,M)$ be observables in $d$ dimension ($d$ is a power of $2$) with $\|O_j\|\leq 1$, and let $p$ be a positive integer. We assume the following two conditions holds for $\delta':=2^{-14}\in(0,1)$:
    \begin{itemize}
            \item[(i)] $\sigma := \left\lceil{\sqrt{2(M+1)\ln(2d/\delta')}}\right\rceil<M+1$.
            \item[(ii)] For a given $(\log_2 d)$-qubit quantum state $\ket{\psi}$ prepared by $U_{\psi}$, $|\langle{O_j}\rangle |\leq 2^{-q-1}$ holds for all $j=1,2,...,M$, where $q$ denotes a non-negative integer satisfying 
            \begin{equation}\label{supple_eq:grver_sp_condition}
                q\geq \log_4\left[\frac{2^{p}\cdot 33^3}{625\ln (2d/\delta')}\frac{\left\lceil{\sqrt{2(M+1)\ln(2d/\delta')}}\right\rceil}{\sqrt{\ln (2d/\delta')}}\right]
            \end{equation}
    \end{itemize}
    Suppose we have access to $$U'_{\rm SEL}:=\sum_{\bm{x}\in G_p^M\times G_1}\ket{\bm{x}}\bra{\bm{x}}\otimes U^{(\bm{x})},$$
    where $U^{(\bm{x})}$ is an $(a+\lceil\log_2{(M+1)}\rceil)$-block-encoding of
    $$\frac{1}{M+1}\left(\sum_{j=1}^{M+1} x_jO_j\right)
    $$
    for some $a\in \mathbb{N}$. Here, $O_{M+1}$ denotes a $2\times 2$ observable defined as
    $$
    O_{M+1}:=\frac{\pi(1/4+4l)}{2^{p+q}}I,
    $$
    where $l$ is an arbitrary integer satisfying $\|2^{q+1}O_{M+1}\|\leq 1$ and $I$ denotes the 1-qubit identity.
    Then, we can successfully prepare the following $pM$-qubit quantum state with a measurement result of ancilla qubits indicating the success:
    $$
    \frac{1}{\sqrt{2^{pM}}}\sum_{\bm{x}\in G_p^M} e^{2\pi i2^p \sum_{j=1}^M x_j {2^{q+1}\pi^{-1}\braket{{O}_j}}} \ket{\boldsymbol{x}}
    $$
    up to 1/12 Euclidean distance error with the success probability $\geq 0.462$.
    This implementation consists of $2t$ uses of $U_{\psi}$ or $U_{\psi}^\dagger$, $mt$ uses of $U'_{\rm SEL}$ or its inverse, $(2m+1)t$ uses of NOT gates controlled by at most $(a+\lceil \log_2(M+1)\rceil+\log_2 d)$ qubits, $\mathcal{O}(mt+pM)$ uses of single-qubit or two-qubit gates, additional $(\lceil \log_2(M+1)\rceil+\log_2 d+a+2)$ ancilla qubits, and $\mathcal{O}({\rm poly}(m))$ classical precomputation for finding quantum circuit parameters, where 
    $$
    t=2^{p+q+2}\sigma~~~\mbox{and}~~~m = \mathcal{O}\left(\frac{M}{\sigma}\left(p+q+\log M\right)\right).
    $$
\end{lem}

\begin{widetext}
\begin{proof}
    [Proof of Lemma~\ref{supple_lem:statepre_Grover}]
    Let $\varepsilon'=\sqrt{\delta'}/(2^{p+q+2}\sigma)\in (0,1/2)$.
    From the assumption and Lemma~\ref{supple_lem:BE_ampH}, we have a unitary 
    \begin{equation}
        U_{\rm obs}:=\sum_{(\bm{x},y)\in G_{{p}}^M\times G_1}\ket{\bm{x},y}\bra{\bm{x},y}\otimes U_{\rm obs}^{(\bm{x},y)},
    \end{equation}
    where $U_{\rm obs}^{(\bm{x},y)}$ is a $(1,a''=a+\lceil\log_2{(M+1)}\rceil+1,\varepsilon')$-block-encoding of the Hamiltonian $$
    {H}^{(\bm{x},y)}:=\frac{1}{\sigma}\left(yO_{M+1}+\sum_{j=1}^{M} x_jO_j\right)$$
    if $(\bm{x},y)\in F$ with $|F|\geq (1-\delta')2|G_p^M|$.
    Thus, using $U_{\psi}, U_{\psi}^\dagger$, and $U_{\rm obs}$, we can implement a $(1,a''+\log_2 d,0)$-block-encoding of the following Hamiltonian 
    \begin{equation}\label{suppleeq:BE_grover_hamiltonian}
        \sum_{(\bm{x},y)\in G_p^M\times G_1} \tilde{f}(\bm{x},y)\ket{\bm{x},y}\bra{\bm{x},y},~~~\mbox{where}~~~\tilde{f}(\bm{x},y):=\bra{0}^{\otimes a''+\log_2 d}U_{\psi}^\dagger\cdot U_{\rm obs}^{(\bm{x},y)}\cdot U_{\psi}\ket{0}^{\otimes a''+\log_2 d}.
    \end{equation}
    Since $U^{(\bm{x},y)}_{\rm obs}$ is the $\varepsilon'$-precise block encoding of ${H}^{(\bm{x},y)}$, $\tilde{f}(\bm{x},y)$ is close to an affine linear function as
    \begin{equation}\label{supple_eq:spbyparallelG_tildef}
        \left|\tilde{f}(\bm{x},y)-\frac{1}{\sigma}\left(y\langle O_{M+1}\rangle+\sum_{j=1}^{M}x_j\langle O_j\rangle\right)\right|\leq \varepsilon'~~\mbox{for~all}~~(\bm{x},y)\in F.
    \end{equation}
    Thus, Grover-like repetition on Eq.~\eqref{suppleeq:BE_grover_hamiltonian}, that is, Lemma~\ref{supple_lem:QET_by_Tmx} yields the quantum circuit $W$ for an $(a''+\log_2 d)$-block-encoding of the following operator 
    $$
        \sum_{(\bm{x},y)\in G_p^M\times G_1} T_{t}\left(\tilde{f}(\bm{x},y)\right)\ket{\bm{x},y}\bra{\bm{x},y}~~\mbox{and}~~t:=2^{p+q+2}\sigma.
    $$
    Applying $W$ to the state $\ket{+}^{\otimes pM+1}\ket{\bm{0}}$ with $\ket{\bm{0}}=\ket{0}^{\otimes a''+\log_2 d}$, we obtain
    \begin{align}\label{supple_eq:spbyparallelg_bfpostselect}
        W\ket{+}^{\otimes pM+1}\ket{\bm{0}}=\frac{\mathcal{N}_t}{\sqrt{2}}\left[\frac{1}{\mathcal{N}_t}\frac{1}{\sqrt{2^{pM}}}\sum_{\bm{x},y} T_{t}\left(\tilde{f}(\bm{x},y)\right)\ket{\bm{x},y}\right]\ket{\bm{0}}+\sqrt{1-\left(\frac{\mathcal{N}_t}{\sqrt{2}}\right)^2}\ket{{\Phi}^{\perp}},
    \end{align} 
    where $\ket{{\Phi}^{\perp}}$ denotes a normalized quantum state satifying $(\bm{1}\otimes \ket{\bm{0}}\bra{\bm{0}})\ket{{\Phi}^{\perp}}=\bm{0}$ and $\mathcal{N}_t$ is defined as:
    \begin{equation}\label{supple_eq:spbyparallelg_successprob}
        \mathcal{N}_t := \left\|\frac{1}{\sqrt{2^{pM}}}\sum_{\bm{x},y} T_{t}\left(\tilde{f}(\bm{x},y)\right)\ket{\bm{x},y}\right\|.
    \end{equation}
    Measuring the $(a''+\log_2 d)$-qubit ancilla registers in \eqref{supple_eq:spbyparallelg_bfpostselect}, we obtain the outcome $\ket{\bm{0}}$ and the following state with the probability $\mathcal{N}_t^2/2$:
    \begin{equation}\label{supple_eq:spbyparallelg_successstate}
        \frac{1}{\mathcal{N}_t}\frac{1}{\sqrt{2^{pM}}}\sum_{\bm{x},y} T_{t}\left(\tilde{f}(\bm{x},y)\right)\ket{\bm{x},y}.
    \end{equation}
    In the following, we show this quantum state is close to the following state 
    \begin{align}\label{supple_eq:spbyparallelg_approximatedstate}
        &\frac{1}{\sqrt{2^{pM}}}\sum_{(\bm{x},x_{M+1})\in G_p^M\times G_1}\cos\left[t\left(\frac{\pi}{2}-\frac{1}{\sigma}\sum_{j=1}^{M+1}x_j\langle O_j\rangle\right)\right]\ket{\bm{x}}\ket{x_{M+1}}\notag\\[6pt]
        &=\frac{1}{\sqrt{2}}\sum_{s=\pm 1}\sum_{\bm{x}\in G_p^M} \frac{1}{\sqrt{2^{pM}}}e^{2\pi i 2^p \sum_{j=1}^M (sx_j) \pi^{-1}2^{q+1}\langle O_j\rangle}\ket{\bm{x}}\otimes {\rm QFT}_{G_1}\ket{{s}/{4}},
    \end{align}
    where we recall that $\ket{-1/4}\equiv \ket{0}$ and $\ket{1/4}\equiv \ket{1}$ hold by the labelling in Eq.~\eqref{eq:varphi_bijection}.
    Note that the equality in Eq.~(\ref{supple_eq:spbyparallelg_approximatedstate}) can be shown as follows.
    \begin{align}
        &\sum_{(\bm{x},x_{M+1})\in G_p^M\times G_1}\cos\left[t\left(\frac{\pi}{2}-\frac{1}{\sigma}\sum_{j=1}^{M+1}x_j\langle O_j\rangle\right)\right]\ket{\bm{x}}\ket{x_{M+1}}\notag\\[6pt]
        &=\frac{1}{2}\sum_{\bm{x},y} \sum_{s=\pm 1}e^{-is\frac{t}{\sigma}\left(y\langle O_{M+1}\rangle+\sum_{j=1}^{M}x_j\langle O_j\rangle\right)}\ket{\bm{x}}\ket{y}~~~(\because e^{ist{\pi}/2}=e^{2\pi is2^{p+q}\sigma}=1)\notag\\[6pt]
        &=\frac{1}{\sqrt{2}}\sum_{s=\pm 1}\sum_{\bm{x}\in G_p^M} e^{-ist\frac{1}{\sigma}\sum_{j=1}^{M}x_j\langle O_j\rangle}\ket{\bm{x}}\otimes \frac{1}{\sqrt{2}}\sum_{y\in G_1}e^{-ist\frac{1}{\sigma}y\langle O_{M+1}\rangle}\ket{y}\notag\\[6pt]
        &=\frac{1}{\sqrt{2}}\sum_{s=\pm 1}\sum_{\bm{x}\in G_p^M} e^{-ist\frac{1}{\sigma}\sum_{j=1}^{M}x_j\langle O_j\rangle}\ket{\bm{x}} \otimes {\rm QFT}_{G_1}\ket{{-s}/{4}}~~~(\because \langle O_{M+1}\rangle=\frac{\pi(1/4+4l)}{2^{p+q}})\notag\\[6pt]
        &=\frac{1}{\sqrt{2}}\sum_{s=\pm 1}\sum_{\bm{x}\in G_p^M} e^{-is2\pi 2^{p}\sum_{j=1}^{M}x_j\pi^{-1}(2^{q+1})\langle O_j\rangle}\ket{\bm{x}} \otimes {\rm QFT}_{G_1}\ket{{-s}/{4}}.
    \end{align}
    Thus, by applying ${\rm QFT}_{G_1}^\dagger$ and controlled $X^{\otimes pM}$ gate to the state \eqref{supple_eq:spbyparallelg_bfpostselect}, we can prepare the target state:
    \begin{equation}\label{supple_eq:spbyparallelg_bfpostselect2}
        \frac{1}{\sqrt{2^{pM}}}\sum_{\bm{x}\in G_p^M} e^{2\pi i2^p \sum_{j=1}^M x_j {2^{q+1}\pi^{-1}\braket{{O}_j}}} \ket{\boldsymbol{x}}\otimes \ket{+}
    \end{equation}
    up to some Euclidean distance error with the probability $\mathcal{N}_t^2/2$, because $X^{\otimes pM}\ket{\bm{x}}=\ket{-\bm{x}}$ holds.
    Note that ${\rm QFT}_{G_1}^\dagger$ and controlled $X^{\otimes pM}$ does not change the probability to obtain the success flag $\ket{\bm{0}}$.
    Here, we show the corresponding quantum circuit to prepare Eq.~(\ref{supple_eq:spbyparallelg_bfpostselect2}) in Fig.~\ref{fig:sp_Grover_diagram}.
        
    For subsequent proof, we here show that for any $\delta'\in (0,1)$ and any real weight $w_j~(j=1,2,...,M+1)$, a subset $F'$ of $G_{p}^M\times G_1$ defined as
    \begin{equation}
        F':=\left\{(\bm{x},y)\in G_{p}^M\times G_1: \left|\sum_{j=1}^{M+1}w_j x_j\right|< \sqrt{{{\frac{\ln(2/\delta')}{2}\sum_{j=1}^{M+1}|w_j|^2}}}\right\}
    \end{equation}
    has the cardinality $|F'|\geq (1-\delta')2^{pM+1}$.
    This fact can be proved from the Hoeffding's inequality as follows.
    Let us consider independent random variables $X_j~(j=1,2,...,M)$ and $Y$. Here, $X_j$ are identically and uniformly distributed in $G_p$ and $Y$ is uniformly distributed in $G_1$.
    Note that the random variables are upper bounded by $1/2$ from the definition of $G_p$.
    Therefore, from the Hoeffding's inequality, we obtain
    \begin{align}
        {\rm Pr}_{X_1,...,X_M,Y}\left[\left|\sum_{j=1}^{M+1} w_j X_j \right|\geq c \right]\leq 2{\rm exp}\left[{-\frac{2c^2}{{\sum_j |w_j|^2}}}\right].
    \end{align}
    Thus, taking $c$ as
    $c=\sqrt{\frac{\ln (2/\delta')}{2}\sum_{j=1}^{M+1}|w_j|^2}$, we obtain 
    \begin{align}
        {\rm Pr}_{X_1,...,X_M,Y}\left[\left|\sum_{j=1}^{M+1} w_j X_j \right|\geq \sqrt{\frac{\ln (2/\delta')}{2}\sum_{j=1}^{M+1}|w_j|^2} \right]\leq \delta'.
    \end{align}
    Since the independent random variables $X_1,...,X_M,Y$ has uniform distributions, $F'$ has the cardinality $\geq 2^{pM+1}(1-\delta')$.    

    Now, we are ready to show the closeness between Eq.~\eqref{supple_eq:spbyparallelg_successstate} and Eq.~\eqref{supple_eq:spbyparallelg_approximatedstate}.
    Defining $F'$ for $\delta'=2^{-14}$ and $w_j=\braket{O_j}$ as
    \begin{equation}\label{supple_eq:dfn_Fp_for_obs}
        F':=\left\{(\bm{x},y)\in G_{p}^M\times G_1: \left|\sum_{j=1}^{M+1}x_j\langle O_j\rangle\right|< \sqrt{{{\frac{\ln(2/\delta')}{2}\sum_{j=1}^{M+1}|\langle O_j\rangle|^2}}}\right\}
    \end{equation}
    and recalling that $|F'|,|F|\geq (1-\delta')2^{pM+1}$, we proceed as follows.
    \begin{align}
    &\left\|\sum_{(\bm{x},y)\in G_p^M\times G_1}T_t(\tilde{f}\left(\bm{x},y)\right)\ket{\bm{x}}\ket{y}-\sum_{(\bm{x},y)\in G_p^M\times G_1}\cos\left[t\left(\frac{\pi}{2}-\frac{1}{\sigma}\sum_{j=1}^{M+1}x_j\langle O_j\rangle\right)\right]\ket{\bm{x}}\ket{y}\right\|^2\notag\\[6pt]
    &=\sum_{(\bm{x},x_{M+1})\in G_p^M\times G_1} \left|T_t(\tilde{f}(\bm{x},x_{M+1}))-\cos\left[t\left(\frac{\pi}{2}-\frac{1}{\sigma}\sum_{j=1}^{M+1}x_j\langle O_j\rangle\right)\right]\right|^2\notag\\[6pt]
    &\leq 4\times 2\times 2^{pM+1}\delta'+\sum_{(\bm{x},x_{M+1})\in F\cap F'} \left|t\cos^{-1}\left[\tilde{f}(\bm{x},x_{M+1})\right]-t\left(\frac{\pi}{2}-\frac{1}{\sigma}\sum_{j=1}^{M+1}x_j\langle O_j\rangle\right)\right|^2,
    \end{align}
    where we used $|\cos(a)-\cos(b)|\leq |a-b|$ and $G_p^M\times G_1=(F\cap F')\cup(F^{\rm c}\cup (F')^{\rm c})$ in the third line.
    Now, we focus on the case $(\bm{x},x_{M+1})\in F \cap F'$.
    \begin{align}\label{eq:B33}
    |\tilde{f}(\bm{x},x_{M+1})|&\leq \varepsilon'+\left|\frac{1}{\sigma}\sum_{j=1}^{M+1}x_j\langle O_j\rangle\right|~~~(\because \mbox{Eq.~(\ref{supple_eq:spbyparallelG_tildef})})\notag\\[6pt]
    &\leq \varepsilon'+\sqrt{\frac{\sum_{j}|\langle O_j\rangle|^2}{4(M+1)}}\sqrt{\frac{\ln(2/\delta')}{\ln(2d/\delta')}}~~~(\because \mbox{Eq.~(\ref{supple_eq:dfn_Fp_for_obs})})\notag\\[6pt]
    &\leq \frac{\sqrt{\delta'}}{2^{p+q+2}\sigma}+\frac{1}{2^{q+2}}\sqrt{\frac{\ln(2/\delta')}{\ln(2d/\delta')}}<\frac{33/10}{2^{q+2}\sqrt{\ln (2d/\delta')}}<1/4.
    \end{align}
    In the case of $|x|\leq 1/4$,
    $\left|\cos^{-1}(x)-\pi/2 + x \right|\leq \frac{1}{5}|x|^3$ holds, and therefore we obtain
    \begin{align}\label{eq:B34}
    &t\left|\cos^{-1}\left[\tilde{f}(\bm{x},x_{M+1})\right]-\left(\frac{\pi}{2}-\frac{1}{\sigma}\sum_{j=1}^{M+1}x_j\langle O_j\rangle\right)\right|\notag\\[6pt]
    &=t\left|\cos^{-1}\left[\tilde{f}(\bm{x},x_{M+1})\right]-\frac{\pi}{2}+\tilde{f}(\bm{x},x_{M+1})+\left(\frac{1}{\sigma}\sum_{j=1}^{M+1}x_j\langle O_j\rangle-\tilde{f}(\bm{x},x_{M+1})\right)\right|\notag\\[6pt]
    &\leq t\varepsilon'+t\left|\cos^{-1}\left[\tilde{f}(\bm{x},x_{M+1})\right]-\frac{\pi}{2}+\tilde{f}(\bm{x},x_{M+1})\right|\notag\\[6pt]
    &\leq t\varepsilon'+\frac{t}{5}\left|\frac{33/10}{2^{q+2}\sqrt{\ln (2d/\delta')}}\right|^{3}.
    \end{align}
    From the assumption, the non-negative integer $q$ satisfies $\frac{t}{5}\left|\frac{33/10}{2^{q+2}\sqrt{\ln (2d/\delta')}}\right|^{3} \leq \sqrt{\delta'}$,
    then we conclude that
    \begin{align}
    &\left\|\sum_{(\bm{x},y)\in G_p^M\times G_1}T_t(\tilde{f}\left(\bm{x},y)\right)\ket{\bm{x}}\ket{y}-\sum_{(\bm{x},y)\in G_p^M\times G_1}\cos\left[t\left(\frac{\pi}{2}-\frac{1}{\sigma}\sum_{j=1}^{M+1}x_j\langle O_j\rangle\right)\right]\ket{\bm{x}}\ket{y}\right\|\notag\\[6pt]
    &\leq \sqrt{4\times 2\times 2^{pM+1}\delta'+\sum_{(\bm{x},x_{M+1})\in F\cap F'} (2\sqrt{\delta'})^2}\notag\\
    &\leq 2\sqrt{6\delta'}\times \sqrt{2^{pM}}.
    \end{align}
    This leads to 
    \begin{equation}
        1-2\sqrt{6\delta'}\leq \mathcal{N}_t = \left\|\frac{1}{\sqrt{2^{pM}}}\sum_{\bm{x},y} T_{t}\left(\tilde{f}(\bm{x},y)\right)\ket{\bm{x},y}\right\|\leq 1+2\sqrt{6\delta'}
    \end{equation}
    and
    \begin{align}
        &\left\|\frac{1}{\mathcal{N}_t}\sum_{(\bm{x},y)\in G_p^M\times G_1}T_t(\tilde{f}\left(\bm{x},y)\right)\ket{\bm{x}}\ket{y}-\sum_{(\bm{x},y)\in G_p^M\times G_1}\cos\left[t\left(\frac{\pi}{2}-\frac{1}{\sigma}\sum_{j=1}^{M+1}x_j\langle O_j\rangle\right)\right]\ket{\bm{x}}\ket{y}\right\|\notag\\
        &\leq \sqrt{2^{pM}}\times \frac{4\sqrt{6\delta'}}{1-2\sqrt{6\delta'}} < \sqrt{2^{pM}}\times \frac{1}{12}.
    \end{align}
    Therefore, we can prepare the target state up to $1/12$ Euclidean distance error with the probability $\mathcal{N}_t^2/2>0.462$.

    Finally, the gate complexity follows from that of Lemma~\ref{supple_lem:BE_ampH} and Lemma~\ref{supple_lem:QET_by_Tmx}.

\end{proof}
\end{widetext}

Note that we can coherently amplify the success probability in Lemma~\ref{supple_lem:statepre_Grover} by quantum amplitude amplification~\cite{brassard2002quantum}, while this requires a quantum circuit with 3-fold depth compared to the quantum circuit (before measurement) in Fig.~\ref{fig:sp_Grover_diagram}.


\subsection{Maximum mean squared error (MSE) and query complexity}\label{supple_sec;complexity_of_alg}

\subsubsection{The proof of Theorem~\ref{thm:main_query_complexity}}

Here, we evaluate the relation between the root mean squared error $\varepsilon$ and the total queries to the state preparation in Algorithm~\ref{supp_alg:main}.


\begin{proof}
    [Proof of Theorem~\ref{thm:main_query_complexity}]
    The core idea of this proof is similar to the previous methods~\cite{kimmel2015robust,dutkiewicz2022heisenberg,dutkiewicz2023advantage}, but we need carefully to deal with the condition of gradient estimation.
    In the following, we prove this theorem based on the state preparation by Lemma~\ref{supple_lem:statepre_HS}; see remarks after this proof for the case of Lemma~\ref{supple_lem:statepre_Grover}. 
    
    We start by clarifying the statistical property of $g^{(q)}_{j}$ in the step 5 of Algorithm~\ref{supp_alg:main}.
    The truncation in Step 7 guarantees $\tilde{u}^{(q)}_j\in [-1,1]$ for all $q$; we can construct a block-encoding $U'_{\rm SEL}$ for $\{\tilde{O}^{(q)}_j\}$, which is assumed to be accessible in Lemma~\ref{supple_lem:statepre_HS}, from the block encodings of $\{O_j\}$; see Fig.~\ref{supplefig:LCU_O_u}.
    Therefore, using the Lemma~\ref{supple_lem:statepre_HS}, we can prepare the quantum state in Step 4.
    
    Now, we consider the case that the following condition for gradient estimation holds at the beginning of an iteration $q$: for all $j$,
    \begin{equation}\label{supple_eq:initcondition4gradest}
        \left|\braket{\tilde{O}_j^{(q)}}\right|=\left|\frac{\braket{{O}_j}-\tilde{u}_j^{(q)}}{2}\right|\leq 2^{-q-1}.
    \end{equation}
    Then, a single shot measurement result $\bm{k}:=(k_1,...,k_M)\in G_p^M$ in Step 4 follows 
    $$
    {\rm Pr}\left[\left|k_j-\frac{2^q(\braket{O_j}-\tilde{u}_j^{(q)})}{\pi}\right|>\frac{3}{2^p}\right]\leq \frac{1}{3},
    $$
    for every $j=1,2,...,M$, from the analysis of gradient estimation in Lemma~\ref{supple_lem:original_gradest}.
    If we take $p=5$, then the additive error $3/2^p$ is smaller than $1/2\pi$, and therefore, the temporal estimate $\tilde{u}^{(q)}_j+\frac{\pi}{2^q}k_j$ becomes a 1-bit more precise estimate of $\braket{O_j}$ at least 2/3 probability.
    However, the choice of $p=5$ is sufficient but not tight.
    Here, we employ the following tighter bound that we numerically found~\footnote{More precisely, in the setup of Lemma~\ref{supple_lem:original_gradest} without any Euclidean distance error, we numerically evaluate the probability of an event $|k_j-g_j|>1/2\pi$ for $g_j\in [-1/\pi,1/\pi]$ when $p=3$.}: if we take $p=3$, 
    \begin{equation}\label{supple_eq:tightbound_g}
        {\rm Pr}\left[\left|k_j-\frac{2^q(\braket{O_j}-\tilde{u}_j^{(q)})}{\pi}\right|>\frac{1}{2\pi}\right]< 0.18+\frac{1}{12},
    \end{equation}
    holds for every $j$, where the term $1/12$ arises from the Euclidean distance error as well as Lemma~\ref{supple_lem:original_gradest}.
    Since the $g_{j}^{(q)}$ is defined as the coordinate-wise median of independent samples $\bm{k}^{(1)},...,\bm{k}^{(\#)}$, 
    the Hoeffding's inequality for the independent and bounded random variables $\{\chi[|k^{(i)}_j-{2^q(\braket{O_j}-\tilde{u}_j^{(q)})}/{\pi}|>1/(2\pi)]\}_i$ yields
    \begin{equation}
        {\rm Pr}\left[\left|g_j^{(q)}-\frac{2^q(\braket{O_j}-\tilde{u}_j^{(q)})}{\pi}\right|>\frac{1}{2\pi}\right]\leq e^{-\frac{\#}{9}},
    \end{equation}
    for every $j$. Here, $\chi[\bullet]$ denotes the indicator function.
    Therefore, taking $\#:=9\ln (M/\delta^{(q)})$ and using the union bound, 
    we can bound the probability of the event 
    $$\mathbf{A}^{(q)}: \max_j \left|g_j^{(q)}-{2^q(\braket{O_j}-\tilde{u}_j^{(q)})}/{\pi}\right|\leq \frac{1}{2\pi}$$ 
    as ${\rm Pr}[\mathbf{A}^{(q)}]\geq 1-\delta^{(q)}$.
    In the event $\mathbf{A}^{(q)}$, $|\braket{\tilde{O}^{(q+1)}_j}|\leq 2^{-q-2}$ holds even if we truncate $\tilde{u}^{(q+1)}_j$ in Step 7.
    On the other hand, if the condition Eq.~(\ref{supple_eq:initcondition4gradest}) is false, the gradient estimation does not work, and we only say that the measurement result $g^{(q)}_j$ is bounded as $g^{(q)}_j\in [-1/2,1/2]$ because of the definition of $G_p^M$.

    In the case of $q=0$, the condition Eq.~(\ref{supple_eq:initcondition4gradest}) holds, and thus, the event $\mathbf{A}^{(0)}$ occurs with the probability $1-\delta^{(0)}$.
    By repeating this, in branches such that all of $\{\mathbf{A}^{(q)}\}_{q=0}^{q'-1}$ occur, the temporal estimate $\tilde{u}^{(q')}_j$ satisfy $$|\tilde{u}^{(q')}_j-\braket{O_j}|\leq 1/2^{q'}$$ 
    for all $j$.
    Moreover, considering branches such that all of $\{\mathbf{A}^{(q)}\}_{q=0}^{q'-1}$ occur but the complement of $\mathbf{A}^{(q')}$ occurs at the iteration $q'$, 
    we bound the additive error of the final estimate $\tilde{u}_j$ in such branches as follows:
    \begin{align}
        |\tilde{u}_j-\braket{O_j}|&\leq|\tilde{u}_j-\tilde{u}^{(q')}_j|+|\tilde{u}^{(q')}_j-\braket{O_j}|\notag\\
        &\leq \pi\left|\sum_{q\geq q'}\frac{g_j^{(q)}}{2^{q}}\right|+\frac{1}{2^{q'}}\notag\\
        &\leq \frac{1+\pi}{2^{q'}}.
    \end{align}
    In the third line, we use the fact $|g_j^{(q)}|\leq 1/2$.
    Thus, we can calculate the mean squared error of $\hat{u}_j$ as
    \begin{align}
        &\mathbb{E}\left[\left(\hat{u}_j-\braket{O_j}\right)^2\right] \notag\\
        &\leq \frac{1}{2^{2(q_{\rm max}+1)}}+ (1+\pi)^2 \sum_{q=0}^{q_{\rm max}} \frac{\delta^{(q)}}{4^q}\notag\\
        &\leq \frac{1}{2^{2(q_{\rm max}+1)}}+ (1+\pi)^2 c 2^{-2q_{\rm max}+1}\notag\\
        &\leq \varepsilon^2,
    \end{align}
    where we defined $q_{\rm max}:=\lceil \log_2 (1/\varepsilon)\rceil$.
    In the final line, we used the fact $c\in (0,3/(8(1+\pi)^2)]$.
    The inequality holds for all $j$, which completes the proof of Eq.~(\ref{eq:thm2_maxmse}).

    Next, we count the total queries to the state preparation $U_{\psi}$.
    At each iteration $q$, we prepare $\#:=9\ln(M/\delta^{(q)})$ copies of a quantum state that approximates the probing state $\ket{\Upsilon(q)}$.
    If we prepare a single copy of this state using the method in Lemma~\ref{supple_lem:statepre_HS}, then the $\mathcal{O}(2^{p+q+2}\sigma)$ queries to $U_{\psi}$ and its inverse are required, and therefore the total queries are calculated as
    \begin{align}
        &\sum_{q=0}^{q_{\rm max}} 2^{p+q+2}\sigma \times 9\ln(M/\delta^{(q)}) \notag\\
        &=9\cdot 2^{p+2}\sigma \times \sum_{q=0}^{q_{\rm max}} 2^{q}\ln(M/\delta^{(q)}) \notag\\
        &=9\cdot 2^{p+2}\sigma \left[2^{q_{\rm max}+1}\ln \frac{8M}{c}-(q_{\rm max}+2)\ln 8 - \ln\frac{M}{c}\right]\notag\\
        &=\mathcal{O}(\varepsilon^{-1}\sqrt{M}\log M).\notag
    \end{align}
\end{proof}

Here, we discuss the case that we employ Lemma~\ref{supple_lem:statepre_Grover} to the state preparation in Step 4 of Algorithm~\ref{supp_alg:main}.
This alternative method can prepare the probing state $\ket{\Upsilon(q)}$ under the following conditions: (i) $|\braket{\tilde{O}^{(q)}_j}|\leq 2^{-q-1}$ and (ii) the iteration step $q$ satisfies Eq.~(\ref{supple_eq:grver_sp_condition}).
In the case that only the condition (i) (or equivalently, Eq.~(\ref{supple_eq:initcondition4gradest})) is violated, the quantum circuit in Fig.~\ref{fig:sp_Grover_diagram} for Lemma~\ref{supple_lem:statepre_Grover} yields a $pM$-qubit quantum state that may be far from the target probing state.
However, this does not affect the proof of Theorem~\ref{thm:main_query_complexity} because it is only required that the measurement results at the end of Fig.~\ref{fig:sp_Grover_diagram} are in the range of $[-1/2,1/2]$.
On the other hand, the condition (ii) restricts the usage of the alternative method, as discussed in Sec.~\ref{sec:numerical_exp}.

\subsubsection{The proof of Theorem~\ref{thm:main_thm_improved}}\label{sec:proof_improved_thm}

In this Appendix, we prove Theorem~\ref{thm:main_thm_improved} by explicitly constructing a quantum algorithm.
The quantum algorithm is the same as Algorithm~\ref{supp_alg:main} except for the Step 4. 
For Theorem~\ref{thm:main_thm_improved}, we use an improved version of Lemma~\ref{supple_lem:statepre_HS} or Lemma~\ref{supple_lem:statepre_Grover}, with a \textit{known} parameter 
\begin{equation}
    \mathcal{B}_M\in \left[\left\|\sum_{j=1}^M O_j^2\right\|,M\right]
\end{equation}
for the state preparation in Step 4.
At the level of quantum circuit, we modify the following in Fig.~\ref{fig:spbyHS_fullcircuit} for the improved Lemma~\ref{supple_lem:statepre_HS}: 
\begin{itemize}
    \item Replace $U_{\rm obs}^{(\bm{x})}$ with a quantum circuit $V_{\rm obs}^{(\bm{x})}$ for $(\sqrt{\delta'}/2^{p+q+1}\overline{\sigma})$-precise block encoding of the Hamiltonian
    $(\overline{\sigma})^{-1}\sum_{j=1}^M x_j O_j$, where $\overline{\sigma}$ is defined as
    \begin{equation}
    \overline{\sigma}:=\sqrt{2\mathcal{B}_M\ln(2d/\delta')}=\mathcal{O}(\sqrt{\mathcal{B}_M\log d}).
    \end{equation}
    \item Set the evolution time $t=2^{p+q+1}\overline{\sigma}$
    \item Apply the following phase oracle for given $\tilde{u}^{(q)}_{j}\in [-1,1]$ ($j=1,2,...,M$) at the end of circuit
    \begin{align}\label{eq:phase_oracle_for_final_phase_adjust}
        &\ket{\bm{x}}\mapsto e^{-2\pi i 2^p \sum_{j=1}^M x_j2^{q}\pi^{-1} \tilde{u}^{(q)}_j}\ket{\bm{x}}\notag\\
        &=\bigotimes_{j=1}^M \left(e^{-i (2^{p+q+1}\tilde{u}^{(q)}_j)x_j}\ket{x_j}\right),
    \end{align}
    which can be easily constructed due to the separable structure.
\end{itemize}
The total number of qubits is given by Eq.~\eqref{eq:Lemma13_total_ancilla_reply}, where the additive constant 9 can be reduced to 7 because the block encoding of $\tilde{O}_j^{(q)}$ is unnecessary.

If the iteration step $q$ in Algorithm~\ref{supp_alg:main} exceeds a new threshold specified below, we can also use an improved version of Lemma~\ref{supple_lem:statepre_Grover}.
In this case, the modification for the circuit in Fig.~\ref{fig:sp_Grover_diagram} is as follows.
\begin{itemize}
    \item Replace $U_{\rm obs}^{(\bm{x},y)}$ with a quantum circuit $V_{\rm obs}^{(\bm{x},y)}$ for $(\sqrt{\delta'}/2^{p+q+2}\overline{\sigma}')$-precise block encoding of the Hamiltonian $(\overline{\sigma}')^{-1}\sum_{j=1}^{M+1} x_j \tilde{O}^{(q)}_j,$
    where $\tilde{O}_{M+1}^{(q)}:=O_{M+1}$ and $\overline{\sigma}'$ is defined as
    \begin{align}
        \overline{\sigma}'&:=\left\lceil\sqrt{2\left(\mathcal{B}_M+4^{-q-1}\right)\ln(2d/\delta')}\right.\notag\\
        &~~~~~+\left.\frac{1}{2^{q+1}}\sqrt{2(M+1){\ln(2/\delta')}}\right\rceil\notag\\
        &=\mathcal{O}(\sqrt{\mathcal{B}_M\log d}),
    \end{align}
    where the final evaluation holds because of the threshold Eq.~\eqref{eq:improved_G_SP_condition1}.
    \item Set $t=2^{p+q+2}\overline{\sigma}'$.
\end{itemize}
Here, the total number of qubits is given by Eq.~\eqref{eq:Lemma14_total_ancilla_reply}.
This state preparation is available when the iteration step $q$ satisfies (we recall $p=3$ and $\delta'=2^{-14}$)
\begin{equation}\label{eq:improved_G_SP_condition2}
    \frac{t}{5}\left(\frac{1}{2^{q+2}\overline{\sigma}'}\left(\frac{\sqrt{\delta'}}{2^{p}}+\sqrt{2(M+1){\ln(2/\delta')}}\right)\right)^3\leq \sqrt{\delta'}
\end{equation}
for ensuring the approximation error in the final state and
\begin{equation}\label{eq:improved_G_SP_condition1}
    2^{q+1}\sqrt{2\left(\mathcal{B}_M+4^{-q-1}\right)\ln(2d/\delta')}>\sqrt{2(M+1){\ln(2/\delta')}}
\end{equation}
for $\overline{\sigma}'=\mathcal{O}(\sqrt{\mathcal{B}_M\log d})$.

Because the final state in Figs.~\ref{supple_lem:statepre_HS} or \ref{fig:sp_Grover_diagram} after these modifications is 1/12-close to the target state Eq.~\eqref{eq:targetstate} in the Euclidean distance as proved below, we complete the Step~4 by measuring the final $3M$-qubit state by Fourier basis (see Appendix~\ref{sec:list_circuits}).
Thus, the same analysis on MSE in the proof of Theorem~\ref{thm:main_query_complexity} holds, and as for the query complexity, we simply replace $\sigma$ with $\overline{\sigma}$ or $\overline{\sigma}'$.
Consequently, the total number of queries to $U_{\psi}$ and $U_{\psi}^\dagger$ is given by $\mathcal{O}(\varepsilon^{-1}\sqrt{\mathcal{B}_M\log d}\log M)$.

In the following, we show the correctness of these modifications.
First of all, Lemma~\ref{supple_lem:BE_ampH} holds when we replace $\sigma$ with $\overline{\sigma}$ from the discussion after the proof of Lemma~\ref{supple_lem:BE_ampH}.
The same replacement is available in Lemma~\ref{supple_lem:statepre_HS}, so we can prepare the $pM$-qubit quantum state
\begin{equation}\label{eq:apdx_lem13_improved_prestate}
    \frac{1}{\sqrt{2^{pM}}}\sum_{\bm{x}\in G_p^M} e^{2\pi i2^p \sum_{j=1}^M x_j {2^{q}\pi^{-1}\braket{{O}_j}}} \ket{\boldsymbol{x}}
\end{equation}
up to 1/12 Euclidean distance error, along with the parameter $Q=\mathcal{O}(2^{p+q+1}\overline{\sigma})$ in the statement of Lemma~\ref{supple_lem:statepre_HS}.
Thus, applying the phase oracle Eq.~\eqref{eq:phase_oracle_for_final_phase_adjust}, we obtain the target state Eq.~\eqref{eq:targetstate}.
Note that the Euclidean distance error remains unchanged under the unitary transformation.

Under the same assumption in Lemma~\ref{supple_lem:statepre_Grover}, 
a subset $F'\subset G_{p}^M\times G_1$ defined as Eq.~\eqref{supple_eq:dfn_Fp_for_obs} with $\tilde{O}^{(q)}_j$ instead of $O_j$ satisfies $|F'|\geq (1-\delta')|G_{p}^M\times G_1|$.
Also, we have another subset $F\subset G_{p}^M\times G_1$ with $|F|\geq (1-\delta')|G_p^M \times G_1|$, and in this subset, we obtain
\begin{equation}
    \left\|\sum_{j=1}^{M+1} x_j O_j\right\|\leq \frac{1}{2}{\sqrt{2\left(\mathcal{B}_M+4^{-q-1}\right)\ln(2d/\delta')}}.
\end{equation}
This follows from the Lemma~\ref{supple_lem:BE_ampH} and the discussion after this lemma.
By combining the above observations, we can say that if $(\bm{x},y)\in F\cap F'$, then 
\begin{align}
    &\left\|\sum_{j=1}^{M}x_j\tilde{u}^{(q)}_j\bm{1}-yO_{M+1}\right\|\notag\\
    &\leq    \left|\sum_{j=1}^{M}x_j\tilde{u}^{(q)}_j-y\frac{\pi(1/4+4l)}{2^{p+q}}\right|\notag\\
    &=\left|\sum_{j=1}^{M+1}x_j\langle O_j\rangle-2\sum_{j=1}^{M+1}x_j\langle \tilde{O}^{(q)}_j\rangle\right| \notag\\
    &\leq \left|\sum_{j=1}^{M+1}x_j\langle O_j\rangle\right|+\sqrt{{{2{\ln(2/\delta')}\sum_{j=1}^{M+1}\langle \tilde{O}^{(q)}_j\rangle^2}}}\notag\\
    &\leq \left\|\sum_{j=1}^{M+1}x_j O_j\right\|+\frac{1}{2}\frac{1}{2^{q}}\sqrt{2(M+1){\ln(2/\delta')}}
\end{align}
and thus,
\begin{align}
    &\left\|\sum_{j=1}^{M+1} x_j \tilde{O}^{(q)}_j\right\|\notag\\
    &\leq \frac{1}{2}\left\|\sum_{j=1}^{M+1} x_j {O}_j\right\|+\frac{1}{2}\left\|\sum_{j=1}^{M} x_j \tilde{u}^{(q)}_j\bm{1}-yO_{M+1}\right\|\notag\\
    &\leq \left\|\sum_{j=1}^{M+1}x_j O_j\right\|+\frac{1}{2}\frac{1}{2^{q+1}}\sqrt{2(M+1){\ln(2/\delta')}}\notag\\
    &\leq \frac{\overline{\sigma}'}{2}.
\end{align}
Therefore, by using Lemma~\ref{supple_lem:uniform_amp} for a quantum circuit for the block encoding of $\propto \sum_{j} x_j\tilde{O}^{(q)}_j$, we obtain
\begin{equation}
    V_{\rm obs}=\sum_{(\bm{x},y)\in  G_{p}^M\times G_1} \ket{\bm{x},y}\bra{\bm{x},y}\otimes V^{(\bm{x},y)}_{\rm obs},
\end{equation}
where $V^{(\bm{x},y)}_{\rm obs}$ is $(\sqrt{\delta'}/2^{p+q+2}\overline{\sigma}')$-precise block encoding of
\begin{equation}
    \frac{1}{\overline{\sigma}'}\left(y\tilde{O}^{(q)}_{M+1}+\sum_{j=1}^{M} x_j \tilde{O}^{(q)}_j\right)\equiv \frac{1}{\overline{\sigma}'}\sum_{j=1}^{M+1} x_j \tilde{O}^{(q)}_j,
\end{equation}
if $(\bm{x},y)\in F\cap F'$.
Then, replacing the $U_{\rm obs}$ and $t=2^{p+q+2}\sigma$ with $V_{\rm obs}$ and $t=2^{p+q+2}\overline{\sigma}'$, respectively, in the circuit Fig.~\ref{fig:sp_Grover_diagram}, we can prepare the target state $\ket{\Upsilon(q)}\ket{+}$ \textit{approximately}
after we obtain all $0$ outcomes in the other registers of the circuit Fig.~\ref{fig:sp_Grover_diagram}.

To ensure that the approximation error is at most 1/12 in the Euclidean distance, the non-negative integer $q$ needs to satisfy Eq.~\eqref{eq:improved_G_SP_condition2}. 
We obtain this condition in a similar way to prove Eqs.~\eqref{eq:B33} and \eqref{eq:B34}.
Also, this condition is roughly evaluated as
\begin{equation}\label{eq:lem14_modified_condition_1}
    2^{q+1}\sqrt{2\left(\mathcal{B}_M+4^{-q-1}\right)\ln(2d/\delta')} \geq \mathcal{O}\left(M^{\frac{3}{4}}\right)
\end{equation}
which is asymptotically tighter than the other condition Eq.~\eqref{eq:improved_G_SP_condition1}.

\subsection{Circuit implementation}\label{sec:list_circuits}
We first show quantum circuits for Lemma~\ref{supple_lem:BE_ampH}.
For clarification, input oracle unitary gates are provided in Fig.~\ref{fig:input_oracles}.
Final output circuit of Lemma~\ref{supple_lem:BE_ampH} is given by Fig.~\ref{fig:amplified_BE_circuit}, which corresponds to $U_{\rm obs}$ in Eq.~\eqref{supple_eq:amplified_BE_obs}. 
In Figs.~\ref{fig:BEofnormalizedUsel}--\ref{fig:BEofsinH}, we describe an efficient circuit construction of (a slightly modified version of) $U_{\rm SEL}'$ with a single use of $U_{\rm SEL}$; see Remark~\ref{rem:design_prepare}.

We give quantum circuits for preparing $\ket{\Upsilon(q)}$ via Hamiltonian simulation protocol Lemma~\ref{supple_lem:statepre_HS}.
The full circuit is shown in Fig.~\ref{fig:spbyHS_fullcircuit}.
In particular, the gate $U_{\rm obs}^{(\bm{x})}$ and its inverse contained in this figure are decomposed as in Fig.~\ref{fig:amplified_BE_circuit}, where $\mathbf{p}=(3,...,3)$.
When we use the circuits in Figs.~\ref{fig:BEofnormalizedUsel}--\ref{fig:BEofsinH} for $U_{\rm SEL}'$, the total number of qubits in the full circuit Fig.~\ref{fig:spbyHS_fullcircuit} is given by
\begin{equation}\label{eq:Lemma13_total_ancilla_reply}
    3M+\lceil\log_2 M\rceil+\log_2 d+a+9.
\end{equation}
Note that each $B_j$ in Fig.~\ref{fig:BEofnormalizedUsel} is replaced by the circuit in Fig.~\ref{supplefig:LCU_O_u}.

Then, we show quantum circuits for preparing $\ket{\Upsilon(q)}$ via Grover-like repetition Lemma~\ref{supple_lem:statepre_Grover}.
In this case, the full quantum circuit is shown in Fig.~\ref{fig:sp_Grover_diagram}, which is simpler than that of Hamiltonian-simulation based protocol.
When we construct $U_{\rm obs}^{(\bm{x},y)}$ by the circuits in Figs.~\ref{fig:BEofnormalizedUsel}--\ref{supplefig:LCU_O_u}, the total number of qubits in the full circuit Fig.~\ref{fig:sp_Grover_diagram} is calculated as
\begin{equation}\label{eq:Lemma14_total_ancilla_reply}
    3M+\lceil\log_2 (M+1)\rceil+\log_2 d+a+8.
\end{equation}

Finally, we make a remark on Step~3 in Algorithm~\ref{alg:main} (or Step~4 in Algorithm~\ref{supp_alg:main}).
To complete this step, we simply measure only the $3M$ qubits of Fig.~\ref{fig:spbyHS_fullcircuit} in the computational basis after inverse QFTs.
This step does not require any post-selection because the final output state without post-selection is ensured to be $1/12$-close to the separable state $\ket{\bm{0}}\ket{\Upsilon(q)}$.
Similarly, we can remove the post-selection in Lemma~\ref{supple_lem:statepre_Grover} by using a single-step amplitude amplification and adjusting $\delta'$ to ensure the closeness condition on a final output state.

\section{Details on numerical simulation}\label{apdx:num_sim_details}


To solve the problem in Sec.~\ref{sec:num_simulation}, our adaptive Algorithm~\ref{alg:main} with Lemma~\ref{thm:sp_HS} uses
\begin{equation}\label{eq:upperbound_our_query_alg1}
    T_{\rm adapt}:=\sum_{q=0}^{q_{\rm max}} 2Q(q)\times 9\ln\left(\frac{M}{\delta^{(q)}}\right)=\mathcal{O}(\varepsilon^{-1}\sqrt{M}\log M)
\end{equation}
queries of $U_{\psi}$ and its inverse in total, where $q_{\rm max}:=\lceil \log_2(1/\varepsilon)\rceil$ for a target root MSE $\varepsilon$; see the proof of Theorem~\ref{thm:main_query_complexity} in Appendix~\ref{supple_sec;complexity_of_alg}.
Here, $Q(q)$ is the minimal integer $Q$ satisfying $4t^{Q}/(2^Q Q!)\leq 2^{-17}$ for the total evolving time $t=2^{5+q}\sqrt{2M\ln(2^{11} d)}$, and we use an explicit upper bound $Q(q)< 1.5t+126$ for the numerical simulation.
It should be noted that $T_{\rm adapt}$ itself is an upper bound of total queries in Algorithm~\ref{alg:main} to theoretically ensure root MSE $\leq \varepsilon$ for every target observables, and thus the required queries might be smaller than the upper bound in practice.
In addition, from the entire circuit diagrams in Appendix~\ref{sec:list_circuits}, the total number of qubits required for quantum circuits in our algorithm is calculated as Eqs.~\eqref{eq:Lemma13_total_ancilla_reply} and \eqref{eq:Lemma14_total_ancilla_reply}.

For comparison, we briefly review the previous method in Ref.~\cite{PhysRevLett.129.240501}.
The method essentially uses the gradient estimation method~\cite{gilyen2019optimizing} that can be summarized as the following theorem.
\begin{thm}[Theorem~25 in Ref.~\cite{gilyen2019optimizing}]\label{thm:gilyen_gradest_thm25}
    Let $\bm{x}\in \mathbb{R}^M$, $\varepsilon_{\rm add}\leq c\in \mathbb{R}_{+}$ be fixed constants and suppose $f:\mathbb{R}^M\to \mathbb{R}$ is analytic and satisfies the following: for every $k\in \mathbb{N}$ and $\alpha\in [M]^k$, $k$-th order partial derivatives of $f$ at $\bm{x}$ (denoted by $\partial_{\alpha} f(\bm{x})$) are upper bounded as $|\partial_{\alpha} f(\bm{x})|\leq c^k 2^{k/2}$.
    Then, there is a quantum algorithm that works for all such functions, and outputs an $\varepsilon_{\rm add}$-approximate gradient $\tilde{\bm{g}}\in \mathbb{R}^M$ such that $\|\nabla f(\bm{x})-\tilde{\bm{g}}\|_{\infty}\leq \varepsilon_{\rm add}$ with probability at least $1-\delta$, using $\mathcal{\widetilde{O}}(c\sqrt{M}\log(M/\delta)/\varepsilon_{\rm add})$ queries to the phase oracle $O_f$.
\end{thm}
\noindent
Then, taking the function $f$ in this theorem as
\begin{equation}\label{eq:google_targ_fn_gradest}
    f(\bm{x}):=\frac{1}{2}-\frac{1}{2}{\rm Im}\left[\langle \psi|\prod_{j=1}^M e^{-2ix_j O_j}|\psi\rangle\right],
\end{equation}
we can estimate the expectation values of target observables because $\nabla f(\bm{0})$ is the collection of target expectation values.
Note that this function for observables with $\|O_j\|\leq 1$ satisfies the condition in Theorem~\ref{thm:gilyen_gradest_thm25} by taking $c=2$.

For circuit implementation, Ref.~\cite{PhysRevLett.129.240501} first provides an explicit quantum circuit for the \textit{probability} oracle to encode the function Eq.~\eqref{eq:google_targ_fn_gradest}.
Here, the probability oracle $U_f$ for a function $f$ is defined as
\begin{align}
    U_f:&\ket{\bm{x}}\ket{\bm{0}}\notag\\
    &\mapsto \ket{\bm{x}}\left(\sqrt{f(\bm{x})}\ket{1}\ket{\phi_1(\bm{x})}+\sqrt{1-f(\bm{x})}\ket{0}\ket{\phi_0(\bm{x})}\right)
\end{align}
for all $\bm{x}\in G_n^{M}$ ($n$ is a positive integer), where $\ket{\phi_0(\bm{x})}$ and $\ket{\phi_1(\bm{x})}$ are arbitrary (normalized) quantum states.
The quantum circuit for the probability oracle has a single use of the state preparation oracle $U_{\psi}$, controlled time evolutions for target observables, and additional 1 ancilla qubit except for the target system qubits.
Then, we obtain the phase oracle for the function Eq.~\eqref{eq:google_targ_fn_gradest} from the probability oracle.
\begin{thm}[Theorem~14 in Ref.~\cite{gilyen2019optimizing}]\label{thm:oracle_conversion_thm}
    Let $U_f$ be a probability oracle for a function $f(\bm{x})$. For any $\epsilon\in (0,1/3)$, we can implement an $\epsilon$-approximate phase oracle $\tilde{O}_f$ such that
    \begin{equation}
        \left\|\tilde{O}_f\ket{\Psi}\ket{\bm{0}}-{O}_f\ket{\Psi}\ket{\bm{0}}\right\|\leq \epsilon
    \end{equation}
    for all input states $\ket{\Psi}$, where $O_f$ denotes an exact phase oracle for $f$. This implementation uses $\mathcal{O}(\log(1/\epsilon))$ invocations of the probability oracle $U_f$ or its inverse and $\mathcal{O}(\log \log(1/\epsilon))$ additional ancilla qubits beyond those required by $U_f$.
\end{thm}
\noindent
In the numerical simulation, we assume the oracle conversion (from the probability oracle to a fractional phase oracle) is perfect $\epsilon=0$ and ignore the overhead in space complexity, for simplicity.
As for the query complexity, we only consider the overhead $10$ (for the $2$-round oblivious amplitude amplification), while the actual overhead is larger.

Given the phase oracle (and its fractional version) for the function Eq.~\eqref{eq:google_targ_fn_gradest}, the estimation method in Theorem~\ref{thm:gilyen_gradest_thm25} effectively constructs a phase oracle for the following function 
\begin{align}
    h(\bm{x})&:= \frac{1}{2^{\lceil \log_2(3cr)\rceil}}\sum_{l=-m}^m a_{l}^{(2m)}
    f(lr\bm{x})\notag\\
    &\simeq \sum_{j=1}^M \frac{r\braket{O_j}}{2^{\lceil \log_2(3cr)\rceil}}{x_j}+\cdots,
\end{align}
where $a_{l}^{(2m)}=0$ $(l=0)$ and for $l\neq 0$, 
\begin{equation}
    a_{l}^{(2m)}:=\frac{(-1)^{l-1}}{l}\frac{\binom{m}{|l|}}{\binom{m+|l|}{|l|}}.
\end{equation}
Here, $m=\lceil \log (c\sqrt{M}/\varepsilon_{\rm add})\rceil$, and $r$ is a rescaled parameter defined as
\begin{equation}\label{eq:rescale_param_r}
    r^{-1}:=9cm \sqrt{M}\left(81\cdot 8\cdot 42\pi m \frac{c\sqrt{M}}{\varepsilon_{\rm add}}\right)^{1/2m}.
\end{equation}
This function $h$ is more close to a linear function in the grid $G_{n}^M$ than the original function $f$, due to the higher-order central-difference formula.
Using the (effective) phase oracle $O_h$ on $G_n^M$, we can estimate the gradient of $h$ at $\bm{x}=0$, i.e., the rescaled expectation values ${r\braket{O_j}}/{2^{\lceil \log_2(3cr)\rceil}}$.
More precisely, measuring the following state by the Fourier basis, we obtain estimates $\bm{k}=(k_1,...,k_M)\in G_n^M$ for the gradient:
\begin{equation}\label{eq:gilyen_google_meas_state}
    \frac{1}{\sqrt{2^{nM}}}\sum_{\bm{x}\in G_n^M} e^{2\pi i 2^n h(\bm{x})}\ket{\bm{x}}.
\end{equation}
Then, we can estimate $\braket{O_j}$ by $(2^{\lceil \log_2(3cr)\rceil}/r)k_j$.
According to Ref.~\cite{gilyen2019optimizing} (see Theorem~21 of Ref.~\cite{gilyen2019optimizing}), we take the number of qubits $n$ for the 1 dimensional grid $G_n$ as
\begin{equation}
    n=\left\lceil \log_2({4}/{\varepsilon_{\rm add}r})\right\rceil+\lceil \log_2 (3cr)\rceil
\end{equation}
to successfully bound the estimation error.
Thus, the total number of qubits required for quantum circuits in this method is given by
\begin{equation}\label{eq:num_ancilla_google}
    \left(\left\lceil \log_2({4}/{\varepsilon_{\rm add}r})\right\rceil+\lceil \log_2 (3cr)\rceil\right)M+1+\log_2 d.
\end{equation}
To enhance the success probability to $1-\delta$, we take the coordinate-wise median of independent $N_{\rm med}=\mathcal{O}(\log(M/\delta))$ samples $\bm{k}^{(1)},\bm{k}^{(2)}...$ from the gradient estimation procedure as a final estimate.
As a result, the entire procedure for the observable estimation based on the Theorem~\ref{thm:gilyen_gradest_thm25} uses 
\begin{align}\label{eq:total_query_thm25_googlef}
    &T_{\rm NonIter}:=N_{\rm med}\times \sum_{l=-m}^m  \left\lceil2\pi\cdot 2^{\lceil \log_2(4/\varepsilon_{\rm add}r)\rceil} \cdot \left|a_{l}^{(2m)}\right| \right\rceil \notag\\
    &= N_{\rm med}\times\mathcal{O}\left(\frac{c\sqrt{M}}{\varepsilon_{\rm add}}\log \left(\frac{c\sqrt{M}}{\varepsilon_{\rm add}}\right)\log\log \left(\frac{c\sqrt{M}}{\varepsilon_{\rm add}}\right)\right)
\end{align}
queries to the state preparation $U_{\psi}$ and its inverse (up to the oracle conversion overhead in Theorem~\ref{thm:oracle_conversion_thm}); the derivation of this query count is provided in the proof of Theorem~25 in Ref.~\cite{gilyen2019optimizing}.

To evaluate MSE, we need full information regarding the output probability distribution for the Fourier-basis measurement on the state Eq.~\eqref{eq:gilyen_google_meas_state}. 
However, direct calculation of the distribution is difficult because of many qubits $nM=\mathcal{O}(M\log(1/\varepsilon_{\rm add}))$.
To overcome this, we perform the following numerical simulation.
Focusing on the top registers related to the estimation for the 1st gradient in Eq.~\eqref{eq:gilyen_google_meas_state}, the reduced state can be written as
\begin{equation}
    \sum_{(x_2,...,x_M)\in G_n^{M-1}} \frac{1}{2^{n(M-1)}}\ket{\Psi(x_2,...,x_M)}\bra{\Psi(x_2,...,x_M)},
\end{equation}
where 
\begin{equation}
    \ket{\Psi(x_2,...,x_M)} = \frac{1}{\sqrt{2^n}} \sum_{x_1\in G_n} e^{2\pi i 2^n {h}(x_1,x_2,...,x_M)} \ket{x_1}.
\end{equation}
Therefore, we can estimate the full probability distribution $p(k_1)$ for the Fourier-basis measurement on the $x_1$ registers by
\begin{equation}
    p(k_1) \simeq \frac{1}{\#} \sum_{i=1}^{\#} \left|\bra{k_1}{\rm QFT}_{G_n}^\dagger \ket{\Psi(x^{(i)}_2,...,x^{(i)}_M)}\right|^2,
\end{equation}
where $i=1,2,...,\#$ denotes the index of trials and $(x^{(i)}_2,...,x^{(i)}_M)$ is independently and uniform-randomly sampled over $G_{n}^{M-1}$.
Here, we take $\#=10000$ for the numerical simulation in this work.

Given the probability distribution of a single shot ${p(k_1)}$, we can calculate the probability distribution of median $k^{{\rm med}}$ calculated by (an odd number) $N_{\rm med}$ independent samples of $k_1$, from the theory of order statistics~\cite{shao2008mathematical}.
The explicit form of the median distribution is given by 
\begin{widetext}
\begin{equation}
    {\rm Pr}(k^{{\rm med}}=x_i)=\sum_{l=(N_{\rm med}+1)/2}^{N_{\rm med}} \binom{N_{\rm med}}{l} \left\{P_i^l(1-P_i)^{N_{\rm med}-l}-P_{i-1}^l(1-P_{i-1})^{N_{\rm med}-l}\right\},
\end{equation}
where $G_n=\{x_i\}_{i=1}^{2^n}$ for $x_i=(2i-1)/2^{n+1}-1/2$, and $P_i$ is the cumulative distribution $P_i:={\rm Pr}(k_1\leq x_i)$, $P_0=0$.
Using the median distribution, we can calculate MSE of estimates in the observable estimation method. 
The corresponding number of total queries is given by Eq.~\eqref{eq:total_query_thm25_googlef}.

To confirm the tightness of $N_{\rm med}=\mathcal{O}(\log(1/\varepsilon))$ ($M$ is fixed) for achieving MSE $\varepsilon^2$, we have focused on a modified factor $\tilde{T}_{\rm NonIter}$ in Fig.~\ref{fig:rmse_vs_T_graph_rescaled}.
In the following, we describe how the modification on $T_{\rm NonIter}$ was performed.
From Eqs.~\eqref{eq:rescale_param_r} and \eqref{eq:total_query_thm25_googlef}, we can obtain 
\begin{align}
    T_{\rm NonIter}&= N_{\rm med}\times \sum_{l=-m}^m  \left\lceil2\pi\cdot 2^{\lceil \log_2(4/\varepsilon_{\rm add}r)\rceil} \cdot \left|a_{l}^{(2m)}\right| \right\rceil \notag\\
    &\sim N_{\rm med} \times \frac{c m\sqrt{M}}{\varepsilon_{\rm add}} \sum_{l=-m}^m \left|a_{l}^{(2m)}\right| \times \left(81\cdot 8\cdot 42\pi m \frac{c\sqrt{M}}{\varepsilon_{\rm add}}\right)^{1/2m}\\
    &\sim N_{\rm med} \times \frac{c\sqrt{M}}{\varepsilon_{\rm add}} m\log(m) \times \left(81\cdot 8\cdot 42\pi m \frac{c\sqrt{M}}{\varepsilon_{\rm add}}\right)^{1/2m}
\end{align}
where we keep only dominant factors in the second and third lines.
As $m$ increases, the factor $\left(81\cdot 8\cdot 42\pi m \frac{c\sqrt{M}}{\varepsilon_{\rm add}}\right)^{1/2m}$ decreases and converges to a constant, while the other terms increase.
Thus, to mitigate the effect of the contrary and asymptotically negligible factor, we define $\tilde{T}_{\rm NonIter}$ as
\begin{equation}\label{eq:def_T_noniter_rescaled}
    \tilde{T}_{\rm NonIter}:=\frac{1}{\left(81\cdot 8\right)^{1/2m}}T_{\rm NonIter}=N_{\rm med}\times\mathcal{O}\left(\frac{c\sqrt{M}}{\varepsilon_{\rm add}}\log \left(\frac{c\sqrt{M}}{\varepsilon_{\rm add}}\right)\log\log \left(\frac{c\sqrt{M}}{\varepsilon_{\rm add}}\right)\right).
\end{equation}
\end{widetext}
In Fig.~\ref{fig:rmse_vs_T_graph_norescaled}, we show a similar plot as Fig.~\ref{fig:rmse_vs_T_graph_rescaled} using the total number of queries $T_{\rm NonIter}$ with no rescaling, instead of $\tilde{T}_{\rm NonIter}$.

While we have discussed the comparison with Ref.~\cite{PhysRevLett.129.240501}, we note that another previous method in Ref.~\cite{van2023quantum} follows a similar estimation procedure with measurements on a uniform superposition state evolved under a phase oracle.
In particular, it also uses the median trick essentially to obtain final estimates for expectation values.
Similar to the numerical simulation suggesting $N_{\rm med}={\Omega}(\log(1/\varepsilon))$ in order to achieve MSE $\leq \varepsilon^2$, such a logarithmic correction regarding $\varepsilon$ in the total queries might be \textit{necessary} in the method of Ref.~\cite{van2023quantum}.

\begin{figure*}[tb]
 \centering
 \begin{tabular}{ccc}
 \includegraphics[scale=0.5]{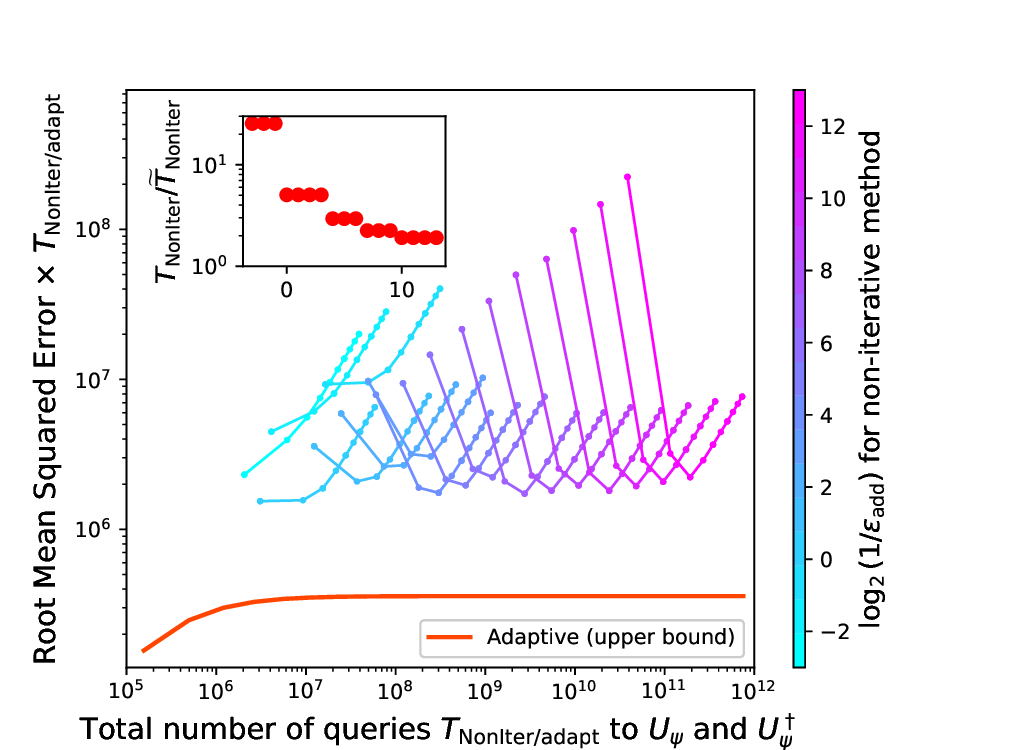}&&\includegraphics[scale=0.5]{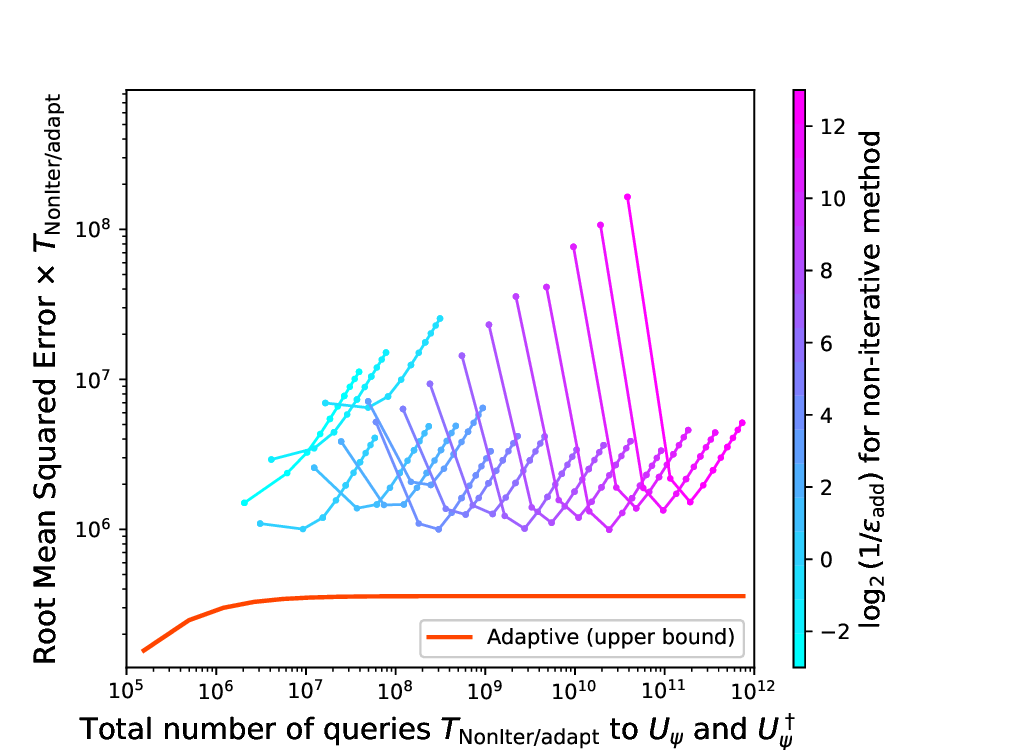}\\
 \\
 (a) Worst-case for the non-iterative method &&(b) Average-case for the non-iterative method\\
 \end{tabular}
 \caption{
 The comparison between the total number of queries (given by $T_{\rm adapt}$ and $T_{\rm NonIter}$ including the oracle-conversion overhead) and the root MSE.
 The colored dots show the (a) worst-case or (b) average-case results of the previous non-iterative method~\cite{PhysRevLett.129.240501} among randomly-generated 26 sets of $\{g_j\}$, with various input pairs $(\varepsilon_{\rm add}, \delta)$.
 For each color, we fix $\varepsilon_{\rm add}$ as in the color bar and vary the failure probability $\delta=2^{-j}$ $(j=0,1,...)$, which gives the number of samples $N_{\rm med}=2\lceil\log(1/\delta)\rceil+1$ for a final median calculation (see Eq.~\eqref{eq:total_query_thm25_googlef}).
 In the upper left subplot of (a), we show the value of rescaling in Eq.~\eqref{eq:def_T_noniter_rescaled}, where the horizontal axis is $\log_2(1/\varepsilon_{\rm add})$.
 The orange line for our adaptive method represents an upper bound.}
 \label{fig:rmse_vs_T_graph_norescaled}
\end{figure*}

\begin{figure*}[htbp]
    \centering
    \begin{tabular}{ccc}
         \Qcircuit @C=1em @R=1.5em {
         \\
         \ustick{\log_2 d}&{/}\qw &\gate{U_{\psi}}&\qw &&& \ustick{\log_2 d}&{/}\qw &\gate{U_{\psi}^\dagger}&\qw
         \\
         }
         ~~~~~&~~~~~
         \Qcircuit @C=1em @R=1.5em {
         \ustick{a}&{/}\qw &\qw &\multigate{2}{B_j}&\qw\\
         \\
         \ustick{\log_2 d}&{/}\qw &\qw&\ghost{B_j}&\qw
         }
         ~~~~~&~~~~~
         \Qcircuit @C=0.5em @R=1.6em {
        \ustick{a'}& {/} \qw&\qw&\qw &\multigate{2}{U_{\rm SEL}'}&\qw&\qw
        &&&&& {/} \qw&\qw&\qw &\multigate{1}{U^{(\bm{x})}}&\qw&\qw\\
        \ustick{\log_2 d}& {/} \qw&\qw&\qw &\ghost{U_{\rm SEL}'}&\qw&\qw
        &&~~~~~=~~~~~&&& {/} \qw&\qw&\qw &\ghost{U^{(\bm{x})}}&\qw&\qw\\
        \ustick{\|\mathbf{p}\|_1} & {/} \qw &\qw&\qw&\ghost{U_{\rm SEL}'}&\qw&\qw
        &&&&& {/} \qw &\dstick{~~~~~~~~~~~~~~~~~~~~\bm{x}\in G_{\mathbf{p}}} \qw &\qw&\ctrl{-1}&\qw&\qw
        }
        \\
        \\
    \end{tabular}
    \caption{Input oracles in this work.
    (Left) $U_{\psi}$ prepares a target state $\ket{\psi}$ from an initial state $\ket{\bm{0}}$. (Middle) For any $j\in\{1,2,...,M\}$, $B_j$ is an $a$-block-encoding of observable $O_j$ for a nonnegative integer $a$.
    (Right) $U^{(\bm{x})}$ is a block encoding of $\tilde{H}^{(\bm{x})}:=M^{-1}\sum_{j=1}^M x_jO_j$ with additional $a':=\mathcal{O}(\log M+a)$ qubits.
    For simplicity, we often assume $U_{\rm SEL}'$ to be an oracle for target observables $\{O_j\}$. Note that we can construct (a slightly modified) $U_{\rm SEL}'$ from $B_j$ with arbitrary precision; see Remark~\ref{rem:design_prepare} and Figs.~\ref{fig:BEofnormalizedUsel}--\ref{fig:BEofsinH}.
    }
    \label{fig:input_oracles}
\end{figure*}

\begin{figure*}[htbp]
  \centering
  \begin{minipage}[c]{0.4\textwidth}
    \centering
    \[
    \Qcircuit @C=0.5em @R=1.6em {
    \ustick{a'}& {/} \qw&\qw&\qw &\multigate{1}{\epsilon\mbox{-BE~of~}\frac{2M}{\pi 2^{\lceil\log_2 M\rceil}}\tilde{H}^{(\bm{x})}}&\qw&\qw\\
    \ustick{\log_2 d}& {/} \qw&\qw&\qw &\ghost{\epsilon\mbox{-BE~of~}\frac{2M}{\pi 2^{\lceil\log_2 M\rceil}}\tilde{H}^{(\bm{x})}}&\qw&\qw&~~~~~~~~~~~=\\
    \ustick{\|\mathbf{p}\|_1} & {/} \qw &\qw&\qw&\ctrl{-1}&\qw&\qw
    }
    \]
  \end{minipage}
  \begin{minipage}[c]{0.45\textwidth}
    \centering
    \[\Qcircuit @C=0.5em @R=1.8em {
  &&&&&\dstick{\mbox{(see caption of Fig.~\ref{fig:BEofsinH})}}\\
  \ustick{4} & {/} \qw & \qw& \qw& \qw &\multigate{1}{\epsilon\mbox{-BE~}\frac{2}{\pi}\mathcal{H}^{(\bm{x})}}&\qw&\qw&\qw&\qw&\qw&\qw&\qw&\qw&\qw&\qw&\qw&\qw&\qw\\
  \ustick{\lceil \log_2 M\rceil} & {/} \qw & \qw& \qw&\gate{H^{\otimes}}&\ghost{\epsilon\mbox{-BE~}\frac{2}{\pi}\mathcal{H}^{(\bm{x})}}&\ctrl{1}&\qw&\qw & \ustick{j\in \{1,...,M\}} \qw&\qw&\qw&\qw&\qw&\qw&\qw&\qw&\gate{H^{\otimes}}&\qw\\ 
  \ustick{a} & {/} \qw & \qw& \qw&\qw &\qw &\multigate{1}{B_j}&\qw&\qw&\qw&\qw&\qw&\qw&\qw&\qw&\qw&\qw&\qw&\qw\\ 
  \ustick{\log_2 d} & {/} \qw & \qw& \qw&\qw &\qw &\ghost{B_j}&\qw&\qw&\qw&\qw&\qw&\qw&\qw&\qw&\qw&\qw&\qw&\qw\\ 
  \ustick{\|\mathbf{p}\|_1} & {/} \qw & \qw & \qw &\dstick{~~~~~~~~~~~~~~~~~~~~~~~~~~~~~~~~~~\bm{x}\in G_{\mathbf{p}}}\qw&\ctrl{-3}&\qw&\qw&\qw&\qw&\qw&\qw&\qw&\qw&\qw&\qw&\qw&\qw&\qw\\ 
  \\
  }
  \]
  \end{minipage}
  \caption{Quantum circuit for a slightly modified version of $U'_{\rm SEL}$, where $a':=4+\lceil\log_2 M\rceil+a$. 
  For any $\bm{x}\in G_{\mathbf{p}}$, the circuit (except for the bottom line) is a $(1,4+\lceil\log_2 M\rceil+a,\epsilon<1/2)$-block-encoding of $(2M/\pi 2^{\lceil\log_2 M\rceil})\tilde{H}^{(\bm{x})}$.
  $H^{\otimes}$ means $H^{\otimes\lceil \log_2 M\rceil}$.}
  \label{fig:BEofnormalizedUsel}
\end{figure*}

\begin{figure*}[tb]
\centering
\begin{tabular}{c}
\\
\\
~~~~~\Qcircuit @C=0.8em @R=1em {
  \ustick{~~~~~~~p}&\ustick{~~~~~~~~~~~~~~~~~~~~~~~~~~~~~~~~~~~~~~~~~~~~~~~~~~~~~~~~~~\nu \in \{0,1,...,2^p-1\}}& {/} \qw& \qw & \ctrl{1} & \qw & \qw \\
  &\lstick{\ket{0}_{\rm w}}  & \qw& \targ & \ctrl{1}  & \targ & \rstick{\ket{0}_{\rm w}} \qw\\
  &  & \qw& \qw & \gate{e^{-iZ\varphi(\nu)}} & \qw & \qw\\
  \\
  \ustick{~~\lceil\log_2 M\rceil}&& {/} \qw & \ctrlo{-3} & \qw & \ctrlo{-3} & \qw\\
  }
\\
\\
\\
(a) Simplified diagram
\\
\\
\Qcircuit @C=0.8em @R=1.1em {
	\lstick{\ket{\nu^{(1)}}} & \qw &\qw & \ctrl{4} & \qw     &\qw&\cdots&&\qw&\qw&\qw \\
    \lstick{\ket{\nu^{(2)}}} & \qw &\qw & \qw      & \ctrl{3}&\qw&\cdots&&\qw&\qw&\qw\\
    &\vdots \\
    \lstick{\ket{\nu^{(p)}}} & \qw &\qw & \qw & \qw&\qw&\cdots&&\ctrl{1}& \qw &\qw\\
    \lstick{\ket{0}_{\rm w}}         &\targ&\ctrl{1}& \ctrl{1} & \ctrl{1}&\qw&\cdots&& \ctrl{1}&\targ &\rstick{\ket{0}_{\rm w}} \qw \\ 
    & \qw &\gate{R_z(\frac{1}{2^p}-1)} & \gate{R_z(\frac{2^1}{2^p})} & \gate{R_z(\frac{2^2}{2^p})} &\qw &\cdots&& \gate{R_z(\frac{2^p}{2^p})} &\qw  &\qw\\
    & \ctrlo{-2}                   &\qw & \qw& \qw &\qw &\cdots&& \qw& \ctrlo{-2} &\qw \\
    & \ctrlo{-1}                   &\qw & \qw& \qw &\qw &\cdots&& \qw& \ctrlo{-1}&\qw\\
    & \ctrlo{-1}                   &\qw & \qw& \qw &\qw &\cdots&& \qw& \ctrlo{-1} &\qw\\
    }
\\
\\
(b) Detailed diagram in the case of $\lceil \log_2 M\rceil=3$
\\
\end{tabular}
\caption{Quantum circuit for $\ket{0}^{\otimes \lceil \log_2 M\rceil}$-controlled unitary $\sum_{\nu}\ket{\nu}\bra{\nu}\otimes e^{-iZ\varphi(\nu)}$. The qubit $\ket{0}_{\rm w}$ is an additional work qubit. Here, we recall $\varphi(\mu):=\frac{\mu}{2^p}-\frac12 +\frac{1}{2^{p+1}}\in G_p$ in Eq.~(\ref{eq:varphi_bijection}). $R_z(x):=e^{-ixZ/2}$.}
\label{fig:block_enc_Hx}
\end{figure*}

\begin{figure*}[tb]
\centering
\begin{tabular}{c}
\\
\\
~~~~~\Qcircuit @C=0.5em @R=1.2em {
  \ustick{~~~~p_1}&\ustick{~~~~~~~~~~~~~~~~~~~~~~~~~~~~~~~~~~~~~~~~~~~\mu_1}& {/} \qw& \qw & \ctrl{4} & \qw & \qw & \qw &\qw  &\qw  &\qw&\cdots && \qw & \qw & \qw & \qw & \qw &\qw&\qw   \\
  \ustick{~~~~p_2}& & {/} \qw& \qw & \qw & \qw &\qw & \ustick{~~~~~~~~~~~~~~~~~~~~~~~~~\mu_2} \qw &\ctrl{3}  &\qw&\qw &\cdots && \qw & \qw & \qw & \qw & \qw &\qw   &\qw\\
  &&&\vdots \\
  \ustick{~~~~p_M}& & {/} \qw& \qw & \qw & \qw &\qw&\qw &\qw&\qw &\qw &\cdots && \qw & \qw & \ustick{~~~~~~~~~~~~~~~~~~~~~~~~~~~\mu_M} \qw & \ctrl{1} & \qw &\qw&\qw\\
  &\lstick{\ket{0}_{\rm w}}  & \qw& \targ & \ctrl{1}  & \targ &\qw & \targ & \ctrl{1} &\targ &\qw &\cdots&& \qw & \qw & \targ & \ctrl{1} & \targ &\qw& \rstick{\ket{0}_{\rm w}} \qw \\
  &\lstick{\ket{0}}  & \gate{X}& \qw & \gate{e^{-iZ\varphi(\mu_1)}} & \qw & \qw&\qw &\gate{e^{-iZ\varphi(\mu_2)}} &\qw & \qw &\cdots&& \qw & \qw & \qw & \gate{e^{-iZ\varphi(\mu_M)}} & \qw &\gate{X} & \qw \\
  \\
  \ustick{~~\lceil\log_2 M\rceil}&& {/} \qw & \ctrlo{-3} & \qw & \ctrlo{-3} & \gate{-1} &\ctrlo{-3} &\qw & \ctrlo{-3} &\qw &\cdots && \qw & \gate{-1} & \ctrlo{-3} & \qw & \ctrlo{-3} &\gate{-1} &\qw  \\
  }
\\
\\
\end{tabular}
\caption{Quantum circuit for a $2$-block-encoding of $\sum_{\bm{x}}\ket{\bm{x}}\bra{\bm{x}}\otimes e^{i\mathcal{H}^{(\bm{x})}}$, where $\mathcal{H}^{(\bm{x})}:=\sum_{j=1}^M x_j\ket{j}\bra{j}$.
The gate $-1$ denotes a bit decrement operation such as $\ket{M}\to \ket{M-1}\to\cdots \to \ket{1}\to \ket{M}$.}
\label{fig:block_enc_Hx2}
\end{figure*}

\begin{figure*}[tb]
\centering
\begin{tabular}{c}
\\
\\
~~~~~\Qcircuit @C=0.5em @R=1.2em {
  \ustick{~~~~p_1}&\ustick{~~~~~~~~~~~~~~~~~~~~~~~~~~~~~~~~~~~~~~~~~~~\mu_1}& {/} \qw&\qw& \qw & \ctrl{4} & \qw & \qw & \qw &\qw  &\qw  &\qw&\cdots && \qw & \qw & \qw & \qw & \qw &\qw&\qw   &\qw   \\
  \ustick{~~~~p_2}& & {/} \qw&\qw& \qw & \qw & \qw &\qw & \ustick{~~~~~~~~~~~~~~~~~~~~~~~~~\mu_2} \qw &\ctrl{3}  &\qw&\qw &\cdots && \qw & \qw & \qw & \qw & \qw &\qw   &\qw&\qw   \\
  &&&&\vdots \\
  \ustick{~~~~p_M}& & {/} \qw&\qw& \qw & \qw & \qw &\qw&\qw &\qw&\qw &\qw &\cdots && \qw & \qw & \ustick{~~~~~~~~~~~~~~~~~~~~~~~~~~~\mu_M} \qw & \ctrl{1} & \qw &\qw&\qw&\qw   \\
  & &\qw& \qw& \targ & \ctrl{1}  & \targ &\qw & \targ & \ctrl{1} &\targ &\qw &\cdots&& \qw & \qw & \targ & \ctrl{1} & \targ &\qw& \qw &\qw   \\
  & &\qw& \gate{X}& \qw & \gate{e^{-iZ\varphi(\mu_1)}} & \qw & \qw&\qw &\gate{e^{-iZ\varphi(\mu_2)}} &\qw & \qw &\cdots&& \qw & \qw & \qw & \gate{e^{-iZ\varphi(\mu_M)}} & \qw &\gate{X} & \qw &\qw   \\
  \\
  &&\gate{H}&\ctrl{-2}&\qw & \gate{iZ} &\qw & \qw & \qw & \qw &\qw &\qw  &\cdots && \qw & \qw & \qw  & \qw & \qw  &\ctrl{-2}& \gate{H}  &\qw   \\
  \dstick{\lceil\log_2 M\rceil~~~}&& {/} \qw &\qw& \ctrlo{-4} & \qw & \ctrlo{-4} & \gate{-1} &\ctrlo{-4} &\qw & \ctrlo{-4} &\qw &\cdots && \qw & \gate{-1} & \ctrlo{-4} & \qw & \ctrlo{-4} &\gate{-1} &\qw  &\qw   \\
  }
\\
\\
\end{tabular}
\caption{Quantum circuit for 3-block-encoding of $\sin\mathcal{H}^{(\bm{x})}$. 
Substituting this circuit and its inverse into $U$ and $U^\dagger$, respectively, in the circuit of Fig.~\ref{fig:qet_circuit} ($a$ is taken as 3) and setting $\mathcal{O}(\log(1/\epsilon))$ circuit parameters $\{\phi_i\}$ for $P(x)\approx(2/\pi)\arcsin(x)$, we obtain an $(1,4,\epsilon)$-block-encoding of $\sum_{\bm{x}}\ket{\bm{x}}\bra{\bm{x}}\otimes (2/\pi)\mathcal{H}^{(\bm{x})}$ for any $\epsilon\in (0,1/2)$.}
\label{fig:BEofsinH}
\end{figure*}


\begin{figure*}[htbp]
  \centering
  \begin{minipage}[c]{0.2\textwidth}
    \centering
    \[\Qcircuit @C=0.5em @R=1.6em {
    \ustick{a'+1}& {/} \qw&\qw&\qw &\multigate{1}{U_{\rm obs}^{(\bm{x})}}&\qw&\qw\\
    \ustick{\log_2 d}& {/} \qw&\qw&\qw &\ghost{U_{\rm obs}^{(\bm{x})}}&\qw&\qw&~~~~~~~=\\
    \ustick{\|\mathbf{p}\|_1} & {/} \qw &\qw&\qw&\ctrl{-1}&\qw&\qw
    \\
    }
    \]
  \end{minipage}
  \begin{minipage}[c]{0.7\textwidth}
    \centering
    \[\Qcircuit @C=0.5em @R=1.8em {
  &\gate{H}                           & \targ  & \gate{e^{-iZ\phi_m }} & \targ 
  &\qw & \targ & \gate{e^{-iZ\phi_{m-1} }} & \targ
  &\qw&\cdots&& \qw & \targ  & \gate{e^{-iZ\phi_1 }} & \targ
  &\gate{H} & \qw \\
  &\multigate{2}{U'_{\rm SEL}}  & \ctrlo{-1} & \qw & \ctrlo{-1}  
  &\multigate{2}{U'^{\dagger}_{\rm SEL}}& \ctrlo{-1}& {/} \qw & \ctrlo{-1} 
  &\dstick{a'~~~~~~~~~~~~~~~~~~~~~~~}\qw&\cdots& & \multigate{2}{U'_{\rm SEL}}& \ctrlo{-1}& \qw & \ctrlo{-1}
  & \qw&\qw \\
                              &\ghost{U'_{\rm SEL}}         & \ustick{~~~~~~~~\log_2 d} \qw & {/} \qw & \qw                             
  &\ghost{U'^{\dagger}_{\rm SEL}}& \qw & \qw & \qw
  &\qw &\cdots&& \ghost{U'_{\rm SEL}} & \qw & \qw & \qw
  & \qw&\qw \\
  &\ghost{U'_{\rm SEL}}& \ustick{~~~~~~~~~\|\mathbf{p}\|_1} \qw & {/} \qw & \qw & \ghost{U'^{\dagger}_{\rm SEL}} &\qw & \qw & \qw &\qw &\cdots&&\ghost{U'_{\rm SEL}} &\qw & \qw & \qw &\qw & \qw \\
  &\ustick{\mbox{(Fig.~\ref{fig:input_oracles}~or~Fig.~\ref{fig:BEofnormalizedUsel})}}
  \\
  }
    \]
  \end{minipage}
  \caption{Quantum circuit for $U_{\rm obs}$ in Eq.~(\ref{supple_eq:amplified_BE_obs}).
  Here, $a'=\mathcal{O}(\log M+a)$.
  The circuit parameters $\{\phi_j\}_{j=1}^m$ are adjusted for real QSP based on the degree-$m$ odd polynomial $P(x)\approx \gamma x,~\gamma=\mathcal{O}(M/\sigma)$.
  Also, $U_{\rm SEL}'$ can be replaced with the circuit in Fig.~\ref{fig:BEofnormalizedUsel} (where $a'$ is taken as $\lceil\log_2 M\rceil+a+4$), and in this case, the circuit equality holds by adjusting $\epsilon$ in Fig.~\ref{fig:BEofnormalizedUsel} and $m$.}
  \label{fig:amplified_BE_circuit}
\end{figure*}

\begin{figure*}[tb]
\centering
\begin{tabular}{c}
\\
\Qcircuit @C=0.8em @R=0.8em {
  \lstick{\ket{0}}& \gate{e^{-iY\phi/2}}   & \qw & \qw & \qw & \qw \\
  \lstick{\ket{0}}& \gate{e^{-iY\theta/2}} & \ctrlo{1} & \ctrl{1} & \gate{e^{iY\theta/2}} & \qw \\
  \lstick{\ket{0}^{\otimes a}} & {/} \qw & \multigate{1}{B_j} & \multigate{1}{-{\rm sgn}(\tilde{u}^{(q)}_j)\bm{1}} & \qw & \qw \\
  \dstick{~~~~~~~~~~~~~~~~~\log_2 d} & {/} \qw & \ghost{B_j} & \ghost{-{\rm sgn}(\tilde{u}^{(q)}_j)\bm{1}} & \qw & \qw \\
  }
\\
\\
\end{tabular}
\caption{Quantum circuit for an $(a+2)$-block-encoding of $\tilde{O}_j^{(q)}:=(O_j-\tilde{u}^{(q)}_j\bm{1})/2$, where $\tilde{u}^{(q)}_j\in [-1,1]$. $B_j$ is an $a$-block-encoding of $O_j$. 
The angles are defined as $\theta:=2\tan^{-1} \sqrt{|\tilde{u}^{(q)}_j|}$ and $\phi:=2\tan^{-1} \sqrt{4-(1+|\tilde{u}^{(q)}_j|)^2}/(1+|\tilde{u}^{(q)}_j|)$.}
\label{supplefig:LCU_O_u}
\end{figure*}

\newcommand{\arrep}[1]{\ar @<6pt> @/^/[#1]|-{\mbox{~Alternate~this~with~its~dagger~$2(Q-1)$~times~}}}
\begin{figure*}[htbp]
  \centering
  \begin{minipage}[c]{0.65\textwidth}
    \centering
    \[\Qcircuit @C=0.5em @R=1.6em {
    \lstick{\ket{0}}&\gate{H}&\qw &\qw&\gate{R^\dagger_z} & \gate{H} &\gate{S} & \ctrl{1}&\qw&\qw&\ctrl{1}&\gate{H}&\gate{R_z}&\qw&\gate{H}&\qw\\
    \lstick{\ket{0}}&\qw&\qw &\qw& \qw & \gate{H} &\ctrlo{1} &\targ &\ctrlo{1}&\gate{H}&\multigate{2}{\mbox{Ref}_{\ket{\bm{0}}}}&\qw&\qw&\qw&\qw&\qw\\
    \lstick{\ket{{0}}^{\otimes a''}}& \qw & {/} \qw&\qw&\qw&\qw &\multigate{1}{U_{\rm obs}^{(\bm{x})}}&\qw&\multigate{1}{U_{\rm obs}^{(\bm{x})\dagger}}&\qw&\ghost{\mbox{Ref}_{\ket{\bm{0}}}}&\qw&\qw&\qw&\qw&\qw&&&=\\
    \lstick{\ket{{0}}^{\otimes \log_2 d}}& \qw& {/} \qw&\qw&\qw&\gate{U_{\psi}} &\ghost{U_{\rm obs}^{(\bm{x})}}&\qw &\ghost{U_{\rm obs}^{(\bm{x})\dagger}}&\gate{U_{\psi}^\dagger}&\ghost{\mbox{Ref}_{\ket{\bm{0}}}}&\qw&\qw&\qw&\qw&\qw\\
    \lstick{\ket{+}^{\otimes 3M}}& \qw& {/} \qw &\qw&\qw&\qw&\ctrl{-1}&\qw&\ctrl{-1}&\qw&\qw&\qw&\qw&\qw&\qw&\qw\\
    & & & & & & & & & && & &\arrep{llllllllll}
    \\
    \\
    }
    \]
  \end{minipage}
  \begin{minipage}[c]{0.25\textwidth}
    \centering
    \[
    ~~~~~~~~~~~~~~~\Qcircuit @C=0.5em @R=1.6em {
    \lstick{\ket{0}^{\otimes 2}}&{/} \qw                      &\qw &\multigate{3}{\varepsilon''\mbox{-BE~of~}{e^{it\mathbf{H}}}}&\qw\\
    \lstick{\ket{0}^{\otimes a''+\log_2 d}}& {/} \qw&\qw&\ghost{\varepsilon''\mbox{-BE~of~}{e^{it\mathbf{H}}}}&\qw\\
    &&\\
    \lstick{\ket{+}^{\otimes 3M}}                    & {/} \qw&\qw&\ghost{\varepsilon''\mbox{-BE~of~}{e^{it\mathbf{H}}}}&\qw
    \\
    \\
    }
    \]
  \end{minipage}
  \caption{Quantum circuit for the approximate probing-state preparation $\ket{\Upsilon(q)}$ by Hamiltonian simulation.  
  When we use the circuits Figs.~\ref{fig:BEofnormalizedUsel} and~\ref{supplefig:LCU_O_u} for encoding of target observables $\{\tilde{O}^{(q)}_j\}$, $a''$ of $U^{(\bm{x})}_{\rm obs}$ is explicitly given by $a''=a+\lceil\log_2 M\rceil+7$, while Lemma~\ref{supple_lem:statepre_HS} itself assumes the oracles in Fig.~\ref{fig:input_oracles} (thereby, $a''=a+\lceil\log_2 M\rceil+1$ in Lemma~\ref{supple_lem:statepre_HS}).
  $Q=\mathcal{O}(2^{q}\sigma+\log (1/\varepsilon''))$ and Ref$_{\ket{\bm{0}}}$ is the reflection on $\ket{\bm{0}}$, i.e., Ref$_{\ket{\bm{0}}}:=2\ket{\bm{0}}\bra{\bm{0}}-\bm{1}$.
  The angles of rotations $R_z$ in the top line are adjusted by QSP to realize $\varepsilon''$-precise Hamiltonian simulation~\cite{low2017optimalHSbyQSP,Low2019hamiltonian}.}
  \label{fig:spbyHS_fullcircuit}
\end{figure*}

\renewcommand{\arrep}[1]{\ar @<6pt> @/^/[#1]|-{ \mbox{~Repeat~$t/2$~times~}}}
\begin{figure*}[bt]
\centering
\begin{tabular}{c}
  \Qcircuit @C=0.8em @R=1.4em {
  \lstick{\ket{+}^{\otimes 3M}} & {/} \qw 
  &\qw & \ctrl{1} &\qw &\qw 
  &\qw & \ctrl{1} &\qw &\qw 
  & \qw &\qw &\cdots&& \qw  & \gate{X^{\otimes 3M}} & \qw \\
  \lstick{\ket{+}} &\qw 
  &\qw & \ctrl{1} &\qw &\qw 
  &\qw & \ctrl{1} &\qw &\qw 
  & \qw &\qw&\cdots& &\gate{\mbox{QFT}_{G_1}^\dagger} &\ctrlo{-1} &\qw \\
  \lstick{\ket{{0}}^{\otimes a''}} &{/} \qw 
  &\qw & \multigate{1}{U_{\rm obs}^{(\bm{x},y)}} &\qw & \multigate{1}{\mbox{Ref}_{\ket{\bm{0}}}}
  &\qw & \multigate{1}{U_{\rm obs}^{(\bm{x},y)\dagger}} &\qw & \multigate{1}{\mbox{Ref}_{\ket{\bm{0}}}}
  & \qw &\qw &\cdots&& \meter & \rstick{\mbox{all}~0} \cw \\
  \lstick{\ket{{0}}^{\otimes \log_2 d}} &{/} \qw 
  & \gate{U_{\psi}} & \ghost{U_{\rm obs}^{(\bm{x},y)}} & \gate{U_{\psi^\dagger}} & \ghost{\mbox{Ref}_{\ket{\bm{0}}}}
  & \gate{U_{\psi}} & \ghost{U_{\rm obs}^{(\bm{x},y)\dagger}} & \gate{U_{\psi^\dagger}} & \ghost{\mbox{Ref}_{\ket{\bm{0}}}}
  &\qw&\qw&\cdots&& \meter & \rstick{\mbox{all}~0} \cw \\
  &&&&&&&&& \arrep{lllllll}
  }
  \\
\end{tabular}
\caption{Quantum circuit for the approximate probing-state preparation $\ket{\Upsilon(q)}$ by Grover-like repetition.
When we use the circuits Figs.~\ref{fig:BEofnormalizedUsel} and~\ref{supplefig:LCU_O_u} for encoding of target observables $\{\tilde{O}^{(q)}_j\}$, $a''$ of $U^{(\bm{x},y)}_{\rm obs}$ is explicitly given by $a''=a+\lceil\log_2 (M+1)\rceil+7$, while Lemma~\ref{supple_lem:statepre_Grover} itself assumes the oracles in Fig.~\ref{fig:input_oracles} (thereby, $a''=a+\lceil\log_2 (M+1)\rceil+1$ in Lemma~\ref{supple_lem:statepre_Grover}).
$\ket{\bm{0}}:=\ket{0}^{\otimes a''+\log_2 d}$ and Ref$_{\ket{\bm{0}}}$ is the reflection on $\ket{\bm{0}}$, i.e., Ref$_{\ket{\bm{0}}}:=2\ket{\bm{0}}\bra{\bm{0}}-\bm{1}$.
The measurement outcome $\ket{\bm{0}}$ is obtained with the probability $\mathcal{N}_t^2/2\geq 0.462$ (see Lemma~\ref{supple_lem:statepre_Grover}).
This post-selection can be removed by using a single-step amplitude amplification and adjusting the constant $\delta'$ for our algorithm.}
\label{fig:sp_Grover_diagram}
\end{figure*}

\end{document}